%% file: ms.tex
\documentclass[a4paper,UKenglish,cleveref, autoref]{lipics-v2019}
\usepackage{amsmath,tabu}
\usepackage{amssymb,amsfonts,textcomp}
\usepackage{amsthm}
\usepackage{stmaryrd}
\usepackage{booktabs}
\usepackage{longtable}
\usepackage{colortbl}
\usepackage{multirow}
\usepackage{tabularx}
\usepackage{pifont}
\usepackage[ruled, linesnumbered, vlined]{algorithm2e} % For algorithms

\usepackage{xcolor}

\xdefinecolor{LightCyan}{rgb}{0.80,1,1}
\definecolor{Gray}{gray}{0.75}

\newcolumntype{g}{>{\columncolor{Gray}}c}

\definecolor{dkgreen}{rgb}{0,0.6,0}
\definecolor{gray}{rgb}{0.5,0.5,0.5}
\definecolor{LightGray}{rgb}{0.93,0.93,0.93}
\definecolor{mauve}{rgb}{0.58,0,0.82}

\title{Complex Event Recognition with Symbolic Register Transducers: Extended Technical Report\footnote{ \textcolor{red}{This is the extended technical report for the paper \emph{Complex Event Recognition with Symbolic Register Transducers} to be presented at VLBD 2024. Please, use the VLDB version, once published, if you need to cite the paper.}}} %TODO Please add

\author{Elias Alevizos}{Institute of Informatics \& Telecommunications, National Center for Scientific Research ``Demokritos'', Greece}{alevizos.elias@iit.demokritos.gr}{https://orcid.org/0000-0002-9260-0024}{}
\author{Alexander Artikis}{Department of Maritime Studies, University of Piraeus, Greece \and Institute of Informatics \& Telecommunications, National Center for Scientific Research ``Demokritos'', Greece}{a.artikis@unipi.gr}{https://orcid.org/0000-0001-6899-4599}{}

\author{Georgios Paliouras}{Institute of Informatics \& Telecommunications, National Center for Scientific Research ``Demokritos'', Greece}{paliourg@iit.demokritos.gr}{https://orcid.org/0000-0001-9629-2367}{}

\authorrunning{Alevizos et al.}
\Copyright{Elias Alevizos, Alexander Artikis and Georgios Paliouras}%TODO mandatory, please use full first names. LIPIcs license is "CC-BY";  http://creativecommons.org/licenses/by/3.0/

\ccsdesc[300]{Theory of computation~Formal languages and automata theory}
\ccsdesc[300]{Theory of computation~Pattern matching}
\ccsdesc[500]{Information systems~Data streaming}

\keywords{Finite Automata, Regular Expressions, Complex Event Processing, Symbolic Automata}%TODO mandatory; please add comma-separated list of keywords

\category{}%optional, e.g. invited paper

\relatedversion{}%optional, e.g. full version hosted on arXiv, HAL, or other respository/website
\supplement{}%optional, e.g. related research data, source code, ... hosted on a repository like zenodo, figshare, GitHub, ...

\nolinenumbers %uncomment to disable line numbering

\hideLIPIcs  %uncomment to remove references to LIPIcs series (logo, DOI, ...), e.g. when preparing a pre-final version to be uploaded to arXiv or another public repository

\def\true{\textsf{\normalsize TRUE}}

\def\srem{$\mathit{SREM}$}

\def\wsremo{$\mathit{wSREMO}$}
\def\sra{$\mathit{SRA}$}
\def\srt{$\mathit{SRT}$}
\def\dsrt{$\mathit{dSRT}$}
\def\strictcont{\textsf{\normalsize strict-contiguity}}
\def\skipany{\textsf{\normalsize skip-till-any-match}}
\def\skipnext{\textsf{\normalsize skip-till-next-match}}

\def\rem{$\mathit{REM}$}
\def\sremo{$\mathit{SREMO}$}

\def\ssrt{$\mathit{sSRT}$}

\def\ssremo{$\mathit{sSREMO}$}

\def\dsra{$\mathit{dSRA}$}
\def\nsra{$\mathit{nSRA}$}

\newtheorem*{proposition*}{Proposition}
\newtheorem*{theorem*}{Theorem}
\newtheorem*{lemma*}{Lemma}
\newtheorem*{corollary*}{Corollary}

\begin{document}

\maketitle

\begin{abstract}
We present a system for Complex Event Recognition (CER) based on automata.
While multiple such systems have been described in the literature,
they typically suffer from a lack of clear and denotational semantics,
a limitation which often leads to confusion with respect to their expressive power. 
In order to address this issue, 
our system is based on an automaton model which is a combination of symbolic and register automata.
We extend previous work on these types of automata, 
in order to construct a formalism with clear semantics and a corresponding automaton model whose properties can be formally investigated. 
We call such automata Symbolic Register Transducers (\srt).
%\srt\ extend the expressive power of symbolic automata,
%by allowing Boolean formulas to be applied not only to the last element read from the input string,
%but to multiple elements, 
%stored in their registers. 
%\srt\ also extend register automata, 
%by allowing arbitrary Boolean formulas, 
%besides equality predicates.
The distinctive feature of \srt,
compared to previous automaton models used in CER,
is that they can encode patterns relating multiple input events from an event stream,
without sacrificing rigor and clarity.
We study the closure properties of \srt\ under union, intersection, concatenation, Kleene closure, complement and determinization by extending previous relevant results from the field of languages and automata theory.
We show that \srt\ are closed under various operators,
but are not in general closed under complement and they are not determinizable.
However,
they are closed under these operations when a window operator, 
quintessential in Complex Event Recognition, 
is used.
We show how \srt\ can be used in CER in order to detect patterns upon streams of events,
using our framework that provides declarative and compositional semantics, 
and that allows for a systematic treatment of such automata.
For \srt\ to work in pattern detection,
we allow them to mark events from the input stream as belonging to a complex event or not,
hence the name ``transducers''.
We also present an implementation of \srt\ which can perform CER.
We compare our \srt -based CER engine against other state-of-the-art CER systems and show that it is both more expressive and more efficient.
\end{abstract}

\input{intro}

\input{related}

\input{core}

\input{extensions}

\input{impl}

\input{exp}

\input{outro}

\appendix
\input{appendix}

%%
%% Bibliography
%%

%% Please use bibtex, 

\bibliography{refs}

\end{document}

%% file: intro.tex
\section{Introduction}

A Complex Event Recognition (CER) system takes as input a stream of ``simple events'',
along with a set of patterns,
defining relations among the input events,
and detects instances of pattern satisfaction,
thus producing an output stream of ``complex events'' \cite{DBLP:journals/vldb/GiatrakosAADG20,DBLP:books/daglib/0017658,DBLP:journals/csur/CugolaM12}.
Typically, an event has the structure of a tuple of values which might be numerical or categorical.
Time is of critical importance for CER and thus
a temporal formalism is used in order to define the patterns to be detected.
Such a pattern imposes temporal (and possibly atemporal) constraints on the input events,
which, if satisfied, lead to the detection of a complex event.
Atemporal constraints may be ``local'',
applying only to the last event read from the input stream.
For example, 
in streams from temperature sensors,
the constraint that the temperature of the last event is higher than some constant threshold would constitute such a local constraint.
More commonly, 
these constraints involve multiple events of the pattern,
e.g., 
the constraint that the temperature of the last event is higher than that of the previous event.
%The input to a CER system thus consists of two %main components: 
%a stream of events, 
%also called simple derived events\index{simple %derived event} (SDEs);
%and a set of patterns that define relations %among the SDEs.
%Instances of pattern satisfaction are called %Complex Events\index{complex event} (CEs.
%The output of the system is another stream, 
%composed of the detected CEs.
Complex events must often be detected with very low latency,
which, in certain cases, 
may even be in the order of a few milliseconds \cite{DBLP:books/daglib/0017658,DBLP:books/daglib/0024062,hedtstuck_complex_2017}.

Automata are of particular interest for the field of CER,
because they provide a natural way of handling sequences.
As a result, 
the usual operators of regular expressions, 
like concatenation, union and Kleene-star,
have often been given an implicit temporal interpretation in CER.
For example, 
the concatenation of two events is said to occur whenever the second event is read by an automaton after the first one,
i.e.,
whenever the timestamp of the second event is greater than the timestamp of the first 
(assuming the input events are temporally ordered).
On the other hand,
atemporal constraints are not easy to define using classical automata,
since they either work without memory or, 
even if they do include a memory structure,
e.g., as with push-down automata,
they can only work with a finite alphabet of input symbols.
For this reason,
the CER community has proposed several extensions of classical automata.
These extended automata have the ability to store input events and later retrieve them in order to evaluate whether a constraint is satisfied \cite{DBLP:conf/cidr/DemersGPRSW07,DBLP:conf/sigmod/AgrawalDGI08,DBLP:journals/csur/CugolaM12}.
They resemble both register automata \cite{DBLP:journals/tcs/KaminskiF94},
through their ability to store events,
and symbolic automata \cite{DBLP:conf/cav/DAntoniV17},
through the use of predicates on their transitions.
They differ from symbolic automata in that predicates apply to multiple events, 
retrieved from the memory structure that holds previous events.
They differ from register automata in that predicates may be more complex than that of (in)equality.

One issue with these CER-specific automata is that their properties have not been systematically investigated,
in contrast to models derived directly from the field of languages and automata;
see \cite{DBLP:conf/icdt/GrezRU19} for a discussion about the weaknesses of automaton models in CER. 
Moreover, they sometimes need to impose restrictions on the use of regular expression operators in a pattern, 
e.g., nesting of Kleene-star operators is not allowed.
A recently proposed formal framework for CER attempts to address these issues \cite{DBLP:conf/icdt/GrezRU19}.
Its advantage is that it provides a logic for CER patterns, 
with denotational and compositional semantics, 
but without imposing severe restrictions on the use of operators.
An automaton model is also proposed which may be conceived as a variation of symbolic transducers \cite{DBLP:conf/cav/DAntoniV17}. 
However, this automaton model can only handle ``local'' constraints,
i.e.,
the formulas on their transitions are unary and thus are applied only to the last event read.
A model which combines symbolic and register automata (called symbolic register automata) has recently been proposed in \cite{DBLP:conf/cav/DAntoniFS019}.
However, 
this work focuses on the more theoretical aspects of the proposed automaton model,
without investigating how this model may be applied to CER
(e.g., by providing a language appropriate for CER or by examining the effects of windows).

We propose a system for CER, 
based on an automaton model that is a combination of symbolic and register automata.
It has the ability to store events and its transitions have guards in the form of $n$-ary conditions.
These conditions may be applied both to the last event and to past events that have been stored. 
Conditions on multiple events are crucial in CER because they allow us to express many patterns of interest,
e.g., an increasing trend in the speed of a vehicle.
We call such automata \emph{Symbolic Register Transducers} (\srt).
\srt\ extend the expressive power of symbolic and register automata, 
by allowing for more complex patterns to be defined and detected on a stream of events.
They also extend the power of symbolic register automata,
by allowing events in a stream to be marked as belonging to a pattern match or not.
This feature is crucial in cases where we need to enumerate all complex events detected at any given timepoint (i.e., exactly report all simple events which compose the complex ones) instead of simply reporting that a complex event has been detected. 
We also present a language with which we can define patterns for complex events that can then be translated to \srt.
We call such patterns \emph{Symbolic Regular Expressions with Memory and Output} (\sremo),
as an extension of the work presented in \cite{DBLP:journals/jcss/LibkinTV15},
where \emph{Regular Expressions with Memory} (\rem) are defined and investigated.
\rem\ are extensions of classical regular expressions with which some of the terminal symbols of an expression can be stored and later be compared for (in)equality.
\sremo\ allow for more complex conditions to be used,
besides those of (in)equality.
They additionally allow each terminal sub-expression to mark an element as belonging or not to the string/match that is to be recognized,
thus acting as transducers.

Our contributions may then be summarized as follows:
\begin{itemize}
	\item We present a CER system based on a formal framework with denotational and compositional semantics, where patterns may be written as Symbolic Regular Expressions with Memory and Output (\sremo).
	\item We show how this framework subsumes, in terms of expressive power, previous similar attempts. It allows for nesting operators and selection strategies. It also allows $n$-ary expressions to be used as conditions in patterns, thus opening the way for the detection of relational patterns.
	\item We extend previous work on automata and present a computational model for patterns written in \sremo, 
	Symbolic Register Transducers (\srt),
	whose main feature is that it supports relations between multiple events in a pattern.
	Constraints with multiple events are essential in CER,
	since they are required in order to capture many patterns of interest,
	e.g., an increasing or decreasing trend in stock prices.
	\srt\ also have the ability to mark exactly those simple events comprising a complex one.
	\item We study the closure properties of \srt. 
	By extending previous results from automata theory,
	we show that, 
	in the general case, 
	\srt\ are closed under the most common operators (union, intersection, concatenation and Kleene-star), 
	but not under complement and determinization.
	Failure of closure under complement implies that negation cannot be arbitrarily (i.e., in a compositional manner) used in CER patterns.
	The negative result about determinization implies that certain techniques (like forecasting) requiring deterministic automata are not applicable.
	\item We show that, 
	by using windows, 
	\srt\ are able to retain their nice closure properties,
	i.e., they remain closed under complement and determinization.
	Windows are an indispensable operator in CER because, among others, they limit the search space for pattern matching.
	%\item We discuss the main complexity issues of automata-based models for CER and the various optimization techniques which have been proposed in the literature.
	\item We describe the implementation of a CER engine with \srt\ at its core and present relevant experimental results. Our engine is both more efficient than other engines and supports a language that is more expressive than that of other systems.
\end{itemize}

\begin{example}
\label{example:stock}
\begin{table}[t]
\centering
\caption{Example of a stream.}
\begin{tabular}{cccccccc} 
\toprule
type & B & B & B & S & S & B & ... \\ 
\midrule
id & 1 & 1 & 2 & 1 & 1 & 2 & ... \\
\midrule
price & 22 & 24 & 32 & 70 & 68 & 33 & ... \\
\midrule
volume & 300 & 225 & 1210 & 760 & 2000 & 95 & ... \\
\midrule
index & 1 & 2 & 3 & 4 & 5 & 6 & ... \\
\bottomrule
\end{tabular}
\label{table:example_stream}
\end{table}
We now introduce an example to provide intuition.
The example is that of a set of stock market ticks.
A stream is a sequence of input events,
where each such event is a tuple of the form $(\mathit{type},\mathit{id},\mathit{price},\mathit{volume})$.
The first attribute ($\mathit{type}$) is the type of transaction: 
$S$ for SELL and $B$ for BUY.
The second one ($\mathit{id}$) is an integer identifier, 
unique for each company.
It has a finite set of possible values.
The third one ($\mathit{price}$) is a real-valued number for the price of a given stock.
Finally, the fourth one ($\mathit{volume}$) is a natural number referring to the volume of the transaction.
Table \ref{table:example_stream} shows an example of such a stream.
We assume that events are temporally ordered and their order is implicitly provided through the index.
We also assume that concurrent events cannot occur, 
i.e., each index is unique to a single event.
\end{example}

In Table \ref{table:notation} we have gathered the notation that we use throughout the paper, 
along with a brief description of every symbol.
\input{notation_table}

%% file: notation_table.tex
\begin{table*}
\centering
\scriptsize
\caption{Notation used throughout the paper.}
\begin{tabular}{ccc}
\toprule
Symbol & Meaning \\
\midrule
$\mathcal{V}$, $\mathcal{U}$ & vocabulary, universe \\
\midrule
$\mathcal{L}$ ($\mathcal{L} \subseteq \mathcal{U}^{*}$) & a language over $\mathcal{U}$ \\
\midrule
$t_{i} \in \mathcal{U}$ & term / character \\
\midrule
$S=t_{1},t_{2},\cdots$, $S_{i..j}=t_{i},\cdots,t_{j}$ & stream / stream ``slice'' from index $i$ to $j$ \\
\midrule
$f(t_{1}, \cdots, t_{m})$ & function  \\
\midrule
$P$, $\top$ & relation, unary \true\ relation  \\
\midrule
$\phi$ & formula \\
\midrule 
$\mathcal{M}$ & $\mathcal{V}$-structure \\
\midrule
$\mathcal{M} \models \phi$ & $\mathcal{M}$ models $\phi$\\
\midrule
$R = \{r_{1}, \cdots, r_{k}\}$ & register variables \\
\midrule
$v: R \hookrightarrow \mathcal{U}$ & valuation \\
\midrule
$F(r_{1}, \cdots, r_{k})$ & set of all valuations on $R$ \\
\midrule
$\sharp$, $\sim$ & contents of empty register, automaton head  \\
\midrule
$(u,v) \models \phi$ & condition $\phi$ satisfied by element $u$ and valuation $v$\\
\midrule
$\epsilon$ & the ``empty'' symbol \\
\midrule
$\bullet$, $\otimes$ & outputs \\
\midrule
$e_{1} + e_{2}$, $e_{1} \cdot e_{2}$, $e^{*}$, $!e$ & regular disjunction / concatenation / iteration / negation\\
\midrule
$\circlearrowleft e$, $@ e$ & \skipany, \skipnext\ operators \\
\midrule
$e^{[1..w]}$ & windowed expression with window size $w$ \\
\midrule
$(e,S,M,v) \vdash v'$ & \parbox{8.0cm}{\centering string $S$ and match $M$ on expression $e$ with initial valuation $v$ induce valuation $v'$} \\
\midrule
$\mathit{Lang}(e)$ & language accepted by expression $e$ \\
\midrule
$\mathit{Match}(e,S)$ & matches detected by $e$ on $S$ \\
\midrule
$T$ & automaton / transducer \\
\midrule
$Q$, $q^{s}$, $Q^{f}$ & automaton states / start state / final states \\
\midrule
$\Delta$, $\delta$ & automaton transition function / transition \\
\midrule
$W$ & write registers of a transition \\
\midrule
$c=[j,q,v]$ & \parbox{8.0cm}{\centering automaton configuration ($j$ current position, $q$ current state, $v$ current valuation)} \\
\midrule
$[j,q,v] \overset{\delta}{\rightarrow} [j',q',v']$ & configuration succession \\
\midrule
$\varrho = [1,q_{1},v_{1}] \overset{\delta_{1}}{\rightarrow} \cdots \overset{\delta_{k}}{\rightarrow} [k+1,q_{k+1},v_{k+1}]$ & run of automaton $T$ over  stream $S_{1..k}$ \\
\midrule
$\mathit{Lang}(T)$ & language accepted by automaton $T$ \\
\midrule
$\mathit{Match}(T,S)$ & matches detected by $T$ on $S$ \\
\bottomrule
\end{tabular}
\label{table:notation}
\end{table*}

%% file: related.tex
\section{Related Work}

Due to their ability to naturally handle sequences of characters,
automata have been extensively adopted in CER,
where they are adapted in order to handle streams composed of tuples.
Typical cases of CER systems that employ automata are
the Chronicle Recognition System \cite{DBLP:conf/kr/Ghallab96,DBLP:conf/ijcai/DoussonM07},
Cayuga \cite{DBLP:conf/edbt/DemersGHRW06,DBLP:conf/cidr/DemersGPRSW07},
TESLA \cite{DBLP:conf/debs/CugolaM10},
SASE \cite{DBLP:conf/sigmod/AgrawalDGI08,DBLP:conf/sigmod/ZhangDI14},
CORE \cite{DBLP:conf/icdt/GrezRU19,DBLP:journals/pvldb/BucchiGQRV22} and
Wayeb \cite{DBLP:journals/vldb/AlevizosAP22,DBLP:conf/lpar/AlevizosAP18}.
There also exist systems that do not employ automata as their computational model,
e.g., 
there are logic-based systems \cite{DBLP:journals/jair/TsilionisAP22,DBLP:conf/kr/MantenoglouKA23} or systems that use trees \cite{DBLP:conf/sigmod/MeiM09},
but the standard operators of concatenation, union and Kleene-star are quite common and they may be considered as a reasonable set of core operators for CER.
The abundance of different CER systems,
employing various computational models and using various formalisms 
has recently led to some attempts to provide a unifying framework 
\cite{DBLP:conf/icdt/GrezRU19,DBLP:journals/corr/Halle17}.
Specifically, 
in \cite{DBLP:conf/icdt/GrezRU19},
a set of core CER operators is identified,
a formal framework is proposed that provides denotational semantics for CER patterns, 
and a computational model is described for capturing such patterns.
For an overview of CER languages, 
see \cite{DBLP:journals/vldb/GiatrakosAADG20},
and for a general review of CER systems, 
see \cite{DBLP:journals/csur/CugolaM12}.
In this Section,
we present previous related work along three axes.
First, 
we discuss previous theoretical work on automata that is related to CER.
We subsequently present previous automata-based CER systems.
Finally, 
we briefly discuss some solutions which are beyond the scope of CER in the strict sense of the term,
but have characteristics that are of interest to CER.
Table \ref{table:lang} summarizes our discussion and provides a compact way to compare our proposal against previous solutions.

\input{lang_table}

\subsection{Extended automaton models: theory}

Outside the field of CER,
research on automata has evolved towards various directions.
Besides the well-known push-down automata that can store elements from a finite set to a stack,
there have appeared other automaton models with memory,
such as register automata, 
pebble automata and 
data automata \cite{DBLP:journals/tcs/KaminskiF94,DBLP:journals/tocl/NevenSV04,DBLP:journals/tocl/BojanczykDMSS11}.
For a review, 
see \cite{DBLP:conf/csl/Segoufin06}.
Such models are especially useful when the input alphabet cannot be assumed to be finite,
as is often the case with CER.
Register automata (initially called finite-memory automata) constitute one of the earliest such proposals \cite{DBLP:journals/tcs/KaminskiF94}.
At each transition,
a register automaton may choose to store its current input 
(more precisely, the current input's data payload)
to one of a finite set of registers.
A transition is followed if the current input is equal to the contents of some register.
With register automata,
it is possible to recognize strings constructed from an infinite alphabet,
through the use of (in)equality comparisons among the data carried by the current input and the data stored in the registers.
However,
register automata do not always have nice closure properties,
e.g.,
they are not closed under determinization.
For an extensive study of register automata, 
see \cite{DBLP:journals/jcss/LibkinTV15,DBLP:conf/lpar/LibkinV12}.
We build on the framework presented in \cite{DBLP:journals/jcss/LibkinTV15,DBLP:conf/lpar/LibkinV12} in order to construct register automata with the ability to handle ``arbitrary'' structures,
besides those containing only (in)equality relations.

Another model that is of interest for CER is the symbolic automaton,
which allows CER patterns to apply constraints on the attributes of events.
Automata that have predicates on their transitions were already proposed in \cite{DBLP:journals/grammars/NoordG01}.
This initial idea has recently been expanded and more fully investigated in symbolic automata 
\cite{DBLP:conf/lpar/VeanesBM10,DBLP:conf/wia/Veanes13,DBLP:conf/cav/DAntoniV17}.
In symbolic automata,
transitions are equipped with formulas constructed from a Boolean algebra.
A transition is followed
if its formula,
applied to the current input,
evaluates to \true.
Contrary to register automata,
symbolic automata have nice closure properties,
but their formulas are unary and thus can only be applied to a single element from the input string.

This is one limitation that we address here.
We use \emph{Symbolic Regular Expressions with Memory and Output} (\sremo) and \emph{Symbolic Register Transducers} (\srt), 
a language and an automaton model respectively, 
that can handle $n$-ary formulas and be applied for the purposes of CER.
With \sremo\ we can designate which elements of a pattern need to be stored for later evaluation and which must be marked as being part of a match.
\sremo\ can be compiled into \srt\, 
whose transitions can apply $n$-ary formulas/conditions (with $n{>}1$) on multiple elements.
As a result,
\srt\ are more expressive than symbolic and register automata,
thus being suitable for practical CER applications,
while, at the same time,
their properties can be systematically investigated,
as in standard automata theory.
In fact, our model subsumes these two automaton models as special cases. 
It is also an extension of Symbolic Register Automata \cite{DBLP:conf/cav/DAntoniFS019},
which do not have any output on their transitions and cannot thus enumerate the detected complex events,
since they do not have the ability to mark input events as being part of match.
Moreover, the applicability of \srt\ for CER is studied here for the first time.
%For example,
%no language is provided for defining complex events.
We show precisely how \srt\ can be used for CER and how the use of \srt\ provides expressive power without sacrificing clarity and rigor.

We initially presented the results regarding \srt\ in \cite{DBLP:journals/corr/abs-1804-09999} (we called them Register Match Automata in that report).
The difference between that report and the present paper is that now we use a different formalism for expressing patterns at the language level.
However,
the automaton model remains essentially the same.
Automaton models similar to \srt\ have been independently presented in \cite{DBLP:conf/cav/DAntoniFS019} and \cite{DBLP:journals/corr/abs-2110-04032}.
In both cases,
the focus was on Symbolic Register Automata,
i.e.,
on automata without any output on their transitions.
The former work focused on an extensive theoretical analysis,
while the latter 
%(conducted by us) 
on the theoretical applicability of this type of automata for CER,
without presenting an implementation.

\subsection{Extended automaton models as applied in CER}

Automata with registers have been proposed in the past for CER,
e.g., in SASE and Cayuga. 
However, 
previous systems typically provide operational semantics and it is not always clear 
a) what operators are allowed, 
b) at which combinations 
c) what the properties of their automaton models are. 
For example, SASE's language seems to support nested Kleene operators. 
However, this is not the case. 
SASE constructs automata whose states are linearly ordered. 
Therefore, Kleene operators can only be applied to single states.
They cannot be nested and they cannot contain other expressions,
except for single events. 
As a result, disjunction is also not allowed.
Cayuga attempts to address these issues of constraints on its expressive power through the method of resubscription,
i.e., 
expressions which cannot be captured by a single automaton are compiled into multiple automata \cite{demers2005general}.
Each sub-automaton can then subscribe to the output of other automata,
thus creating a hierarchy of automata.
Although this is an interesting solution,
the resulting semantics remains ambiguous,
since the correctness and limits of this approach have not been thoroughly investigated.  
Our system does not suffer from these limitations.
Its novelty is that it provides formal, compositional semantics which allows us to address all of the above issues.
We show that negation is the only problematic operator. 
The other operators may be arbitrarily combined in a completely compositional manner and each pattern can be compiled into a single automaton, 
something which has not been previously achieved.
CORE \cite{DBLP:conf/icdt/GrezRU19,DBLP:conf/icdt/GrezRUV20} and Wayeb \cite{DBLP:journals/vldb/AlevizosAP22,DBLP:conf/lpar/AlevizosAP18} constitute two more recent automata-based CER systems.
CORE automata may be categorized under the class of ``unary'' symbolic automata 
(or transducers, to be more precise),
i.e.,
they do not support patterns relating multiple events.
The same is true for Wayeb,
which also employs ``unary'' symbolic automata. 

\subsection{Extended automaton models beyond CER}

An adaptation of finite automata in the context of Data Stream Management Systems (which have strong similarities to CER systems) has also been proposed in \cite{DBLP:journals/pvldb/ChandramouliGM10}.
These automata are called augmented finite automata (AFA) and are enriched with registers, 
in order to capture trends.
With respect to compositionality,
AFA are similar to \srt: 
Like \srt, 
Augmented Finite Automata (AFA) \cite{DBLP:journals/pvldb/ChandramouliGM10}
support arbitrary edges and are compositional.
On the other hand, 
AFA have different limitations.
Each AFA has a single register (one per active state), 
whereas there is no such restriction for \srt. 
AFA are thus less expressive than \srt.
Additionally, AFA are not transducers and cannot enumerate the input events of a complex event. 
They can report event lifetimes, i.e., the duration of a complex event. 
\srt\ can also report individual input events. 
The input events can be reconstructed in a port-processing step, if needed, from the lifetime, 
but this seems to hold only for contiguous patterns. 
It is unclear whether this is feasible for non-contiguous patterns.
Finally, the properties of AFA have not been theoretically studied, for example with respect to determinization and negation. 
AFA can handle certain instances of negation, 
but there are strong reasons to suspect that they are not in general closed under complement, 
as is the case of register automata.
In summary, \srt\ are more expressive than AFA. 

Another way to implement CER patterns, 
in relational databases, 
is through  SQL's MATCH\_RECOGNIZE,
a proposed clause that can perform pattern recognition on rows \cite{MRISO,DBLP:journals/dbsk/Petkovic22}.
MATCH\_RECOGNIZE is very expressive and can in principle capture almost any pattern expressed in a CER language.
However, it is uncertain whether it would work in a streaming setting as efficiently as CER systems.
Recent work has proposed implementations of MATCH\_RECOGNIZE that are more efficient than the one already available in Flink \cite{DBLP:journals/pvldb/ZhuHC23,DBLP:conf/sigmod/KorberGS21}.
The proposed optimizations rely on the use of prefiltering and clever indices so that the automaton responsible for pattern recognition is fed only with a small subset of the initial rows.
They target the scenario of historical analysis and their extension to a streaming setting is not considered.
It still remains an open issue whether and to what extent the proposed optimizations would work for patterns processing events in real time.

%% file: lang_table.tex
\begin{table*}[!ht]
\centering
%\footnotesize
\scriptsize
%\small
\setlength{\tabcolsep}{3pt}
\begin{tabular}{lgcgcgcgcgcg} 
\toprule
%\multicolumn{11}{c}{\bf Language Expressiveness} \\
%\midrule
System &
$\sigma_{1}$ &  % selection (unary)
$\sigma_{n}$ &     % selection (n-ary)
$\vee$ & 	% disjunction
$\wedge$ &	% conjunction
$\neg$ &	% negation		
; &			% sequence
* &  		% iteration
D &			% determinizability
E & 		% enumeration
S.P. &		% selection policies
Remarks
\\ 
\midrule		

\rowcolor{LightGray}
\multicolumn{12}{c}{\textbf{Theory}} \\
\midrule

Register automata
& \textcolor{red}{\ding{56}} % selection (unary)
& \textcolor{red}{\ding{56}} % selection (n-ary)
& \textcolor{green}{\ding{52}} % disjunction
& \textcolor{green}{\ding{52}} % conjunction
& \textcolor{red}{\ding{56}} % negation
& \textcolor{green}{\ding{52}} % sequence
& \textcolor{green}{\ding{52}} % iteration
& \textcolor{red}{\ding{56}} % determinizability
& \textcolor{red}{\ding{56}} % enumeration
& \textsf{\scriptsize Sc} % selection stategies
& \parbox{4.2cm}{\centering Selection only for unary (in-)equality.}
\\
\midrule

Symbolic automata
& \textcolor{green}{\ding{52}} % selection (unary)
& \textcolor{red}{\ding{56}} % selection (n-ary)
& \textcolor{green}{\ding{52}} % disjunction
& \textcolor{green}{\ding{52}} % conjunction
& \textcolor{green}{\ding{52}} % negation
& \textcolor{green}{\ding{52}} % sequence
& \textcolor{green}{\ding{52}} % iteration
& \textcolor{green}{\ding{52}} % determinizability
& \textcolor{red}{\ding{56}} % enumeration
& \textsf{\scriptsize Sc} % selection stategies
& \parbox{4.2cm}{}
\\
\midrule

Symbolic register automata
& \textcolor{green}{\ding{52}} % selection (unary)
& \textcolor{green}{\ding{52}} % selection (n-ary)
& \textcolor{green}{\ding{52}} % disjunction
& \textcolor{green}{\ding{52}} % conjunction
& \textcolor{red}{\ding{56}} % negation
& \textcolor{green}{\ding{52}} % sequence
& \textcolor{green}{\ding{52}} % iteration
& \textcolor{red}{\ding{56}} % determinizability
& \textcolor{red}{\ding{56}} % enumeration
& \textsf{\scriptsize Sc} % selection stategies
& \parbox{4.2cm}{}
\\
\midrule

\rowcolor{LightGray}
\multicolumn{12}{c}{\textbf{Automata-based CER solutions}} \\
\midrule

SASE
& \textcolor{green}{\ding{52}} % selection (unary)
& \textcolor{green}{\ding{52}} % selection (n-ary)
& \textcolor{red}{\ding{56}} % disjunction
& \textcolor{red}{\ding{56}} % conjunction
& \textcolor{green}{\ding{52}} % negation
& \textcolor{green}{\ding{52}} % sequence
& \textcolor{green}{\ding{52}} % iteration
& \textcolor{red}{\ding{56}} % determinizability
& \textcolor{green}{\ding{52}} % enumeration
& all % selection stategies
& \parbox{4.2cm}{\centering Iteration and selection strategies cannot be nested.\\ $\vee$, $\wedge$ and $\neg$ possible in principle but not available in source code. \\ Soundness issues with selection strategies}
\\
\midrule

Cayuga
& \textcolor{green}{\ding{52}} % selection (unary)
& \textcolor{green}{\ding{52}} % selection (n-ary)
& \textcolor{green}{\ding{52}} % disjunction
& ? % conjunction
& \textcolor{red}{\ding{56}} % negation
& \textcolor{green}{\ding{52}} % sequence
& \textcolor{green}{\ding{52}} % iteration
& \textcolor{red}{\ding{56}} % determinizability
& \textcolor{red}{\ding{56}} % enumeration
& \textsf{\scriptsize Stam} % selection stategies
& \parbox{4.2cm}{\centering Re-subscription with multiple automata for nested expressions.}
\\
\midrule

FlinkCEP
& \textcolor{green}{\ding{52}} % selection (unary)
& \textcolor{green}{\ding{52}} % selection (n-ary)
& \textcolor{green}{\ding{52}} % disjunction
& ? % conjunction
& \textcolor{green}{\ding{52}} % negation
& \textcolor{green}{\ding{52}} % sequence
& ? % iteration
& \textcolor{red}{\ding{56}} % determinizability
& \textcolor{green}{\ding{52}} % enumeration
& ? % selection stategies
& \parbox{4.2cm}{\centering Soundness issues with selection strategies and iteration.}
\\
\midrule

Esper
& \textcolor{green}{\ding{52}} % selection (unary) 
& \textcolor{green}{\ding{52}} % selection (n-ary)
& \textcolor{green}{\ding{52}} % disjunction  
& ? % conjunction
& \textcolor{green}{\ding{52}} % negation 
& \textcolor{green}{\ding{52}} % sequence 
& \textcolor{green}{\ding{52}} % iteration
& ? % determinizability
& \textcolor{green}{\ding{52}} % enumeration
& all % selection stategies
& \parbox{4.2cm}{\centering Mixture of trees, automata and Allen's interval algebra.}
\\
\midrule

CORE
& \textcolor{green}{\ding{52}} % selection (unary)
& \textcolor{red}{\ding{56}} % selection (n-ary)
& \textcolor{green}{\ding{52}} % disjunction
& ? % conjunction
& ? % negation
& \textcolor{green}{\ding{52}} % sequence
& \textcolor{green}{\ding{52}} % iteration
& \textcolor{green}{\ding{52}} % determinizability
& \textcolor{green}{\ding{52}} % enumeration
& all % selection stategies
& 
\\
\midrule

Wayeb (symbolic automata)
& \textcolor{green}{\ding{52}} % selection (unary)
& \textcolor{red}{\ding{56}} % selection (n-ary)
& \textcolor{green}{\ding{52}} % disjunction
& \textcolor{green}{\ding{52}} % conjunction
& \textcolor{green}{\ding{52}} % negation
& \textcolor{green}{\ding{52}} % sequence
& \textcolor{green}{\ding{52}} % iteration
& \textcolor{green}{\ding{52}} % determinizability
& \textcolor{red}{\ding{56}} % enumeration
& all % selection stategies
& 
\\
\midrule

\rowcolor{LightGray}
\multicolumn{12}{c}{\textbf{Beyond CER}} \\
\midrule

AFA
& \textcolor{green}{\ding{52}} % selection (unary)
& ? % selection (n-ary)
& \textcolor{green}{\ding{52}} % disjunction
& ? % conjunction
& ? % negation
& \textcolor{green}{\ding{52}} % sequence
& \textcolor{green}{\ding{52}} % iteration
& ? % determinizability
& \textcolor{red}{\ding{56}} % enumeration
& \textsf{\scriptsize Sc} % selection stategies
& \parbox{4.2cm}{\centering Partial support of negation.\\ $\sigma_{n}$ with a single register.}
\\
\midrule

MATCH\_RECOGNIZE
& \textcolor{green}{\ding{52}} % selection (unary)
& \textcolor{green}{\ding{52}} % selection (n-ary)
& \textcolor{red}{\ding{56}} % disjunction
& ? % conjunction
& \textcolor{green}{\ding{52}} % negation
& \textcolor{green}{\ding{52}} % sequence
& \textcolor{red}{\ding{56}} % iteration
& ? % determinizability
& \textcolor{red}{\ding{56}} % enumeration
& all % selection stategies
& \parbox{4.2cm}{\centering Supported features depend on the implementation.}
\\

\rowcolor{LightGray}
\multicolumn{12}{c}{\textbf{Our proposal}} \\
\midrule

Wayeb (SRT)
& \textcolor{green}{\ding{52}} % selection (unary)
& \textcolor{green}{\ding{52}} % selection (n-ary)
& \textcolor{green}{\ding{52}} % disjunction
& \textcolor{green}{\ding{52}} % conjunction
& \textcolor{green}{\ding{52}} % negation
& \textcolor{green}{\ding{52}} % sequence
& \textcolor{green}{\ding{52}} % iteration
& \textcolor{green}{\ding{52}} % determinizability
& \textcolor{green}{\ding{52}} % enumeration
& all % selection stategies
& \parbox{4.2cm}{\centering $\neg$ and determinization supported only for windowed expressions.}
\\

\bottomrule

\end{tabular}
\caption{Comparing state-of-the-art with our proposal. \newline
         $\sigma_{1}$: unary selection, $\sigma_{n}$: $n$-ary selection, $\wedge$: intersection, $\vee$: union, $\neg$: negation, 
		 ;: sequence, *: iteration, 
		 D: determinizability, E: enumeration, S.P.: selection policies,
		 \textsf{\footnotesize Stam} : \skipany, \textsf{\footnotesize Stnm} : \skipnext, \textsf{\footnotesize Sc} : \strictcont.}  
\label{table:lang}
\end{table*}

%% file: core.tex
\input{core_sremo.tex}

\input{core_srt.tex}

%% file: core_sremo.tex
\section{Symbolic Regular Expressions with Memory and Output}

The field of CER has been growing strong for the past 20 years.
It is thus no surprise that there is no lack of languages, formalisms and systems from which one may choose according to their needs.
As a result, there is considerable variability concerning the most relevant and useful operators of CER patterns, 
their semantics and the corresponding computational models to be used for the actual detecting of complex events.
On the one hand, this variability may be viewed as a sign of vigor for the field.
On the other hand, the fact that operators and their semantics are sometimes defined informally makes it hard to compare different systems in terms of their expressive capabilities.
It also makes it hard to study a single system in itself in a more systematic manner,
other than actually running it and observing its behavior. 

As an attempt to mitigate these problems,
we present and describe a framework for CER which has formal, denotational semantics.
We first present a language for CER and discuss its semantics.
The main feature of this language is that it allows for most of the common CER operators (such as selection, sequence, disjunction and iteration), without imposing restrictions on how they may be used and nested.
Our proposed language can also accommodate n-ary conditions,
i.e., we can impose constraints on the patterns which relate multiple events of a stream,
e.g., that the number of cells in a simulated tumor at the current timepoint is higher than their number at the previous timepoint. 
We also discuss the semantics of patterns written in our proposed language and show that these are well-defined.
As a result, 
in order to know whether a given stream contains any complex events corresponding to a given pattern,
we do not need to resort to a procedural computational model. 
The semantics of the language may be studied independently of the chosen computational model.
Not only is this feature critical in itself,
allowing for a systematic understanding of the use of operators,
but it could also be of importance for optimization,
which often relies on pattern re-writing,
assuming that we can know when two patterns are equivalent without actually having to run their computational models.
Previous work on CER has produced systems which are highly expressive (e.g., FlinkCEP \cite{FlinkCEP}),
but lack a proper, formal description.
Some more recent work (\cite{DBLP:journals/pvldb/BucchiGQRV22}) has attempted to construct a system which is both formal and efficient.
However, it does not support n-ary expressions, 
allowing (non-temporal) constraints which are applied only to the last event read from a stream.

Before presenting \srt,
we first present a high-level formalism for defining CER patterns.
We extend the work presented in \cite{DBLP:journals/jcss/LibkinTV15},
where the notion of regular expressions with memory (\rem) was introduced.
These regular expressions can store some terminal symbols in order to compare them later against a new input element for (in)equality.
One important limitation of \rem\ with respect to CER is that they can handle only (in)equality relations.
In this section,
we extend \rem\ so as to endow them with the capacity to use relations from ``arbitrary'' structures.
We call these extended \rem\ \emph{Symbolic Regular Expressions with Memory and Output} (\sremo).

First, 
in Section \ref{sec:formulas_models} we repeat some basic definitions from logic theory.
We also describe how we can adapt them and simplify them to suit our needs.
Next,
in Section \ref{sec:conditions} we precisely define the notion of conditions.
In \sremo,
conditions will act in a manner equivalent to that of terminal symbols in classical regular expressions.
The difference is of course that conditions are essentially logic formulas that can reference both the current element read from a string/stream and possibly some past elements.
In Section \ref{sec:sremo} we present the syntax for \sremo\ and in Section \ref{sec:sremo:semantics} the definition of their semantics.

\subsection{Formulas and models}
\label{sec:formulas_models}

In this section,
we follow the notation and notions presented in \cite{hedman2004first}.
The first notion that we need is that of a $\mathcal{V}$-structure.
A $\mathcal{V}$-structure essentially describes a domain along with the operations that can be performed on the elements of this domain and their interpretation. 
\begin{definition}[$\mathcal{V}$-structure\index{V-structure} \cite{hedman2004first}]
A vocabulary $\mathcal{V}$ is a set of function, relation and constant symbols.
A $\mathcal{V}$-structure is an underlying set $\mathcal{U}$, called a universe, and an interpretation of $\mathcal{V}$.
An interpretation assigns an element of $\mathcal{U}$ to each constant in $\mathcal{V}$, a function from $\mathcal{U}^{n}$ to $\mathcal{U}$ to each $n$-ary function in $\mathcal{V}$ and a subset of $\mathcal{U}^{n}$ to each $n$-ary relation in $\mathcal{V}$.
$\blacktriangleleft$
\end{definition}

\begin{example}
Using Example \ref{example:stock},
we can define the following vocabulary
\begin{equation*}
\mathcal{V} = \{R,c_{1},c_{2},c_{3},c_{4},c_{5},c_{6}\}
\end{equation*}
and the universe 
\begin{equation*}
\mathcal{U} = \{(B,1,22,300),(B,1,24,225),(B,2,32,1210),(S,1,70,760),(S,1,68,2000),(B,2,33,95)\}
\end{equation*}
We can also define an interpretation of $V$ by assigning each $c_{i}$ to an element of $\mathcal{U}$,
e.g., $c_{1}$ to $(B,1,22,300)$, $c_{2}$ to $(B,1,24,225)$, etc.
$R$ may also be interpreted as $R(x,y) := x.\mathit{id} = y.\mathit{id}$,
i.e., this binary relation contains all pairs of $\mathcal{U}$ which have the same $\mathit{id}$.
For example, $((B,1,22,300),(S,1,70,760)) \in R$ and  $((B,1,22,300),(B,2,33,95)) \notin R$.
If there are more (even infinite) tuples in a stream/string,
then we would also need more constants (even infinite).
\end{example}

We extend the terminology of classical regular expressions to define characters, strings and languages.
Elements of $\mathcal{U}$ are called \emph{characters}\index{character} and finite sequences of characters are called \emph{strings}\index{string}. 
A set of strings $\mathcal{L}$ constructed from elements of $\mathcal{U}$, 
i.e., 
$\mathcal{L} \subseteq \mathcal{U}^{*}$, 
where $^{*}$ denotes Kleene-star, 
is called a language over $\mathcal{U}$.
Then, a stream $S$ is an infinite sequence $S=t_{1},t_{2},\cdots$, 
where each $t_{i} \in \mathcal{U}$ is a character.
By $S_{1..k}$ we denote the sub-string of $S$ composed of the first $k$ elements of $S$.
$S_{m..k}$ denotes the slice of $S$ starting from the $\mathit{m^{th}}$ and ending at the $\mathit{k^{th}}$ element.

We now define the syntax and semantics of formulas that can be constructed from the constants, relations and functions of a $\mathcal{V}$-structure.
We begin with the definition of terms.
\begin{definition}[Term\index{term} \cite{hedman2004first}]
A term is defined inductively as follows:
\begin{itemize}
	\item Every constant is a term.
	%\item Every variable is a term.
	\item If $f$ is an $m$-ary function and $t_{1}, \cdots, t_{m}$ are terms,
	then $f(t_{1}, \cdots, t_{m})$ is also a term. $\blacktriangleleft$
\end{itemize}
\end{definition}

Using terms, relations and the usual Boolean constructs of conjunction, disjunction and negation,
we can define formulas.
\begin{definition}[Formula\index{formula} \cite{hedman2004first}]
\label{definition:formula}
Let $t_{i}$ be terms.
A formula is defined as follows:
\begin{itemize}
	\item If $P$ is an $n$-ary relation, then $P(t_{1}, \cdots, t_{n})$ is a formula (an atomic formula).
	\item If $\phi$ is a formula, $\neg \phi$ is also a formula.
	\item If $\phi_{1}$ and $\phi_{2}$ are formulas, $\phi_{1} \wedge \phi_{2}$ is also a formula.
	\item If $\phi_{1}$ and $\phi_{2}$ are formulas, $\phi_{1} \vee \phi_{2}$ is also a formula. $\blacktriangleleft$
\end{itemize}
\end{definition}
\begin{definition}[$\mathcal{V}$-formula\index{V-formula} \cite{hedman2004first}]
If $\mathcal{V}$ is a vocabulary, 
then a formula in which every function, relation and constant is in $\mathcal{V}$ is called a $\mathcal{V}$-formula.
$\blacktriangleleft$
\end{definition}

\begin{example}
Continuing with our example,
$R(c_{1},c_{4})$ is an atomic $\mathcal{V}$-formula. 
$R(c_{1},c_{4}) \wedge \neg R(c_{1},c_{3})$ is also a (complex) $\mathcal{V}$-formula,
where $\mathcal{V} = \{R,c_{1},c_{2},c_{3},c_{4},c_{5},c_{6}\}$.
\end{example}

Notice that in typical definitions of terms and formulas
(as found in \cite{hedman2004first})
variables are also present.
A variable is also a term.
Variables are also used in existential and universal quantifiers to construct formulas.
In our case, 
we will not be using variables in the above sense
(instead, as explained below, we will use variables to refer to registers).
Thus, existential and universal formulas will not be used.
In principle, 
they could be used, 
but their use would be counter-intuitive.
At every new event,
we need to check whether this event satisfies some properties,
possibly in relation to previous events.
A universal or existential formula would need to check every event
(variables would refer to events),
both past and future,
to see if all of them or at least one of them (from the universe $\mathcal{U}$) satisfy a given property.
Since we will not be using variables,
there is also no notion of free variables in formulas
(variables occurring in formulas that are not quantified).
Thus, every formula is also a sentence,
since sentences are formulas without free variables.
In what follows,
we will thus not differentiate between formulas and sentences.

We can now define the semantics of a formula with respect to a $\mathcal{V}$-structure.
\begin{definition}[Model of $\mathcal{V}$-formulas \cite{hedman2004first}]
\label{definition:models_formulas}
Let $\mathcal{M}$ be a $\mathcal{V}$-structure and $\phi$ a $\mathcal{V}$-formula.
We define $\mathcal{M} \models \phi$ ($\mathcal{M}$ models $\phi$) as follows:
\begin{itemize}
	\item If $\phi$ is atomic, i.e. $\phi = P(t_{1}, \cdots, t_{m})$, then $\mathcal{M} \models P(t_{1}, \cdots, t_{m})$ iff the tuple $(a_{1}, \cdots, a_{m})$ is in the subset of $\mathcal{U}^{m}$ assigned to $P$,
	where $a_{i}$ are the elements of $\mathcal{U}$ assigned to the terms $t_{i}$.
	\item If $\phi := \neg \psi$, then $\mathcal{M} \models \phi$ iff $\mathcal{M} \nvDash \psi$. 
	\item If $\phi := \phi_{1} \wedge \phi_{2}$, then $\mathcal{M} \models \phi$ iff $\mathcal{M} \models \phi_{1}$ and $\mathcal{M} \models \phi_{2}$. 
	\item If $\phi := \phi_{1} \vee \phi_{2}$, then $\mathcal{M} \models \phi$ iff $\mathcal{M} \models \phi_{1}$ or $\mathcal{M} \models \phi_{2}$. $\blacktriangleleft$
\end{itemize}
\end{definition}

\begin{example}
If $\mathcal{M}$ is the $\mathcal{V}$-structure of our example,
then $\mathcal{M} \models R(c_{1},c_{4})$,
since $c_{1} \rightarrow (B,1,22,300)$, $c_{4} \rightarrow (S,1,70,760)$ and $((B,1,22,300),(S,1,70,760)) \in R$.
We can also see that $\mathcal{M} \models R(c_{1},c_{4}) \wedge \neg R(c_{1},c_{3})$,
since $c_{3} \rightarrow (B,2,32,1210)$ and $((B,1,22,300),(B,2,32,1210)) \notin R$.
\end{example}

\subsection{Conditions}
\label{sec:conditions}

Based on the above definitions, 
we will now define conditions over registers. 
Conditions are the basic building blocks of \sremo.
In the simplest case,
they are applied to single events and act as filters.
In the general case,
we need them to be applied to multiple events,
some of which may be stored to registers.
Conditions will essentially be the $n$-ary guards on the transitions of \srt.

\begin{definition}[Condition\index{condition}]
\label{definition:condition}
Let $\mathcal{M}$ be a $\mathcal{V}$-structure always equipped with the unary relation $\top$ for which it holds that $u \in \top$, ${\forall u \in \mathcal{U}}$,
i.e., this relation holds for all elements of the universe $\mathcal{U}$.
Let $R = \{r_{1}, \cdots, r_{k}\}$ be variables denoting the registers and $\sim$ a special variable denoting an automaton's head which reads new elements. 
The ``contents'' of the head  always correspond to the most recent element.
We call $R$ register variables.
A condition is essentially a $\mathcal{V}$-formula,
as defined above (Definition \ref{definition:formula}),
where,
instead of terms,
we use register variables.
A condition is then defined by the following grammar:
\begin{itemize}
	\item $\top$ is a condition.
	\item $P(r_{1}, \cdots, r_{n})$, where $r_{i} \in R \cup \{ \sim \}$ and $P$ an $n$-ary relation, is a condition.
	\item $\neg \phi$ is a condition, if $\phi$ is a condition.
	\item $\phi_{1} \wedge \phi_{2}$ is a condition if $\phi_{1}$ and $\phi_{2}$ are conditions.
	\item $\phi_{1} \vee \phi_{2}$ is a condition if $\phi_{1}$ and $\phi_{2}$ are conditions. $\blacktriangleleft$
\end{itemize}
\end{definition}

\begin{example}
As an example,
consider the simple case where we want to detect stock ticks of type BUY (B),
followed by a tick of type SELL (S) for the same company.
We would thus need a simple condition on the first tick,
denoted as $\mathit{TypeIsB}(\sim)$,
where $\mathit{TypeIsB}(x){:=}x.\mathit{type}{=}B$.
$\mathit{TypeIsB}(\sim)$ has a single argument,
the automaton head.
We also need another condition for the SELL tick and the company comparison,
denoted as 
$\mathit{TypeIsS}(\sim) \wedge \mathit{EqualId}(\sim,r_{1})$.
We assume that $\mathit{TypeIsS}(x) := x.\mathit{type} = S$ and $\mathit{EqualId}(x,y) := x.\mathit{id} = y.\mathit{id}$.
Note that,
beyond the head variable,  
$\mathit{EqualId}$ also has a register variable as an input argument. 
We will show later how registers are written.
\end{example}

Since terms now refer to registers,
we need a way to access the contents of these registers.
We will assume that each register has the capacity to store exactly one element from $\mathcal{U}$.
The notion of valuations provides us with a way to access the contents of registers.
\begin{definition}[Valuation\index{valuation}]
Let $R = \{r_{1}, \cdots, r_{k}\}$ be a set of register variables.
A valuation on $R$ is a partial function $v: R \hookrightarrow \mathcal{U}$,
i.e., some registers may be ``empty''.
The set of all valuations on $R$ is denoted by $F(r_{1}, \cdots, r_{k})$.
Register update happens with $v[r_{i} \leftarrow u]$,
denoting the valuation where we replace the content of $r_{i}$ with a new element $u$,
producing a new valuation $v'$:
\begin{equation}
v'(r_{j}) = v[r_{i} \leftarrow u] = 
  \begin{cases}
    u & \quad \text{if } r_{j} = r_{i}   \\
    v(r_{j}) & \quad \text{otherwise} \\
  \end{cases}
\end{equation}
Similarly,
$v[W \leftarrow u]$, 
where $W \subseteq R$,
denotes the valuation obtained by replacing the contents of all registers in $W$ with $u$.
We say that a valuation $v$ is \emph{compatible} with a condition $\phi$ if, 
for every register variable $r_{i}$ that appears in $\phi$,
$v(r_{i})$ is defined,
i.e., 
$r_{i}$ is not empty.
We will also use the notation $v(r_{i}) = \sharp$ to denote the fact that register $r_{i}$ is empty,
i.e., we extend the range of $v$ to $\mathcal{U} \cup \{ \sharp \}$.
We also extend the domain of $v$ to $R \cup \{ \sim \}$.
By $v(\sim)$ we will denote the ``contents'' of the automaton's head,
i.e., the last element read from the string.
$\blacktriangleleft$
\end{definition}
A valuation $v$ is essentially a function with which we can retrieve the contents of any register.
%Obviously, $v(\sim) \neq \sharp$.

We can now define the semantics of conditions,
similarly to the way we defined models of $\mathcal{V}$-formulas in Definition \ref{definition:models_formulas}.
The difference is that the arguments to relations are no longer elements assigned to terms but elements stored in registers,
as retrieved by a given valuation.
\begin{definition}[Semantics of conditions\index{semantics!of condition}]
\label{definition:condition_semantics}
Let $\mathcal{M}$ be a $\mathcal{V}$-structure, 
$u \in \mathcal{U}$ an element of the universe of $\mathcal{M}$ 
and $v \in F(r_{1}, \cdots, r_{k})$ a valuation.
We say that a condition $\phi$ is satisfied by $(u,v)$,
denoted by $(u,v) \models \phi$, 
iff one of the following holds:
\begin{itemize}
	\item $\phi := \top$, i.e., $(u,v) \models \top$ for every element and valuation.
	\item $\phi := P(x_{1}, \cdots, x_{n})$, $x_{i} \in R \cup \{ \sim \}$, $v(x_{i})$ is defined for all $x_{i}$ and $u \in P(v(x_{1}), \cdots, v(x_{n}))$.
	\item $\phi := \neg \psi$ and $(u,v) \nvDash \psi$.
	\item $\phi := \phi_{1} \wedge \phi_{2}$, $(u,v) \models \phi_{1}$ and $(u,v) \models \phi_{2}$. 
	\item $\phi := \phi_{1} \vee \phi_{2}$, $(u,v) \models \phi_{1}$ or $(u,v) \models \phi_{2}$. $\blacktriangleleft$
\end{itemize}
\end{definition}

\begin{example}
Returning to our example,
we can check whether the condition $\phi_{1}{:=}\mathit{TypeIsB}(\sim)$ is satisfied by the first element of Table \ref{table:example_stream},
$(B,1,22,300)$.
We assume that we start with a valuation $v$ where all registers are empty.
Indeed $((B,1,22,300),v){\models}\phi_{1}$,
since $v(\sim)$ is defined and $(B,1,22,300){\in}\mathit{TypeIsB}(v(\sim))$.
Note that $v(\sim)$ is always defined because the automaton head always points to an element.
The only exception is when we are at the very beginning of a string,
without having read any elements.
\end{example}

\subsection{SREMO syntax}
\label{sec:sremo}

We are now in a position to define Symbolic Regular Expressions with Memory and Output \sremo.
We achieve this by combining conditions via the standard regular operators.
Conditions act as terminal expressions, 
i.e., the base case upon which we construct more complex expressions.
Each condition may be accompanied by a register variable,
indicating that an event satisfying the condition must be written to that register.
It may also be accompanied by an output,
either $\bullet$,
indicating that the event must be marked as being part of the complex event,
or $\otimes$,
indicating that the event is irrelevant and should be excluded from any detected complex events. 

\begin{definition}[Symbolic regular expression with memory and output (\sremo)]
\label{definition:sremo}
A symbolic regular expression with memory and output over a $\mathcal{V}$-structure $\mathcal{M}$ and a set of register variables $R = \{r_{1}, \cdots, r_{k}\}$ is inductively defined as follows:
\begin{enumerate}
	\item $\epsilon$ and $\emptyset$ are \sremo.
	\item If $\phi$ is a condition (as in Definition \ref{definition:condition}) and $o \in \{\bullet, \otimes\}$ an output, then $\phi \uparrow o$ is a \sremo.
	\item If $\phi$ is a condition, $o \in \{\bullet, \otimes\}$ an output and $r_{i}$ a register variable, then $\phi \uparrow o \downarrow r_{i}$ is a \sremo. 
	This is the case where we need to store the current element read from the automaton's head to register $r_{i}$.
	\item If $e_{1}$ and $e_{2}$ are \sremo, then $e_{1} + e_{2}$ is also a \sremo. This corresponds to disjunction. 
	\item If $e_{1}$ and $e_{2}$ are \sremo, then $e_{1} \cdot e_{2}$ is also a \sremo. This corresponds to concatenation. 
	\item If $e$ is a \sremo, then $e^{*}$ is also a \sremo. This corresponds to Kleene-star. $\blacktriangleleft$
	%\item \skipany\ selection policy: If $R_{1},R_{2},\cdots,R_{n}$ are SREM,
	%then $R' := \#(R_{1},R_{2},\cdots,R_{n})$ is also a SREM, with $R' := R_{1} \cdot \top^{*} \cdot R_{2} \cdot \top^{*} \cdots \top^{*} \cdot R_{n}$.
	%\item \skipnext\ selection policy: If $R_{1},R_{2},\cdots,R_{n}$ are SREM,
	%then $R' := @(R_{1},R_{2},\cdots,R_{n})$ is also a SREM, with $R' :=   R_{1} \cdot !(\top^{*} \cdot R_{2} \cdot \top^{*}) \cdot R_{2} \cdots !(\top^{*} \cdot R_{n} \cdot \top^{*}) \cdot R_{n}$. $\blacktriangleleft$
\end{enumerate} 
\end{definition}
$\epsilon$ is the regular expression (known from classical automata) satisfied by the ``empty'' string, 
i.e., without any characters.
With \sremo\ of the form $\phi \uparrow o \downarrow r_{i}$ (case 3 above),
we denote cases where we need to store the current element read from the automaton's head to register $r_{i}$.
If we additionally need to mark the event as part of the match,
we write $o = \bullet$.
We write $o = \otimes$ when we do not want to mark the current element.
Case 4 corresponds to the usual disjunction,
whereas case 5 to concatenation.
Finally, case 6 is the Kleene-star operator.
Disjunction, concatenation and Kleene-star are the three standard operators in regular expressions which are also used here.
We will see later if and under which requirements other possible operators,
like intersection and negation,
may also be added to \sremo.

\begin{example}
\label{example:b_seq_s_filter_eq_id}
We now have everything we need to express the pattern of our example.
Consider the following \sremo:
\begin{equation}
\label{srem:b_seq_s_filter_eq_id}
\begin{aligned}
e_{1} := & (\mathit{TypeIsB}(\sim) \uparrow \bullet \downarrow r_{1}) \cdot (\top \uparrow \otimes)^{*} \cdot \\
		& ((\mathit{TypeIsS}(\sim) \wedge \mathit{EqualId}(\sim,r_{1})) \uparrow \bullet)
\end{aligned}
\end{equation}
$e_{1}$ first looks for elements of type BUY.
When it finds one,
it marks it as belonging to a (candidate match) and writes it to register $r_{1}$.
$r_{1}$ stores the whole element.
For example,
if $e_{1}$ starts processing the stream of Table \ref{table:example_stream},
after reading the first element,
$r_{1}$ will have stored $(B,1,22,300)$.
$e_{1}$ can then skip any number of elements,
without marking or storing them,
until encountering a SELL element from the same company.
It marks this event as part of the match as well.  
\end{example}

\subsection{SREMO semantics}
\label{sec:sremo:semantics}

In order to define the semantics of \sremo,
we need to define how the contents of the registers may change.
We thus need to first define how a \sremo,
starting from a given valuation $v$ and reading a given string $S$,
reaches another valuation $v'$.
Our final aim is to detect matches of a \sremo\ $e$ in a string $S{=}t_{1},\cdots,t_{n}$.
A match $M{=}\{i_{1},\cdots,i_{k}\}$ of $e$ on $S$ is a totally ordered set of natural numbers,
referring to indices in the string $S$,
i.e., $i_{1} \geq 1$ and $i_{k} \leq n$.
If $M{=}\{i_{1},\cdots,i_{k}\}$ is a match of $e$ on $S$,
then the set of elements referenced by $M$,
$S[M]{=}\{t_{i_{1}}, \cdots, t_{i_{k}} \}$ represents a \emph{complex event}.
We write $M = M_{1} \cdot M_{2}$ for two matches $M_{1}$, $M_{2}$ to denote the fact that 
$M_{1} \cap M_{2} = \emptyset$, $M_{1} \cup M_{2} = M$ and $max(M_{1}) < min(M_{2})$.

\begin{definition}[Semantics of \sremo]
\label{definition:sremo_semantics}
Let $e$ be a \sremo\ over a $\mathcal{V}$-structure $\mathcal{M}$ and a set of register variables $R = \{r_{1}, \cdots, r_{k}\}$, 
$S$ a string constructed from elements of the universe of $\mathcal{M}$,
$M$ a candidate match of $e$ on $S$ and 
$v,v' \in F(r_{1}, \cdots, r_{k})$.
We define the relation $(e,S,M,v) \vdash v'$ as follows:
\begin{enumerate}
	\item $(\epsilon,S,M,v) \vdash v'$ iff $S = \epsilon$ and $v=v'$ (by definition, $M=\emptyset$).
	\item $(\phi \uparrow o, S, M, v) \vdash v'$ iff $\phi \neq \epsilon$, $S=u$, $(u,v) \models \phi$, $v' = v$ and 					\begin{equation*}
  \begin{cases}
    o=\otimes\ \text{and } M=\emptyset & \quad \text{or }   \\
    o=\bullet\ \text{and } M=\{i_{u}\} & \quad \text{} \\
  \end{cases}
\end{equation*}
where $i_{u}$ is the index of $u$. 
	\item $(\phi \uparrow o \downarrow r_{i}, S, M, v) \vdash v'$ iff $\phi \neq \epsilon$, $S=u$, $(u,v) \models \phi$, $v' = v[r_{i} \leftarrow u]$ and 	\begin{equation*}
  \begin{cases}
    o=\otimes\ \text{and } M=\emptyset & \quad \text{or }   \\
    o=\bullet\ \text{and } M=\{i_{u}\} & \quad \text{} \\
  \end{cases}
	\end{equation*}. 
	\item $(e_{1} \cdot e_{2}, S, M, v) \vdash v'$ iff $S{=}S_{1} \cdot S_{2}$ and $M{=}M_{1} \cdot M_{2}$: $(e_{1},S_{1},M_{1},v) \vdash v''$ and $(e_{2},S_{2},M_{2},v'') \vdash v'$.
	\item $(e_{1} + e_{2}, S, M, v) \vdash v'$ iff $(e_{1},S,M,v) \vdash v'$ or $(e_{2},S,M,v) \vdash v'$.
	\item $(e^{*}, S, v) \vdash v'$ iff 
	\begin{equation*}
  \begin{cases}
    S=\epsilon\ \text{and } v'=v & \quad \text{or }   \\
    S{=}S_{1} \cdot S_{2}, M{=}M_{1} \cdot M_{2}: (e,S_{1},M_{1},v) \vdash v''\ \text{and }  (e^{*},S_{2},M_{2},v'') \vdash v' & \quad \text{} \\
  \end{cases}
	\end{equation*} $\blacktriangleleft$
\end{enumerate}
\end{definition}

In the first case,
we have an $\epsilon$ \sremo.
It may reach another valuation only if it reads an $\epsilon$ string and this new valuation is the same as the initial one, i.e., the registers do not change.
In the second case, 
where we have a condition $\phi \neq \epsilon$,
we move to a new valuation only if the condition is satisfied with the current element $u$ and the given register contents $v$.
Again, the registers do not change.
Additionally, if $o = \bullet$,
we accept the current element (its index $i_{u}$) as part of the match $M$.
Otherwise, 
we ignore it.
The third case is similar to the second,
with the important difference that the register $r_{i}$ needs to change and to store the current element.
For the fourth case (concatenation),
we need to be able to break the initial string into two sub-strings such that the first one reaches a certain valuation and the second one can start from this new valuation and reach another one.
Similarly,
the fifth case represents a disjunction of \sremo.
Finally, 
the sixth case (iteration) requires that we break the initial string into multiple sub-strings such that each one of these sub-strings can reach a valuation and the next one can start from this valuation and reach another one.

Based on the above definition,
we may now define the language that a \sremo\ accepts and the matches that it detects on a string $S$.
The language of a \sremo\ contains all the strings with which we can reach a valuation,
starting from the empty valuation,
where all registers are empty.
The set of matches is composed of all the matches computed after a \sremo\ has processed a string $S$.
\begin{definition}[Language accepted and matches detected by a \sremo]
\label{definition:language_matches_sremo}
The language accepted by a \sremo\ $e$ is defined as $\mathit{Lang}(e){=}\{S \mid (e,S,M,\sharp) \vdash v\}$ for some valuation $v$ and some match $M$ of $e$ on the corresponding $S$,
where $\sharp$ denotes the valuation in which no $v(r_{i})$ is defined,
i.e., all registers are empty.
The matches detected by a \sremo\ $e$ on a string $S$ is defined as $\mathit{Match}(e,S){=}\{M \mid (e,S,M,\sharp) \vdash v\}$ for some valuation $v$.
$\blacktriangleleft$
\end{definition}

\begin{example}
We can now continue with our example.
If we feed the string/stream of Table \ref{table:example_stream} to \sremo\ \eqref{srem:b_seq_s_filter_eq_id} of Example \ref{example:b_seq_s_filter_eq_id},
then we will have the following.
First,
we apply case (4) of Definition \ref{definition:sremo_semantics}.
We have $$e_{1} := (\mathit{TypeIsB}(\sim) \uparrow \bullet \downarrow r_{1})$$
and $$e_{2} := (\top \uparrow \otimes)^{*} \cdot ((\mathit{TypeIsS}(\sim) \wedge \mathit{EqualId}(\sim,r_{1})) \uparrow \bullet)$$
We break the string $S$ of Table \ref{table:example_stream} into two sub-strings $S_{1}$ and $S_{2}$,
where $S_{1}$ is the first element of $S$,
$(B,1,22,300)$,
and $S_{2}$ the remaining five.
We check whether $S_{1}$ satisfies $e_{1}$,
by applying case (3) of Definition \ref{definition:sremo_semantics}.
Since the type of $S_{1}$ is $B$ (BUY),
we will move on and store $(B,1,22,300)$ to register $r_{1}$,
i.e., we will move from the empty valuation where $v(r_{1}) = \sharp$ to $v'$,
where $v'(r_{1}) = (B,1,22,300)$.
We will also ``accept'' this element as part of a potential future match.
We then check $e_{2}$ and $S_{2}$.
We apply again case (4) of Definition \ref{definition:sremo_semantics}
and break $e_{2}$ into $(\top \uparrow \otimes)^{*}$ and $((\mathit{TypeIsS}(\sim) \wedge \mathit{EqualId}(\sim,r_{1})) \uparrow \bullet)$.
Then, the sub-expression $(\top \uparrow \otimes)^{*}$ lets us skip and ignore any number of elements.
We can thus skip the second and third elements without changing the register contents.
Now, upon reading the fourth element $(S,1,70,760)$,
there are two options.
Either skip it again to read the fifth element or try to move on by checking the sub-expression $(\mathit{TypeIsS}(\sim) \wedge \mathit{EqualId}(\sim,r_{1}) \uparrow \bullet)$.
This latter condition is actually satisfied,
since the type of this element is indeed $S$ and its $\mathit{id}$ is equal to the $\mathit{id}$ of the element store in $r_{1}$.
Thus, $S_{1..4}$ is indeed accepted by $e_{1}$.
$M = \{1,4\}$ is also a match of $e_{1}$ on $S$ (and on $S_{1..4}$).
With a similar reasoning we can see that the same is also true for $S_{1..5}$ and $M=\{1,5\}$,
had we chosen to skip the fourth element.
\end{example}

It can be shown that the concept of matches ``subsumes'' that of languages,
i.e., 
if two \sremo\ have the same matches for every string $S$,
then they also have the same languages.
\begin{theorem}
\label{theorem:matchesInduceLanguage}
Let $e,e'$ be two \sremo. 
If, for every string $S$, $\mathit{Match}(e,S) = \mathit{Match}(e',S)$,
then $\mathit{Lang}(e) = \mathit{Lang}(e')$.
\end{theorem}
\begin{proof}
The proof may be found in the Appendix, see \ref{sec:proof:sremo2srt}.
\end{proof}

%\subsection{Expressive power of SREMO}

The above introduction highlights the expressiveness,
flexibility and formal semantics of \sremo.
\sremo\ can express relational patterns with $n$-ary constraints,
by being able to relate the most recently read element with any of the preceding ones.
They also allow for arbitrary nesting of the regular operators, 
without imposing ad hoc restrictions.
Moreover, their expressive power is combined with clear, denotational semantics.

%% file: core_srt.tex
\section{Symbolic Register Transducers}

We now show how \sremo\ can be translated to an appropriate automaton model and how this model may then be used to perform CER.
We also study the closure properties of this automaton model.

\subsection{Definition of Symbolic Register Transducers}

In order to capture \sremo, 
we propose Symbolic Register Transducers (\srt), 
an automaton model equipped  with memory, logical conditions on its transitions and a single output on every transition.
The basic idea is the following.
We add a set of registers $R$ to an automaton in order to be able to store events from the stream
that will be used later in $n$-ary formulas. 
Each register can store at most one event.
In order to evaluate whether to follow a transition or not,
each transition is equipped with a guard, in the form of a Boolean formula.
If the formula evaluates to \true, then the transition is followed.
Since a formula might be $n$-ary, with $n{>}1$,
the values passed to its arguments during evaluation may be either the current event
or the contents of some registers,
i.e.,
some past events.
In other words, the transition is also equipped with a \textit{register selection}.
Before evaluation, the automaton reads the contents of the required registers,
passes them as arguments to the formula 
and the formula is evaluated.
Additionally, if, during a run of the automaton, a transition is followed,
then the transition has the option to write the event that triggered it
to some of the automaton's registers.
These are called its \textit{write registers} $W$,
i.e.,
the registers whose contents may be changed by the transition.
Finally, each transition, when followed,
produces an output,
either $\otimes$,
denoting that the event is not part of the match for the pattern that the \srt\ tries to capture,
or $\bullet$,
denoting that the event is part of the match.
We also allow for $\epsilon$-transitions, as in classical automata,
i.e., 
transitions that are followed without consuming any events and without altering the contents of the registers.

We now formally define \srt.
To aid understanding,
we present three separate definitions:
one for the automaton itself,
one for its configurations and one for its runs.
The first concerns the automaton itself, 
describing its structure,
i.e., its states and transitions. 
The remaining two describe the running behavior of a \srt.
For this we need to know its current state and register contents after every new event,
i.e., its so-called configuration.
We also need to know how the automaton changes configurations and how such a succession of configurations  (a so-called run) may lead to a match.

\begin{definition}[Symbolic Register Transducer]
\label{definition:srt}
A symbolic register transducer (\srt) with $k$ registers over a $\mathcal{V}$-structure $\mathcal{M}$ is a tuple ($Q$, $q_{s}$, $Q_{f}$, $R$, $\Delta$)
where 
\begin{itemize}
	\item $Q$ is a finite set of states,
	\item $q_{s} \in Q$ the start state, 
	\item $Q_{f}\subseteq Q$ the set of final states, 
	\item $R = (r_{1}, \cdots, r_{k})$ a finite set of registers and
	\item $\Delta$ the set of transitions.
\end{itemize}
A transition $\delta \in \Delta$ is a tuple $(q,\phi,W,q',o)$, 
also written as $q,\phi \uparrow o \downarrow W  \rightarrow q'$,
where
\begin{itemize}
	\item $q,q' \in Q$, where $q$ is the source and $q'$ the target state, 
	\item $\phi$ is a condition, as per Definition \ref{definition:condition} or $\phi = \epsilon$,
	\item $W \in 2^{R}$ are the write registers and
	\item $o \in \{\otimes,\bullet\}$ is the output. $\blacktriangleleft$
\end{itemize}
\end{definition}

We will use the dot notation to refer to elements of tuples.
For example, 
if $T$ is a \srt, 
then $T.Q$ is the set of its states.
For a transition $\delta$,
we will also use the notation $\delta.\mathit{source}$ and $\delta.\mathit{target}$ to refer to its source and target states respectively.

\begin{figure}[t]
\begin{centering}
\includegraphics[width=0.75\linewidth]{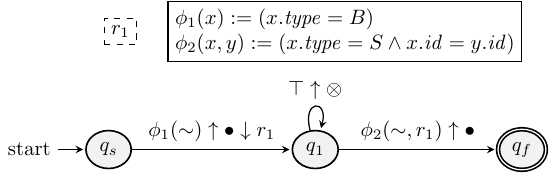}
\caption{\srt\ corresponding to the \sremo\ of eq. \eqref{srem:b_seq_s_filter_eq_id}.}
\label{fig:example1}
\end{centering}
\end{figure}

\begin{example}
As an example,
consider the \srt\ of Figure \ref{fig:example1}.
Each transition is represented as $\phi \uparrow o \downarrow W$,
where $\phi$ is its condition, $o$ its output and $W$ its set of write registers 
(or simply $r_{i}$ if only a single register is written).
$W$ may also be an empty set,
implying that no register is written.
In this case, 
we avoid writing $W$ on the transition
(see, for example, the transition from $q_{1}$ to $q_{f}$ in Figure \ref{fig:example1}).
$o$ may be omitted, 
in which case it is implicitly assumed that $o = \otimes$.
The definitions for the conditions of the transitions are presented in a separate box,
above the \srt.
Note that the arguments of the conditions correspond to registers, 
through the register selection.
Take the transition from $q_{s}$ to $q_{1}$ as an example.
It takes the last element consumed from the string/stream ($\sim$) and
passes it as argument to the unary formula $\phi_{1}$.
If $\phi_{1}$ evaluates to \true,
it writes this last event to register $r_{1}$,
displayed as a dashed square in Figure \ref{fig:example1}.
On the other hand, 
the transition from $q_{1}$ to $q_{f}$ uses both the current element and the element stored in $r_{1}$ ($(\sim,r_{1})$) and passes them to the binary formula $\phi_{2}$.
The condition $\top$ (in the self-loop of $q_{1}$) is a unary condition that always evaluates to \true\ and allows us to skip and ignore any number of events.
The \srt\ of Figure \ref{fig:example1} captures the \sremo\ of eq. \eqref{srem:b_seq_s_filter_eq_id}.
\end{example}

We can describe formally the rules for the behavior of a \srt\ through the notion of configuration:
\begin{definition}[Configuration of \srt]
\label{definition:configuration}
Assume a string
$S=t_{1},t_{2},\cdots,t_{l}$ 
and a \srt\ $T$ consuming $S$.
A configuration of $T$ is a triple
$c=[j,q,v] \in \mathbb{N} \times Q \times F(r_{1}, \cdots, r_{k})$, 
where
\begin{itemize}
	\item $j$ is the index of the next event/character to be consumed,
	\item $q$ is the current state of $T$ and
	\item $v$ the current valuation, i.e., the current contents of $T$'s registers.
\end{itemize}
We say that $c'=[j',q',v']$ is a \emph{successor} of $c$ iff one of the following holds:
\begin{itemize}
	\item $\exists \delta: \delta.\mathit{source} = q,\ \delta.\mathit{target}=q',\ \delta.\phi = \epsilon,\ j'=j,\ v'=v$, i.e., if this is an $\epsilon$ transition, we move to the target state without changing the index or the registers' contents.
	\item $\exists \delta: \delta.\mathit{source} = q,\ \delta.\mathit{target}=q',\ \delta.W = \emptyset,\ (t_{j},v) \models \delta.\phi,\ j'=j+1,\ v'=v$, i.e., if the condition is satisfied according to the current event and the registers' contents and there are no write registers, we move to the target state, we increase the index by 1 and we leave the registers untouched.
	\item $\exists \delta: \delta.\mathit{source} = q,\ \delta.\mathit{target}=q',\ \delta.W \neq \emptyset,\ (t_{j},v) \models \delta.\phi,\ j'=j+1,\ v'=v[W \leftarrow t_{j}]$, i.e., if the condition is satisfied according to the current event and the registers' contents and there are write registers, we move to the target state, we increase the index by 1 and we replace the contents of all write registers (all $r_{i} \in W$) with the current element from the string. $\blacktriangleleft$
\end{itemize}
\end{definition}
We denote a succession of configurations by $[j,q,v] \rightarrow [j',q',v']$,
or $[j,q,v] \overset{\delta}{\rightarrow} [j',q',v']$ if we need to refer to the transition as well.
For the initial configuration,
before any elements have been consumed,
we assume that
$j=1$, $q=q_{s}$ and $v(r_{i}) = \sharp,\ \forall r_{i} \in R$.
In order to move to a successor configuration,
we need a transition whose condition evaluates to \true,
when applied to $\sim$, 
if it is unary, 
or to $\sim$ and the contents of its register selection, 
if it is $n$-ary.
If this is the case, 
we move one position ahead in the stream and update the contents of this transition's write registers,
if any, 
with the event that was read. 
If the transition is an $\epsilon$ transition, 
we do not move the stream pointer 
(since $\epsilon$ transitions are followed ``spontaneously'', 
without reading any events) 
and do not update the registers,
but only move to the next state.

The actual behavior of a \srt\ upon reading a stream is captured by the notion of the run:
\begin{definition}[Run of \srt\ over string/stream]
\label{definition:run}
A run $\varrho$ of a \srt\ $T$ over a stream $S=t_{1},\cdots,t_{n}$ is a sequence of successor configurations
$[1,q_{1},v_{1}] \overset{\delta_{1}}{\rightarrow} [2,q_{2},v_{2}] \overset{\delta_{2}}{\rightarrow} \cdots \overset{\delta_{n}}{\rightarrow} [n+1,q_{n+1},v_{n+1}]$.
%where $q_{1} = A.q_{s}$ and $v_{1} = \sharp$.
A run is called accepting iff $q_{n+1} \in T.Q_{f}$ and $\delta_{n}.o=\bullet$.
By $\mathit{Match}(\varrho)$ we denote all the indices in the string that were ``marked'' by the run,
i.e.,
$\mathit{Match}(\varrho){=}\{ i \in [1,n]: \delta_{i}.o{=}\bullet\}$.
$\blacktriangleleft$
\end{definition}

The set of all runs over a stream $S$ that $T$ can follow is denoted by $\mathit{Run}(T,S)$
and the set of all accepting runs by $\mathit{Run_{f}}(T,S)$.

\begin{example}
An accepting run of the \srt\ of Figure \ref{fig:example1}, 
while consuming the first four events from the stream of Table \ref{table:example_stream}, 
is the following:
\begin{equation}
\label{run:example}
\begin{aligned}
\varrho = & [1,q_{s},\sharp] \overset{\delta_{s,1}}{\rightarrow} [2,q_{1},(B,1,22,300)] \overset{\delta_{1,1}}{\rightarrow} [3,q_{1},(B,1,22,300)] \overset{\delta_{1,1}}{\rightarrow} \\
		& [4,q_{1},(B,1,22,300)] \overset{\delta_{1,f}}{\rightarrow} [5,q_{f},(B,1,22,300)]
\end{aligned}
\end{equation}
Transition subscripts in this example refer to states of the \srt,
e.g.,
$\delta_{s,s}$ is the transition from the start state to itself,
$\delta_{s,1}$ is the transition from the start state to $q_{1}$, etc.
Note that the valuation (contents or register $r_{1}$) changes only once,
from $\sharp$ (empty) to $(B,1,22,300)$,
after the transition from $q_{s}$ to $q_{1}$ with the first event.
For the remaining configurations,
the valuation remains the same. 
This is the only transition that writes to $r_{1}$.
The contents of $r_{1}$ are retrieved and used in the last transition,
from $q_{1}$ to $q_{f}$.
See also Figure \ref{fig:example_run}.
Run \eqref{run:example} is not the only run,
since the \srt\ could have followed other transitions with the same input,
e.g.,
moving directly from $q_{s}$ to $q_{1}$.
Another possible (and non-accepting) run would be the one where the \srt\ always remains in $q_{1}$ after its first transition.
\begin{figure}
\centering
\begin{subfigure}[t]{0.49\textwidth}
	\includegraphics[width=0.99\textwidth]{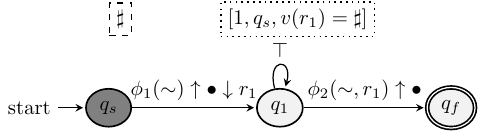}
	\caption{Initial configuration.}
\end{subfigure}
\begin{subfigure}[t]{0.49\textwidth}
	\includegraphics[width=0.99\textwidth]{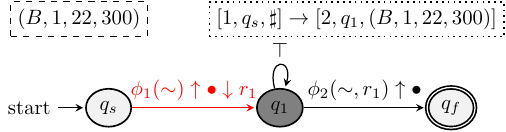}
	\caption{Configuration after reading $t_{1}$.}
\end{subfigure}\\
\begin{subfigure}[t]{0.49\textwidth}
	\includegraphics[width=0.99\textwidth]{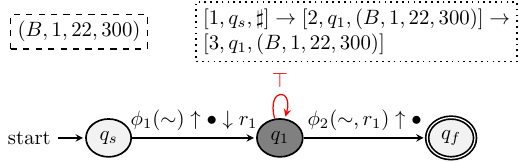}
	\caption{Configuration after reading $t_{2}$.}
\end{subfigure}
\begin{subfigure}[t]{0.49\textwidth}
	\includegraphics[width=0.99\textwidth]{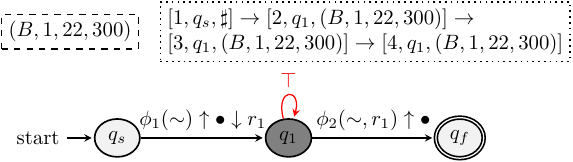}
	\caption{Configuration after reading $t_{3}$.}
\end{subfigure}\\
\begin{subfigure}[t]{0.5\textwidth}
	\includegraphics[width=0.99\textwidth]{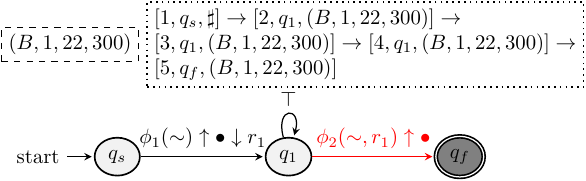}
	\caption{Configuration after reading $t_{4}$.}
\end{subfigure}
\caption{A run of the \srt\ of Figure \ref{fig:example1}, while consuming the first four events from the stream of Table \ref{table:example_stream}. Triggered transitions are shown in red and the current state of the \srt\ in dark gray.
The dashed box represents a register.
The contents of the register at each configuration are shown inside the dashed box.
Inside the dotted boxes, the run is shown.}
\label{fig:example_run}
\end{figure}
\end{example}

Finally,
we can define the language of a \srt\ as the set of strings for which the \srt\ has an accepting run,
starting from an empty configuration.
Similarly, 
the matches of a \srt\ on a string are the matches the \srt\ ``produces'' by marking the input elements as it reads the string,
starting from an empty configuration. 
\begin{definition}[Language recognized and matches detected by \srt]
\label{definition:sra_language}
We say that a \srt\ $T$ accepts a string $S$ iff there exists an accepting run $\varrho=[1,q_{1},v_{1}] \overset{\delta_{1}}{\rightarrow} [2,q_{2},v_{2}] \overset{\delta_{2}}{\rightarrow} \cdots \overset{\delta_{n}}{\rightarrow} [n+1,q_{n+1},v_{n+1}]$ of $T$ over $S$,
where $q_{1} = T.q_{s}$ and $v_{1} = \sharp$.
The set of all strings accepted by $T$ is called the language recognized by $T$ and is denoted by $\mathit{Lang}(T)$.
The set of matches detected by $T$ on a string $S$ is defined as $\mathit{Match}(T,S) = \{\mathit{Match}(\varrho) \mid \varrho \in  \mathit{Run_{f}}(T,S) \}$.
$\blacktriangleleft$
\end{definition}

\subsection{Properties of SRT}

We now study the properties of \srt.
First, we prove that \sremo\ can be compiled to \srt.
We then show that \srt\ are closed under union, intersection, concatenation and Kleene-start but not under complement and determinization.
We can thus construct \sremo\ and \srt\ by using arbitrarily (in whatever order and depth is required) the four basic operators of union, intersection, concatenation and Kleene-star.
However, 
the negative result about complement suggests that the use of \emph{negation} in CER patterns cannot be equally arbitrary.
Moreover, 
deterministic \srt\ cannot be used in cases where this might be required,
as in Complex Event Forecasting \cite{DBLP:journals/vldb/AlevizosAP22}.
If, however, we use an extra window operator,
effectively limiting the length of strings accepted by a \srt,
we can then show that closure under complement and determinization is also possible. 

%\subsubsection{Translation of \sremo\ to \srt}
\label{section:sremo2srt}

We first prove that,
for every \sremo\ there exists an equivalent \srt.
The proof is constructive,
similar to that for classical automata.
%For the inverse direction,
%i.e. converting a \srt\ to an equivalent \sremo,
%we use the notion of generalized \srt.
%These are \srt\ which have complete \sremo\ on their transitions.
%By incrementally removing states from the \srt,
%we are finally left with two states and the \sremo\ which connects them is the \sremo\ we are looking for.
%We now show how,
%for each \sremo\,
%we can construct an equivalent \srt.
Equivalence between an expression $e$ and a \srt\ $T$ means that they recognize the same language and have the same matches.
See Definitions \ref{definition:language_matches_sremo} and \ref{definition:sra_language}.
\begin{theorem}
\label{theorem:sremo2srt}
For every \sremo\ $e$ there exists an equivalent \srt\ $T$, i.e., a \srt\ such that $\mathit{Lang}(e) = \mathit{Lang}(T)$ and $\mathit{Match}(e,S)=\mathit{Match}(T,S)$ for every string $S$.
\end{theorem}
\begin{proof}
The complete construction process and proof may be found in Appendix \ref{sec:proof:sremo2srt}.
\end{proof}

\begin{figure}[!ht]
\centering
\begin{subfigure}[t]{0.76\textwidth}
	\includegraphics[width=0.99\textwidth]{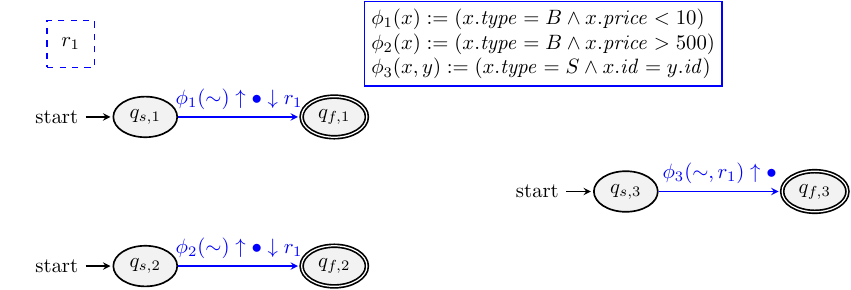}
	\caption{Constructing \srt\ for terminal sub-expressions.}
	\label{fig:sremo2srt:example1}
\end{subfigure}\\
\begin{subfigure}[t]{0.76\textwidth}
	\includegraphics[width=0.99\textwidth]{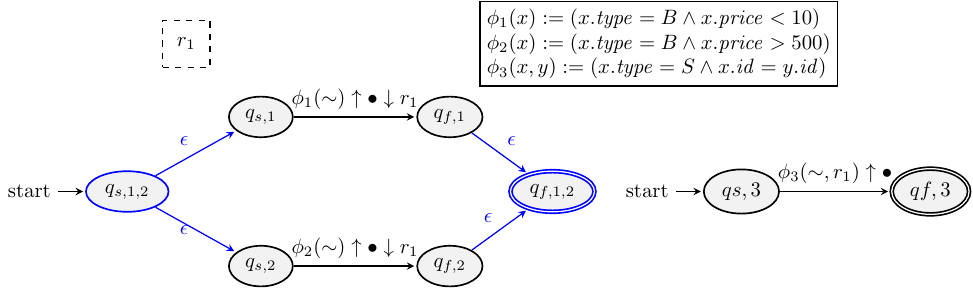}
	\caption{Connecting \srt\ via disjunction.}
	\label{fig:sremo2srt:example2}
\end{subfigure}\\
\begin{subfigure}[t]{0.76\textwidth}
	\includegraphics[width=0.99\textwidth]{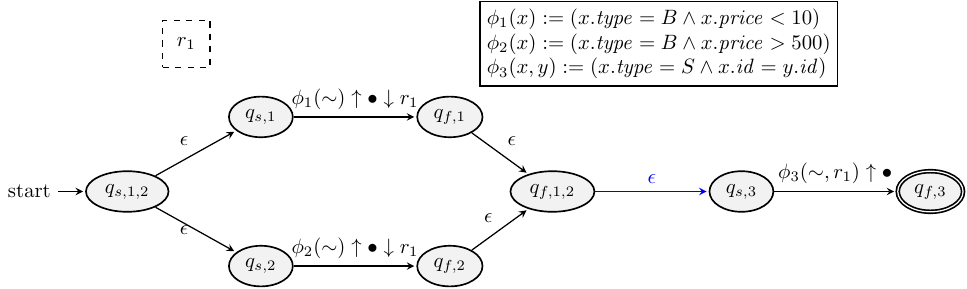}
	\caption{Connecting \srt\ via concatenation.}
	\label{fig:sremo2srt:example3}
\end{subfigure}
\caption{Constructing \srt\ from \sremo\ \eqref{sremo:sremo2srt_example}. New elements added at every step are shown in blue.}
\label{fig:sremo2srt:example}
\end{figure}

\begin{example}
Here, 
we present an example, 
to give the intuition.
Let 
\begin{equation}
\label{sremo:sremo2srt_example}
\begin{aligned}
e_{2} := & ((\phi_{1}(\sim) \uparrow \bullet \downarrow r_{1})\ + (\phi_{2}(\sim) \uparrow \bullet \downarrow r_{1})) \cdot \\
		& (\phi_{3}(\sim,r_{1}) \uparrow \bullet)
\end{aligned}
\end{equation}
be a \sremo,
where 
\begin{equation*}
\label{sremo:sremo2srt_example_conditions}
\begin{aligned}
\phi_{1}(x) := & (x.\mathit{type}=B \wedge x.\mathit{price} < 10) \\
\phi_{2}(x) :=	& (x.\mathit{type}=B \wedge x.\mathit{price} > 500) \\
\phi_{3}(x,y) :=	& (x.\mathit{type}=S \wedge x.\mathit{id} = y.\mathit{id})
\end{aligned}
\end{equation*}
With this expression,
we want to monitor stocks for possible suspicious transactions.
We want to detect cases where a stock is initially bought at a very low or high price
(first line of \sremo\ \eqref{sremo:sremo2srt_example})
and then sold.
The last condition is a binary formula,
applied to both $\sim$ and $r_{1}$.
It ensures that matches refer to the same company.
Figure \ref{fig:sremo2srt:example} shows the process for constructing the \srt\ which is equivalent to \sremo\ \eqref{sremo:sremo2srt_example}.

The algorithm is compositional,
starting from the base cases $e {:=} \phi \uparrow o$ or $e {:=} \phi \uparrow o \downarrow W$.
The three regular expression operators (concatenation, disjunction, Kleene-star) are handled in a manner almost identical as for classical automata.
The subtlety here concerns the handling of registers.
The simplest solution is to gather from the very start all registers mentioned in any sub-expressions of the original \sremo\ $e$,
i.e. any registers in the register selection of any transitions and any write registers.
We first create those registers and then start the construction of the sub-automata.
Note that some registers may be mentioned in multiple sub-expressions (e.g., in one that writes to it and then in one that reads its contents).
We only add such registers once.
We treat the registers as a set with no repetitions.

For the example of Figure \ref{fig:sremo2srt:example},
only one register is mentioned, 
$r_{1}$.
We start by creating this register.
Then, we move on to the terminal sub-expressions.
There are three basic sub-expressions and three basic automata are constructed:
from $q_{s,1}$ to $q_{f,1}$,
from $q_{s,2}$ to $q_{f,2}$ and
from $q_{s,3}$ to $q_{f,3}$.
See Figure \ref{fig:sremo2srt:example1}.
To the first two transitions,
we add the relevant \emph{unary} conditions,
e.g.,
we add $\phi_{1}(x){:=} (x.\mathit{type}{=}B\ {\wedge}\ x.\mathit{price}{<}10)$ to $q_{s,1} {\rightarrow} q_{f,1}$.
To the third transition,
we add the relevant \emph{binary} condition
$\phi_{3}(x,y) := (x.\mathit{type}=S \wedge x.\mathit{id} = y.\mathit{id})$.
The $+$ operator is handled by joining the \srt\ of the disjuncts through new states and $\epsilon$-transitions.
See Figure \ref{fig:sremo2srt:example2}. 
The concatenation operator is handled by connecting the \srt\ of its sub-expressions through an $\epsilon$-transition,
without adding any new states.
See Figure \ref{fig:sremo2srt:example3}. 
Iteration,
not applicable in this example,
is handled by joining the final state of the original automaton to its start state through an $\epsilon$-transition.
\end{example}

The standard result about $\epsilon$ elimination also holds,
stating that we can always eliminate all $\epsilon$ transitions from a \srt\ to get an equivalent \srt\ with no $\epsilon$ transitions.
\begin{lemma}
\label{lemma:epsilon}
For every \srt\ $T_{\epsilon}$ with $\epsilon$ transitions there exists an equivalent  \srt\ $T_{\notin}$ without $\epsilon$ transitions, i.e., a \srt\ such that $\mathit{Match}(T_{\epsilon},S)=\mathit{Match}(T_{\notin},S)$ for every string $S$.
\end{lemma}
\begin{proof}
See Appendix \ref{sec:proof:epsilon}.
\end{proof}

%\begin{theorem}
%\label{theorem:srt2srem}
%For every \srt\ $T$ there exists an equivalent \sremo\ $e$ ,i.e., a \sremo\ such that $\mathit{Lang}(e) = \mathit{Lang}(T)$ and $\mathit{Match}(e,S)=\mathit{Match}(T,S)$ for every string $S$.
%\end{theorem}
%\begin{proof}
%The proof is similar to the corresponding proof in \cite{DBLP:journals/corr/abs-2110-04032}.
%\end{proof}

%\subsubsection{Closure properties of SRT}

We now study the closure properties of \srt\ under union, concatenation and Kleene-star.
We give the definition for closure under these operations:
\begin{definition}[Closure of \srt]
\label{definition:closure}
We say that \srt\ are closed under: 
\begin{itemize}
	\item union\index{union} if, for every \srt\ $T_{1}$ and $T_{2}$,
	there exists a \srt\ $T$ such that $\mathit{Match(T,S)} = \mathit{Match}(T_{1},S) \cup \mathit{Match}(T_{2},S)$, i.e., $M$ is a match of $T$ iff it is a match of $T_{1}$ or $T_{2}$.
%	\item intersection\index{intersection} if, for every \srt\ $T_{1}$ and $T_{2}$,
%	there exists a \srt\ $T$ such that $\mathit{Match(T,S)} = \mathit{Match}(T_{1},S) \cap \mathit{Match}(T_{2},S)$, i.e., $M$ is a match of $T$ iff it is a match of $T_{1}$ and $T_{2}$.
	\item concatenation\index{concatenation} if, for every \srt\ $T_{1}$ and $T_{2}$ and strings $S_{1}$, $S_{2}$,
	there exists a \srt\ $T$ such that $\mathit{Match}(T,S) = \mathit{Match}(T_{1},S_{1}) \cdot \mathit{Match}(T_{2},S_{2})$, where $S = S_{1} \cdot S_{2}$, i.e., $M$ is a match of $T$ iff $M_{1}$ is a match of $T_{1}$, $M_{2}$ is a match of $T_{2}$ and $M$ is the concatenation of $M_{1}$ and $M_{2}$ (i.e., $M = M_{1} \cup M_{2}$ and $min(M_{2}) > max(M_{1})$).
	\item Kleene-star\index{Kleene-star} if, for every \srt\ $T$ and string $S$,
	there exists a \srt\ $T_{*}$ such that $\mathit{Match}(T_{*},S) = \{M : M = M_{1} \cdot M_{2} \cdots M_{n}, M_{i} = \mathit{Match}(T,S_{i}), S = S_{1} \cdot S_{2} \cdots S_{n} \}$, i.e., $M$ is  a match of $T_{*}$ iff each $M_{i}$ is a match of $T$ and $M$ is the concatenation of all $M$.
%	\item complement\index{complement} if, for every \srt\ $T$,
%	there exists a \srt\ $T_{c}$ such that for every string $S$ it holds that $M \in \mathit{Match}(T,S) \Leftrightarrow M \notin \mathit{Match}(T_{c},S)$.
\end{itemize}
$\blacktriangleleft$
\end{definition}

We thus have the following for union, concatenation and Kleene-star:
\begin{theorem}
\label{theorem:closure}
\srt\ are closed under union, concatenation and Kleene-star.
\end{theorem}
\begin{proof}
See Appendix \ref{sec:proof:closure}.
\end{proof}

\srt\ can thus be constructed from the three basic operators in a compositional manner,
providing substantial flexibility and expressive power for CER applications.

\subsection{Streaming symbolic register transducers}
\label{section:streaming}

We have thus far described how \sremo\ and \srt\ can be applied to bounded strings that are known in their totality before recognition.
A string is given to a \srt\ and an answer is expected about whether the whole string belongs to the automaton's language or not along with any matches detected.
However, 
in CER we are required to handle continuously updated streams of events and detect instances of \sremo\ satisfaction as soon as they appear in a stream. 
For example, 
the automaton of the classical regular expression $a \cdot b$ would accept only the string $a,b$.
In a streaming setting, 
we would like the automaton to report a match every time this string appears in a stream.
For the stream $a,b,c,a,b,c$, 
we would thus expect two matches to be reported,
one after the second symbol and one after the fifth
(assuming that we are interested only in contiguous matches).

Slight modifications are required so that \sremo\ and \srt\ may work in a streaming setting
(the discussion in this section develops along the lines presented in our previous work \cite{DBLP:journals/vldb/AlevizosAP22}, with the difference that here we are concerned with symbolic automata with memory and output).
First, 
we need to make sure that the automaton can start its recognition after every new element.
If we have a classical regular expression $R$,
we can achieve this by applying on the stream the expression $\Sigma^{*} \cdot R$,
where $\Sigma$ is the automaton's (classical) alphabet.
For example,
if we apply $R := \{a,b,c\}^{*} \cdot (a \cdot b)$ on the stream $a,b,c,a,b,c$,
the corresponding automaton would indeed reach its final state after reading the second and the fifth symbols.
In our case, 
events come in the form of tuples with both numerical and categorical values. 
Using database systems terminology we can speak of tuples from relations of a database schema \cite{DBLP:conf/icdt/GrezRU19}.
These tuples constitute the universe $\mathcal{U}$ of a $\mathcal{V}$-structure $\mathcal{M}$.
A stream $S$ then has the form of an infinite sequence $S=t_{1},t_{2},\cdots$, where $t_{i} \in \mathcal{U}$.
Our goal is 
\begin{itemize}
	\item first, to report the indices $i$ at which a complex event is detected;
	\item second, to report the indices of the simple events from which a complex event is composed;
	\item while taking into account the fact that, at a given index $i$, multiple complex events may be detected.
\end{itemize}

More precisely,
if $S_{1..k}=\cdots,t_{k-1},t_{k}$ is the prefix of $S$ up to the index $k$,
we say that a \sremo\ $e$ is detected at $k$ iff there exists a suffix $S_{m..k}$ of $S_{1..k}$ such that $S_{m..k} \in \mathit{Lang}(e)$.
Additionally, 
the streaming matches detected at $k$ are defined as 
$\mathit{Match}_{stream}(e,S) = \{M: M \in \mathit{Match}(e,S_{m..k})\ \forall\ m,  1 \leq m \leq k  \}$

In order to detect complex events of a \sremo\ $e$ on a stream, 
we use a streaming version of \sremo\ and \srt.

\begin{definition}[Streaming \sremo\ and \srt]
If $e$ is a \sremo, 
then $e_{s}= \top^{*} \cdot e$ is called the streaming \sremo\ (\ssremo) corresponding to $e$.
A \srt\ $T_{e_{s}}$ constructed from $e_{s}$ is called a streaming \srt\ (\ssrt) corresponding to $e$.
$\blacktriangleleft$
\end{definition}

Using $e_{s} = \top^{*} \cdot e$ we can detect complex events of $e$ while reading a stream $S$,
since a stream segment $S_{m..k}$ belongs to the language of $e$ iff the prefix $S_{1..k}$ belongs to the language of $e_{s}$.
The prefix $\top^{*}$ lets us skip any number of events from the stream and start recognition at any index $m, 1 \leq m \leq k$.

%Note that \ssremo\ and \ssrt\ are just special cases of \sremo\ and \srt\ respectively.
%Therefore, 
%every result that holds for \sremo\ and \srt\ also holds for \ssremo\ and \ssrt\ as well.

%% file: extensions.tex
\section{Closure properties and selection strategies}

Thus far we have described a basic set of operators with which we can define complex event patterns and their corresponding computational model.
We have shown that our framework,
with these basic operators,
has unambiguous, compositional semantics.
Contrary to previous CER systems,
it does not impose ad hoc restrictions on the use of the operators,
which may be used in a fully compositional manner.
Besides concatenation/sequence, union/conjunction and Kleene-star/iteration,
CER systems make extensive use of other operators as well and even constructs which are external to the language itself. 
In this Section,
we focus on the issue of how and if our proposed framework can accommodate these extra operators and constructs. 
We specifically discuss the following aspects of CER which are very common in the literature,
but have been excluded from our presentation thus far:
the operators of intersection/conjunction and complement/negation, 
the possibility of using deterministic automata for CER,
the use of windows
and the semantics of selection strategies,

\subsection{Intersection and complement}

We first study the closure properties of \srt\ under intersection and complement,
two popular operators in CER.

The formal definition of closure under intersection and complement is as follows:
\begin{definition}[Closure of \srt (intersection, complement)]
\label{definition:closure_intersection_complement}
We say that \srt\ are closed under: 
\begin{itemize}
	\item intersection\index{intersection} if, for every \srt\ $T_{1}$ and $T_{2}$,
	there exists a \srt\ $T$ such that $\mathit{Match(T,S)} = \mathit{Match}(T_{1},S) \cap \mathit{Match}(T_{2},S)$, i.e., $M$ is a match of $T$ iff it is a match of $T_{1}$ and $T_{2}$.
	\item complement\index{complement} if, for every \srt\ $T$,
	there exists a \srt\ $T_{c}$ such that for every string $S$ it holds that $M \in \mathit{Match}(T,S) \Leftrightarrow M \notin \mathit{Match}(T_{c},S)$.
\end{itemize}
$\blacktriangleleft$
\end{definition}

With regards to intersection,
we can prove the following:
\begin{theorem}
\label{theorem:intersection}
\srt\ are closed under intersection.
\end{theorem}
\begin{proof}
See Appendix \ref{sec:proof:closure}.
\end{proof}
Note that intersection was not defined as an operator of \sremo\ in Definition \ref{definition:sremo}.
Theorem \ref{theorem:intersection} indicates that we can introduce such an operator without any difficulties.
It is important to distinguish intersection from another operator in CER,
which is also often called conjunction and whose intended semantics is that a sequence of events must occur,
regardless of their temporal order. 
This conjunction operator does not require any special treatment,
as it can be readily expressed in \sremo\ by combining the already available operators of sequence and disjunction.
For example,
if we use $\ast$ to denote that type of conjunction,
then we could write
$\ast(\phi_{1},\phi_{2}) := (\phi_{1} \cdot \phi_{2}) + (\phi_{2} \cdot \phi_{1})$.

On the other hand,
as is the case for register automata \cite{DBLP:journals/tcs/KaminskiF94},
\srt\ are not closed under complement:
\begin{theorem}
\label{theorem:complement}
\srt\ are not closed under complement.
\end{theorem}
\begin{proof}
See Appendix \ref{sec:proof:complement}.
\end{proof}
This result could pose difficulties for handling \emph{negation},
i.e.,
the ability to state that a sub-pattern should not happen for the whole pattern to be detected.
There is a subtle difference between negation as a regular expression operator and negation as a logical operator allowed in conditions (in Definition \ref{definition:condition}).
Logical negation always requires that an event occurs and that it does not satisfy the negated condition.
Regular negation may be ``satisfied'' even in the absence of any events and this is the way it is mostly used in CER patterns. 

However,
we can (partially) overcome the negative results about negation by using windows in \sremo\ and \srt,
i.e.,
by limiting the length of strings accepted by \sremo\ and \srt.
The general idea is that windows allow us to determinize \srt.
With a deterministic \srt\ at hand,
we can easily construct its complement.
The downside is that we lose the ability to mark events that correspond to negated expressions.
Thus, we now study the determinizability of \srt.

\subsubsection{Determinization of SRT}

In CER,
it is typically the case that non-deterministic automata are employed because they can fully enumerate all the detected matches,
i.e., 
report all input events comprising a match.
We also use non-deterministic \srt\ as a computational model for CER because they can enumerate all the detected matches, 
i.e., report all input events comprising a match.
Recall that complex events (or full matches) are defined as sets (of indices) of simple events.
Non-deterministic \srt\ have the ability to create multiple runs as they consume a stream of events.
Each run can mark different input events.
Each run that reaches a final state can then report all the input events that it has marked.
Thus, all complex events can be fully reported.
However, 
deterministic automata are critical in certain applications,
as in Complex Event Forecasting \cite{DBLP:journals/vldb/AlevizosAP22},
where the goal is to forecast whether a complex event is expected to occur,
without necessarily being interested in a complete enumeration.
For this reason,
we also study whether \srt\ are determinizable.

We can show that \srt\ are not closed under determinization,
a result which might seem discouraging.
We first provide the definition for deterministic \srt.
Informally, 
a \srt\ is said to be deterministic if, 
at any time, 
with the same input element, 
it can follow no more than one transition.
The formal definition is as follows:
\begin{definition}[Deterministic \srt\ (\dsrt)]
A \srt\ $T$ with $k$ registers $\{r_{1}, \cdots, r_{k}\}$ over a $\mathcal{V}$-structure $\mathcal{M}$ is deterministic if, 
for all transitions $q,\phi_{1} \uparrow o_{1} \downarrow W_{1} \rightarrow q_{1} \in T.\Delta $ and $q,\phi_{2} \uparrow o_{2} \downarrow W_{2} \rightarrow q_{2} \in T.\Delta$,
if $q_{1} \neq q_{2}$ then,
for all $u \in \mathcal{M}.\mathcal{U}$ and $v \in F(r_{1}, \cdots, r_{k})$,
$(u,v) \models \phi_{1}$ and  $(u,v) \models \phi_{2}$ cannot both hold,
i.e.,
\begin{itemize}
	\item Either $(u,v) \models \phi_{1}$ and $(u,v) \nvDash \phi_{2}$
	\item or $(u,v) \nvDash \phi_{1}$ and $(u,v) \models \phi_{2}$
	\item or $(u,v) \nvDash \phi_{1}$ and $(u,v) \nvDash \phi_{2}$.
\end{itemize}
$\blacktriangleleft$
\end{definition}
In other words,
from all the outgoing transitions from a given state $q$ at most one of them can be triggered on any element $u$ and valuation/register contents $v$.
By definition,
for a deterministic \srt,
at most one run may exist for every string/stream.

%We say that a \srt\ $T$ is output–agnostic determinizable if there exists an output–agnostic \dsrt\ $T_{D}$ such that, there exists an accepting run of $T$ over a string $S$ iff there exists an accepting run of $T_{D}$ over $S$ for every $S$.

We say that a \srt\ $T$ is determinizable if there exists a \dsrt\ $T_{D}$ such that $\mathit{Match}(T,S)=\mathit{Match}(T_{D},S)$ for every string $S$.
This is a strong notion of equivalence.
By definition, 
a deterministic \srt\ can have at most one run and thus at most one match for any string. 
Thus, equivalence based on the matches would be hard to achieve,
since non-deterministic \srt\ typically have multiple runs,
each tracking a (candidate) match.
Another notion of equivalence which is more relaxed can be obtained by requiring that the languages of $T$ and $T_{D}$ are the same, 
effectively ignoring the output of the transitions. 
In terms of CER,
ignoring transition outputs would still allow us to detect complex events,
in the sense that we could report at every timepoint whether at least one match has been fully completed.
On the other hand, 
we would not be able to say neither whether more than one such matches occurred nor report the simple events comprising a complex one. 
\begin{definition}[Output-agnostic \srt]
A \srt\ $T$ is output-agnostic determinizable if there exists a deterministic \srt\ $T_{D}$ such that $\mathit{Lang}(T)=\mathit{Lang}(T_{D})$.
\end{definition} 

Even with this more relaxed requirement,
it is not always possible to determinize \srt:
\begin{theorem}
\label{theorem:determinization}
Not every \srt\ is output-agnostic determinizable.
\end{theorem}
\begin{proof}
%See Appendix \ref{sec:proof:determinization}.
The proof is by a counter example.
Let $T$ denote the \srt\ of Figure \ref{fig:determinization_example_main}.
It detects events of type $B$ followed by events of type $S$ with the same identifier.
Between the two events, $B$ and $S$, any other events may also occur,
due to the presence of a self-loop on state $q_{1}$ with the \true\ condition.
This self-loop is what makes this \srt\ non-deterministic.
Whenever the \srt\ is in state $q_{1}$ and an event of type $S$ arrives with the same identifier as the stored $B$ event,
the automaton has two options.
Either move to the final state $q_{f}$ or remain in state $q_{1}$.
The self-loop on the start state $q_{s}$ also makes the \srt\ non-deterministic.
$T$ thus accepts strings $S$ that contain a $B$ followed by a $S$,
whose identifiers are equal,
regardless of the length of $S$. 
Any number of irrelevant events may precede the fist $B$ event.

Assume there exist a deterministic \srt\ $T_{d}$ with $k$ registers which is equivalent to $T$.
Let 
\begin{equation*}
S = (B, 1) (S, 2) 
\end{equation*}
be a string given to $T_{d}$.
After reading $S_{1}=(B,1)$,
$A_{d}$ must store it in a register $r_{1}$ in order to be able to compare it when $(S,2)$ arrives.
Let 
\begin{equation*}
S' = (B, 1) (B, 3) (S, 2) 
\end{equation*}
After reading $S_{1}'=(B,1)$,
$T_{d}$ must store it in the register $r_{1}$,
since $T_{d}$ is deterministic and follows a single run.
Thus, it must have the exact same behavior after reading $S_{1}$ and $S_{1}'$.
But we must also store $S_{2}'=(S,3)$ after reading it.
Additionally,
$S_{2}'$ must be stored in a different register $r_{2}$.
We cannot overwrite $r_{1}$.
If we did this and $S_{1}'$ were $(B,2)$,
then we would not be able to match $(B,2)$ to $S_{3}'=(S,2)$ and $S'=(B,2)(B,3)(S,2)$ would not be accepted.
Now, let
\begin{equation*}
S'' = \underbrace{(B, \cdots) (B, \cdots) \cdots (B, \cdots)}_{k+1 \text{ elements}}  (S, 2) 
\end{equation*}
With a similar reasoning,
all of the first $k+1$ elements of $S''$ must be stored after reading them.
But this is a contradiction,
as $T_{d}$ can store at most $k$ different elements.
Therefore, there does not exist a deterministic \srt\ which is equivalent to $T$.

\begin{figure}[t]
\begin{centering}
\includegraphics[width=0.75\linewidth]{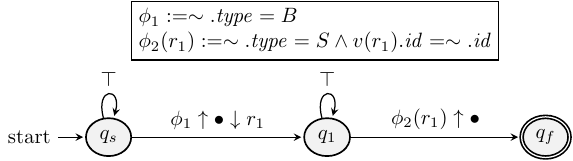}
\caption{Example of a non-deterministic \srt.}
\label{fig:determinization_example_main}
\end{centering}
\end{figure}
\end{proof}

This result could probably be generalized to state that \sremo\ cannot be captured by any deterministic automata with finite memory,
at least if we retain the usual notions of determinism and memory.
Determinism, 
in the sense that there can be at most one run of the automaton for every stream / string.
Memory, 
in the sense that it can store a finite number of the input elements from the stream / string.
We make this remark,
because one can imagine finite memory structures that do not store elements.
For example,
a memory slot could store finite mathematical structures that could act as generators of a 
(possibly infinite) stream of past elements.
Automata themselves are a typical case of a finite structure that can generate infinite sequences.
Automata that act as recognizers and can store other automata, acting as generators,
is thus something not inconceivable.
Investigating such automata is, however, beyond the scope of this thesis.

\subsubsection{Windowed SREMO/SRT}

We can overcome the negative results about determinization by using windows in \sremo\ and \srt.
We show that there exists a sub-class of \sremo\ for which a translation to  output-agnostic deterministic \srt\ is indeed possible.
This is achieved if we apply a windowing operator and limit the length of strings accepted by \sremo\ and \srt.
In general,
CER systems are not expected to remember every past event of a stream
and produce matches involving events that are very distant.
On the contrary,
it is usually the case that CER patterns include an operator that limits the search space of input events,
through the notion of windowing.
This observation motivates the introduction of windowing in \sremo.
\begin{definition}[Windowed \sremo]
\label{definition:windowed_srem}
Let $e$ be a \sremo\ over a $\mathcal{V}$-structure $\mathcal{M}$ and a set of register variables $R = \{r_{1}, \cdots, r_{k}\}$, 
$S$ a string constructed from elements of the universe of $\mathcal{M}$ 
and $v,v' \in F(r_{1}, \cdots, r_{k})$.
A windowed \sremo\ (\wsremo) is an expression of the form $e_{w} := e^{[1..w]}$,
where $w \in \mathbb{N}_{1}$. 
We define the relation $(e_{w},S,M,v) \vdash v'$ as equivalent to:
$(e,S,M,v) \vdash v'$ and $(\mathit{max}(M) - \mathit{min}(M) + 1) \leq w$.
$\blacktriangleleft$
\end{definition}

Essentially, 
the only difference to regular \sremo\ is a slight change in the definition of the semantics (see Definition \ref{definition:sremo_semantics}).
For windowed \sremo,
the additional requirement is that the ``interval'' from the smallest match index ($\mathit{min}(M)$) to the largest ($\mathit{max}(M)$) does not exceed the given window threshold.
\begin{example}
If we apply a window $w=4$ on \sremo\ \eqref{srem:b_seq_s_filter_eq_id},
then $M=\{1,4\}$ will still be a match,
since $4 - 1 + 1 \leq 4$ obviously holds.
$M=\{1,5\}$,
on the other hand,
is no longer a match.
\end{example}

The windowing operator does not add any expressive power to \sremo.
We could use the index of an event in the stream as an event attribute
and then add binary conditions in an expression which ensure that the difference between the index of the last event read and the first is no greater that $w$.
%This could possibly require complex re-writing rules of an initial expression.
It is more convenient, however, to have an explicit operator for windowing.

In order to derive a deterministic \srt,
we can first construct a so-called ``unrolled \srt''\ from a windowed expression,
i.e., a \srt\ without any loops where each state may be visited at most once.
The window allows us to do this,
effectively removing any unbounded iterations.
We can then apply a standard determinization algorithm to the unrolled \srt.

\begin{theorem}
\label{theorem:wsremo2dsrt}
For every windowed \sremo\ there exists an equivalent output-agnostic deterministic \srt.
\end{theorem}
\begin{proof}
See Appendix \ref{sec:proof:wsremo2dsrt}.
\end{proof}

\begin{example}
As an example,
consider Figure \ref{fig:unroll_determinization}.
Figure \ref{fig:unrolled} shows the unrolled version of the \srt\ of Figure \ref{fig:determinization_example_main} for two different window values, 2 and 3.
All cycles have been eliminated,
with the overhead of one extra register being added.
We also do not show the output,
since we now focus on output-agnostic \srt.
Note that the black \srt\ (for $w=2$) is already deterministic.
Figure \ref{fig:determinization_w3} shows (part of) the deterministic \srt\ for $w=3$.
Due to space limitations,
we show only part of the complete automaton.
The idea is clear though.
We start with the initial state ($q_{s}$) and create its mutually exclusive transitions.
For example,
if $(\phi_{1} \wedge \top)$ (i.e., $\phi_{1}$) evaluates to \true,
then we move from $q_{s}$ to both $q_{s,s}$ and $q_{s,1}$.
We thus create a relevant hyper-state $\{q_{s,s}, q_{s,1}\}$ and connect it to the start state.
We repeat this process until we have exhausted all possible states.
\end{example}

\begin{figure}[t]
\centering
\begin{subfigure}[t]{0.9\textwidth}
	\includegraphics[width=0.99\textwidth]{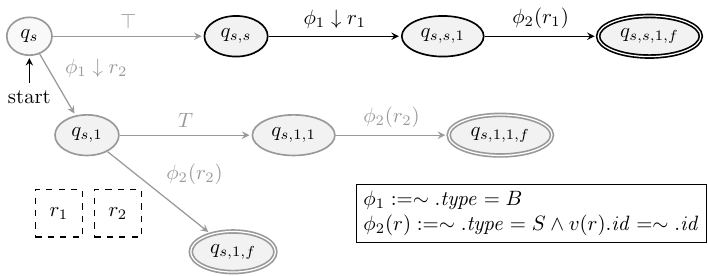}
	\caption{\srt\ after unrolling cycles, for $w = 3$ (whole \srt, black and light gray states) and $w = 2$ (top 3 states in black).}
	\label{fig:unrolled}
\end{subfigure}\\
\begin{subfigure}[t]{0.99\textwidth}
	\includegraphics[width=0.99\textwidth]{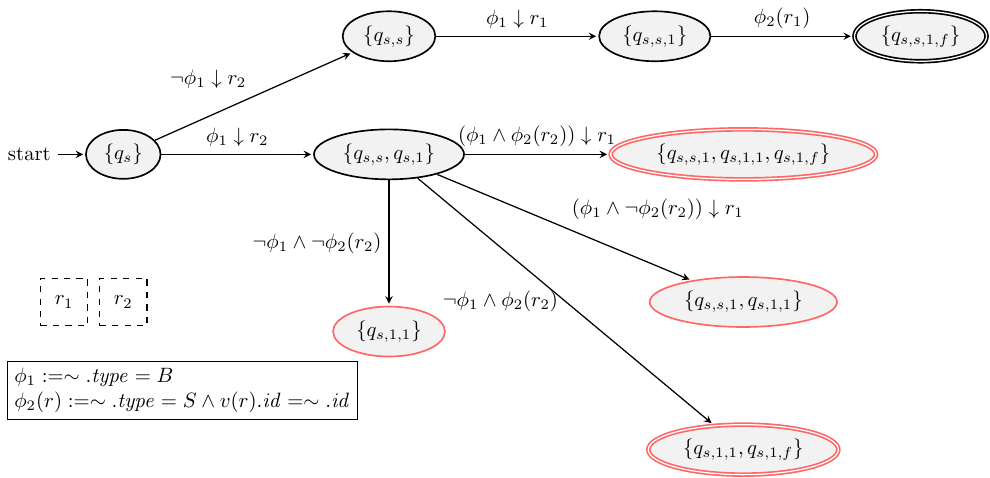}
	\caption{Output–agnostic deterministic \srt, for $w = 3$. Only part of it is shown. Red states can be further expanded.}
	\label{fig:determinization_w3}
\end{subfigure}
\caption{Constructing a deterministic \srt\ from the \srt\ of Figure \ref{fig:determinization_example_main}.}
\label{fig:unroll_determinization}
\end{figure}

Deterministic \srt\ cannot thus be used for patterns without windows. 
By using windows,
we can recover the property of \srt\ to be (output-agnostic) determinizable.
The limitation at this point is that determinization is possible only at the language level
(in Complex Event Forecasting, 
this does not constitute an issue).

If we restrict ourselves to windowed \sremo\ and to output-agnostic deterministic \srt,
then we can prove that such \srt\ are closed under complement.
A standard technique for creating the complement of an automaton is to create its deterministic equivalent and then flip its final states to non-final and vice versa.
Thus,
we may now prove,
as a corollary,
that windowed \srt\ are also closed under complement,
when transition outputs are ignored:
\begin{corollary}
\label{corollary:wsra_complement}
Output-agnostic \srt\ compiled from windowed \sremo\ are closed under complement.
\end{corollary}
\begin{proof}
See Appendix \ref{sec:proof:wsra_complement}.
\end{proof}

This result is important because it allows us to extend (windowed) \sremo\ so as to also include a negation operator. 
Although in theory the result about closure under complement holds only when outputs are ignored,
in practice it could be useful even when we are indeed interested in the output of transitions and in marking some elements as belonging to a match.
This could be the case when we have a \sremo\ containing a negation operator. 
Typically,
we are not interested to mark any elements that are negated.
Negation often implies that outputs should be ignored,
especially when we want to imply absence of simple events,
in which case there is no sense in marking those absent events.
For example,
consider the expression
$((\phi_{1}(\sim) \uparrow \bullet \downarrow r_{1})\ \cdot\ !(\phi_{2}) \cdot (\phi_{3}(\sim,r_{1}) \uparrow \bullet))^{[1..w]}$,
where $!$ stands for negation.
In this case,
we are only interested to mark the elements matching the first and third sub-expressions but not the second, negated one.
We could construct a sub-automaton for the complement of $\phi_{2}$,
ignoring any outputs.
At the same time,
we could construct the automatons for the first and third sub-expressions as usual,
with their outputs.
By concatenating these three sub-automata,
we would be able to properly mark the elements that we are interested in,
despite the fact that the expression contains negation.

\subsection{Selection strategies}

CER patterns are usually characterized by their so-called selection strategy \cite{DBLP:journals/vldb/GiatrakosAADG20}.
This strategy determines whether the input events in a match should occur contiguously in a stream (the standard interpretation of regular expressions) or intermittently,
with other, irrelevant events happening between the relevant ones.
\strictcont, \skipany\ and \skipnext\ are the three common such strategies.
\strictcont\ requires all simple events to occur contiguously. 
\skipany\ allows any irrelevant events to occur between the relevant ones.
This behavior may typically be modeled in automata by introducing self-loops on states with $\top$ as their condition and $\otimes$ as their output
(e.g., see Figure \ref{fig:example1}).
On the other hand,
\skipnext\ also allows multiple irrelevant events between two relevant events, 
say S and B,
except for B itself. 

Given the properties of \sremo\ and \srt,
the question is whether and which selection strategies may be accommodated,
besides \strictcont,
which is the standard interpretation of regular expressions. 
%The answer depends on the kind of \sremo\ used.
%For a windowed \sremo\ $e$,
%both \skipany\ and \skipnext\ may be applied to $e$. 
%In fact, 
We can show that selection strategies may be applied as operators,
through certain rewriting rules.
This then implies that multiple and even nested strategies may be used in a pattern.

We define \skipany\ and \skipnext\ as extra operators (and not extra-pattern constructs) in the following manner:

\begin{definition}[\skipany]
\label{definition:skipany}
If $e_{1},e_{2},\cdots,e_{n}$ are \sremo,
then $e_{\mathit{any}} := \circlearrowleft(e_{1}, e_{2}, \cdots , e_{n})$ is a \sremo\, 
with $\circlearrowleft$ denoting the \skipany\ selection strategy and
\begin{equation*}
e_{\mathit{any}} := e_{1} \cdot (\top \uparrow \otimes)^{*} \cdot e_{2} \cdot (\top \uparrow \otimes)^{*} \cdots (\top \uparrow \otimes)^{*} \cdot e_{n}
\end{equation*}
\end{definition}

\begin{definition}[\skipnext]
\label{definition:skipnext}
If $e_{1},e_{2},\cdots,e_{n}$ are windowed \sremo,
then $e_{\mathit{next}} := @(e_{1}, e_{2}, \cdots, e_{n})$ is a \sremo, 
with $@$ denoting the \skipnext\ selection strategy and
\begin{equation*}
e_{\mathit{next}} :=  e_{1} \cdot  (!e_{2})^{*} \cdot e_{2} \cdots  (!e_{n})^{*}  \cdot e_{n}
\end{equation*}
%\begin{equation*}
%e_{\mathit{next}} :=  e_{1} \cdot !((\top \uparrow \otimes)^{*} \cdot e_{2} \cdot (\top \uparrow \otimes)^{*}) \cdot e_{2} \cdots !((\top \uparrow \otimes)^{*} \cdot e_{n} \cdot (\top \uparrow \otimes)^{*}) \cdot e_{n}
%\end{equation*}
$!$ denotes the regular operator of negation/complement.
\end{definition}

Similar definitions may be provided for selection strategies applied on iteration (Kleene-star or Kleene-plus),
since iteration is also essentially sequential.
With respect to disjunction,
we make the assumption that a selection strategy applied on a disjunction operator has no effect.
The strategy is not applied to any of the sub-expressions inside the top-level disjunction expression.
If the user needs to apply a selection strategy to any sub-expression,
he may do so by applying the relevant strategy operator on this specific disjunct / sub-expression.

The intuition behind the definition of \skipany\ is that we would like to be able to skip any events occurring between instances of $e_{i}$ and $e_{i+1}$.
We can actually achieve exactly this behavior by injecting between every pair of $e_{i}$ and $e_{i+1}$ the expression $(\top \uparrow \otimes)^{*}$.
Since $\top$ evaluates to \true\ for every element,
this means that $(\top \uparrow \otimes)^{*}$ allows us to skip any number of elements.
These elements are skipped because that output of the expression is $\otimes$.

For \skipnext,
our goal is for an expression to exhibit a ``greedy'' behavior, 
i.e., after an instance of $e_{i}$ we want to accept the immediately next instance of $e_{i+1}$ and afterwards ignore any other instances of $e_{i+1}$.
The above definition for \skipnext\ satisfies this constraint.
For example, consider the sub-expression $e_{1} \cdot (!e_{2})^{*} \cdot e_{2}$.
Between $e_{1}$ and $e_{2}$ we have injected the sub-expression $(!e_{2})^{*}$.
This sub-expression ensures that,
between instances of $e_{1}$ and $e_{2}$,
no other instance of $e_{2}$ may occur. 
Thus, if after an instance of $e_{1}$,
we encounter multiple instances of $e_{2}$,
only the first one will be accepted.

It should be noted though that \skipnext,
in its most general form presented above, 
may be used only with windowed expressions.
The reason is that it relies on negation, 
which, in turn, relies on determinization.
A possible issue at this point is that windowed \sremo\ may be converted only to output-agnostic deterministic \srt.
Thus, the deterministic sub-automata corresponding to the negated sub-expressions in \skipnext\ 
(e.g., $(!e_{2})^{*}$)
do not have the ability to mark elements of the input string as relevant or irrelevant.
However, this is not a serious limitation in this case.
Since the negated sub-expressions are injected with the aim of skipping irrelevant events,
we can simply force all transitions of these sub-automata to output $\otimes$,
after we have constructed the automata.
If, however, the sub-expressions $e_{i}$ are terminal conditions $\phi_{i}$
(which is the typical case),
then regular negation can be replaced with logical negation ($\neg$, as per Definition \ref{definition:condition}) and \skipnext\ may then be used even in \sremo\ without windows.

As far as the implementation of selection strategies is concerned,
we have included at the moment \skipany\ in our our system,
but not yet \skipnext.
We have not had a need for the latter yet,
which also depends on negation (also not implemented currently).
We focused on \skipany\ which is the most demanding,
both in terms of time and memory.
Note that no special treatment is reserved for the selection strategies from an implementation point of view.
A \sremo\ with a selection strategy applied to it treats this strategy as another operator.
The \sremo\ is first re-written according to Definitions \ref{definition:skipany} and \ref{definition:skipnext}.
The re-written \sremo\ is then compiled into a \srt, 
just like every other \sremo.
This \srt\ is then used for recognition,
without any strategy-specific optimizations. 
Even in the absence of optimizations,
we show that our system can handle even \skipany, 
the most relaxed and demanding strategy,
due to its lightweight representation of runs.

\subsection{Summary}

In summary,
we can state the following.
Intersection is an operator that can be supported by our framework without any constraints.
Negation and determinization can also be supported,
but only for windowed expressions and with the understanding that negated events cannot be marked as being part of a match.
With respect to selection strategies,
\skipany\ can be accommodated without any constraints.
\skipnext\ is also available,
but only for windowed expressions.
When applied to simple conditions,
it is available even for expressions without windows.
%Aggregates could also be supported in theory,
%but an efficient implementation would require optimizations on the basic computational model of \srt.
%Finally, 
%hierarchies are not currently supported.
%The reason is that \srt\ do not have the ability to process events with history or duration,
%an essential feature for a proper treatment of hierarchies.

%% file: impl.tex
\section{Implementation and Complexity}

In the theory of formal languages it is customary to present complexity results for various decision problems, most commonly for the problem of non-emptiness (whether an expression or automaton accepts at least one string),
that of membership (deciding whether a given string belongs to the language of an expression/automaton) and that of universality (deciding whether a given expression/automaton accepts every possible string).
We briefly discuss here these problems for the case of \sremo\ and \srt.

The complexity of these problems for \sremo\ and \srt\ depends heavily on the nature of the conditions used as terminal expressions in \sremo\ and as transition guards in \srt.
This, in turn, depends on the complexity of deciding whether a given element from the universe $\mathcal{U}$ of a $\mathcal{V}$-structure $\mathcal{M}$ belongs to a relation $R$ from $\mathcal{M}$.
Since we have not imposed until now any restrictions on such relations,
the complexity of the aforementioned decision problems can be ``arbitrarily'' high and thus we cannot provide specific bounds.
If, for example, the problem of evaluating a relation $R$ is NP-complete and this relation is used in a \sremo/\srt\ condition,
this then implies that the problem of membership immediately becomes at least NP-complete.
In fact, 
if the problem of deciding whether an element from $\mathcal{U}$ belongs to a relation $R$ is undecidable,
then the membership problem becomes also undecidable.

We can, however, provide some rough bounds by looking at the complexity of these problems for the case of register automata
(see \cite{DBLP:journals/jcss/LibkinTV15}).
Register automata are a special case of \srt,
where the only allowed relations are the binary relations of equality and inequality and the transitions do not generate any output.
We assume that these relations may be evaluated in constant time.
For the problem of universality,
we know that it is undecidable for register automata.
We can thus infer that it remains so for \srt\ as well.
On the other hand,
the problem of non-emptiness is decidable but PSPACE-complete.
The same problem for \srt\ is thus PSPACE-complete.
Finally,
the problem of membership is NP-complete.
Therefore,
it is also at least NP-complete for \srt.
Note that membership is the most important problem for the purposes of CER,
since in CER we continuously try to check whether a string (a suffix of the input stream) belongs to the language of a pattern's automaton.
In general,
if we assume that the problem of membership in all relations $R$ is decidable in constant time,
then the complexity of the decision problems for \srt\ coincides with that for register automata.

If we focus our attention even further on windowed \srt,
as is the case in CER,
then we can estimate more precisely the complexity of processing a single event from a stream.
This is the most important operation for CER.
A windowed \srt\ can first be determinized (offline) to obtain a \dsrt.
Assume that the resulting \dsrt\ $T$ has $r$ registers and $c$ conditions/minterms.
We also assume that evaluating a condition requires constant time and that accessing a register also takes constant time.
In the worst case,
after a new element/event arrives,
we need to evaluate all of the conditions/minterms on the $c$ outgoing transitions of the current state to determine which one of them is triggered.
We may also need to access all of the $r$ registers in order to evaluate the conditions.
Therefore, 
the complexity of updating the state of the \dsrt\ $T$ is $O(c+r)$
(assuming that each register is accessed only once and its contents are provided to every condition which references that register).

In the general case though,
non-deterministic \srt\ are used because they can report exactly the complex events that are detected at every timepoint $k$ and not just the fact that (at least one) such complex event has been detected.
%Recall that complex events (or full matches) are defined as sets (of indices) of simple events.
%Non-deterministic \srt\ have the ability to create multiple runs as they consume a stream of events.
%Each run can mark different input events.
%Each run that reaches a final state can then report all the input events that it has marked.
%Thus, all complex events can be fully reported.
The solution to the problem of estimating the runtime complexity of our CER engine when using non-deterministic \srt\ is thus not as straightforward as in the case of deterministic \srt.

For our implementation of \srt-based CER,
we have used Wayeb as a starting point.
Wayeb is a Complex Event Recognition and Forecasting engine,
based on symbolic automata
\cite{DBLP:conf/lpar/AlevizosAP18,DBLP:journals/vldb/AlevizosAP22}
\footnote{Available here: \url{https://github.com/ElAlev/Wayeb}.}.
We have extended Wayeb so that it can compile (windowed) \sremo\ into \srt\ and then use non-deterministic \srt\ for recognition.
We have implemented both the compiler from \sremo\ to \srt\ and the \srt\ as well.

\input{algorithms_run_handling}

The workflow of our engine is the following
(see Algorithm \ref{algorithm:run_handling}).
The user provides a pattern in the form of a windowed \sremo\
with a specific selection strategy and the engine compiles this pattern into a \srt\ $T$ (see Section \ref{section:sremo2srt}).
Subsequently, 
Wayeb creates a streaming version of this \srt\, $T_{s}$ (see Section \ref{section:streaming}).
This streaming \srt\ is then fed with a stream $S$ of simple events.
Initially, 
before any input event has been consumed,
the set of runs $\mathit{Run}(T_{s},S_{..0})$ is composed of a single run (see Definition \ref{definition:run}), 
$[1,T'.q_{s},\sharp]$.
$S_{..0}$ denotes the stream when no event has yet been processed.
The single run, 
$[1,T'.q_{s},\sharp]$,
points to the first event in the stream,
it is in its start state $q_{s}$ and its registers are empty ($\sharp$).
Wayeb then reads input events one by one and updates its set of runs after every new event.
At each timepoint $k$,
before reading the $k^{th}$ event $t_{k}$,
Wayeb maintains the set $\mathit{Run}(T_{s},S_{..k-1})$.
After processing $t_{k}$,
it produces $\mathit{Run}(T_{s},S_{..k})$.
This is achieved by evaluating $t_{k}$ against every $\varrho \in \mathit{Run}(T_{s},S_{..k-1})$.
Each run $\varrho=[1,q_{1},v_{1}] \overset{\delta_{1}}{\rightarrow} [2,q_{2},v_{2}] \overset{\delta_{2}}{\rightarrow} \cdots \overset{\delta_{k-1}}{\rightarrow} [k,q_{k},v_{k}]$ has to evaluate $t_{k}$ on all the outgoing transitions of state $q_{k}$.
If no transition is triggered,
this means that the \srt\ cannot move to another state and $\varrho$ is thus discarded and not included in $\mathit{Run}(T_{s},S_{..k})$.
If only one transition is triggered,
then $\varrho$ is updated,
becoming $\varrho=[1,q_{1},v_{1}] \overset{\delta_{1}}{\rightarrow} [2,q_{2},v_{2}] \overset{\delta_{2}}{\rightarrow} \cdots \overset{\delta_{k-1}}{\rightarrow} [k,q_{k},v_{k}] \overset{\delta_{k}}{\rightarrow} [k+1,q_{k+1},v_{k+1}]$,
with a new state $q_{k+1}$ and register contents $v_{k+1}$.
If $n$ transitions are triggered and thus $n$ next states are to be reached,
then $\varrho$ may be updated as usual for one of those next states.
For each of the other $n-1$ next states,
$\varrho$ is first cloned,
producing $n-1$ new runs $\varrho'$, $\varrho''$, etc.
Then each of these runs is updated with the new state and register contents
\begin{itemize}
	\item $\varrho'=[1,q_{1},v_{1}] \overset{\delta_{1}}{\rightarrow} [2,q_{2},v_{2}] \overset{\delta_{2}}{\rightarrow} \cdots \overset{\delta_{k-1}}{\rightarrow} [k,q_{k},v_{k}] \overset{\delta'_{k}}{\rightarrow} [k+1,q'_{k+1},v'_{k+1}]$
	\item $\varrho''=[1,q_{1},v_{1}] \overset{\delta_{1}}{\rightarrow} [2,q_{2},v_{2}] \overset{\delta_{2}}{\rightarrow} \cdots \overset{\delta_{k-1}}{\rightarrow} [k,q_{k},v_{k}] \overset{\delta''_{k}}{\rightarrow} [k+1,q''_{k+1},v''_{k+1}]$
	\item ...
\end{itemize}
The updated/new runs are added to the set of runs $\mathit{Run}(T_{s},S_{..k})$.
Accepting runs are the exception here.
If $q_{k+1} \in T_{s}.Q_{f}$ and $\delta_{k}.o = \bullet$ for some run $\varrho$,
then $\varrho$ reports all the input events that it has marked with $\bullet$ and is then ``killed'',
i.e., not added to  $\mathit{Run}(T_{s},S_{..k})$.
This process is repeated for the remaining runs of $\mathit{Run}(T_{s},S_{..k-1})$.

The cost of evaluating a single event $t_{k}$ depends on several factors.
It depends on $\lvert \mathit{Run}(T_{s},S_{..k-1}) \rvert$,
the number of active runs against which $t_{k}$ is to be evaluated.
It also depends on the number of outgoing transitions from the states of active runs as well as on the complexity of evaluating the predicates of transitions.
If we assume a constant cost for predicate evaluation $c_{p}$ and then bound the number of outgoing transitions to be at most $n_{p}$,
where $n_{p}$ is the number of predicates appearing in the initial \sremo\ (including the $\top$ predicate),
then the cost of evaluating $t_{k}$ against a run $\varrho$ is at most $n_{p} \cdot c_{p}$.
Therefore, the total cost of evaluating all runs is $\lvert \mathit{Run}(T_{s},S_{..k-1}) \rvert \cdot n_{p} \cdot c_{p}$.
In the worst case,
all outgoing transitions of all runs are triggered.
We will thus have to create $\lvert \mathit{Run}(T_{s},S_{..k-1}) \rvert \cdot (n_{p} - 1)$ new run clones and perform $\lvert \mathit{Run}(T_{s},S_{..k-1}) \rvert \cdot n_{p}$ run updates. 
If $c_{c}$ is the cost of run cloning and $c_{u}$ the cost of run updating,
then the total cost would be 
\begin{equation*}
\lvert \mathit{Run}(T_{s},S_{..k-1}) \rvert \cdot n_{p} \cdot c_{p} + \lvert \mathit{Run}(T_{s},S_{..k-1}) \rvert \cdot (n_{p} - 1) \cdot c_{c} + \lvert \mathit{Run}(T_{s},S_{..k-1}) \rvert \cdot n_{p} \cdot c_{u} 
\end{equation*}
or 
\begin{equation*}
\lvert \mathit{Run}(T_{s},S_{..k-1}) \rvert \cdot (n_{p} \cdot c_{p} + \cdot (n_{p} - 1) \cdot c_{c} +  n_{p} \cdot c_{u}) = \lvert \mathit{Run}(T_{s},S_{..k-1}) \rvert \cdot (n_{p} \cdot (c_{p} + c_{c} +  c_{u}) -c_{c})
\end{equation*}

The complexity depends highly on the number of active runs at every timepoint.
We can also estimate the runtime complexity on a ``per-window'' basis,
by attempting to calculate the total number of runs created for a window of input events.
Relevant results have been obtained in \cite{DBLP:conf/sigmod/ZhangDI14}.
For a sequential pattern (without disjunction or Kleene-star) under \strictcont\ and a window $w$,
the total number of created runs is $R \cdot w$,
where $R$ is the percentage of input events satisfying predicate $p$ of the outgoing transition from the a state.
Under \strictcont,
there is only one state where cloning may occur and this is the first state,
which has a self-loop with $\top$ and a transition to another state with predicate $p$.
This predicate will be satisfied $R \cdot w$ times.
If the average cost of handling a run is $c_{r}$
(including predicate evaluation, clone creation, etc.),
then the total cost is $R \cdot w \cdot c_{r}$.
Under \skipany,
the first state will create $R \cdot w$ clones,
the second $(R \cdot w)^{2}$, etc.
We thus have a geometric series and the total number of created runs will be $\frac{(R \cdot w)^{i+1} - 1 }{ (R \cdot w) - 1 }$,
where $i$ is the number of ``terminal'' sub-patterns in the original pattern.
If the pattern contains $j$ Kleene ``components''
(and thus the automaton $j$ states with self-loops),
then the total number of runs will be $\frac{(R \cdot w)^{i-j+1} - 1 }{ (R \cdot w) - 1 } \cdot 2^{j \cdot R \cdot w}$.
We see then that the worst-case cost becomes exponential in the size of the window and the number of Kleene-star operators.
%The number of Kleene operators is typically low.
%The window size, 
%on the other hand,
%must be at least as large as the length of a \sremo\
%(e.g., the number of terminal expressions in sequential patterns).
%If this is not the case,
%the \sremo\ will never be satisfied.
%Theoretically,
%there is no upper bound on the window.
%In practice, 
%however,
%in many applications,
%its size is limited.

Note that in the current version of our engine we have not performed any algorithmic optimizations.
Each run is internally in a minimalistic manner,
represented by 3-tuple,
holding the current state of the run, 
its register contents and a list with the indices of the simple events it has marked at every timepoint. 
Thus,
we do not need to explicitly represent each run as a separate class instance and we avoid the cost of cloning and maintaining run objects.
We have not implemented any postponing (\cite{DBLP:conf/sigmod/ZhangDI14}) or match-sharing (\cite{DBLP:journals/pvldb/BucchiGQRV22}) optimizations.
We have only implemented some simple code optimizations.
Our engine has been written in Scala.
Scala generally favors the use of structures such as Sets and Lists and the use of methods such as .map and .filter on such structures.
We have avoided the use of such structures and methods in critical parts of the code,
since they have proven to be sub-optimal.
Instead,
we opted for arrays and while loops with indices.
This C-like implementation proved to be substantially better in terms of performance.

%% file: algorithms_run_handling.tex
\begin{algorithm}
\caption{Running Wayeb with non-deterministic \srt.}
\label{algorithm:run_handling}
%\SetAlgoNoLine
\KwIn{\srt\ $T_{s}$, input event $t_{k}$, active runs $\mathit{Run}(T_{s},S_{..k-1})$}
\KwOut{Active runs $\mathit{Run}(T_{s},S_{..k})$, accepting runs $\mathit{Run_{f}}(T_{s},S_{..k})$}
$\mathit{Run_{f}}(T_{s},S_{..k}) \leftarrow \emptyset$\;
$\mathit{Run}(T_{s},S_{..k}) \leftarrow \emptyset$\;
\ForEach{$\varrho \in \mathit{Run}(T_{s},S_{..k-1})$}{
	$C \leftarrow \mathit{FindSuccessorConfigurations(\varrho,t_{k})}$\; \label{algorithm:run_handling:line:successor}
	\If{$\lvert C \rvert > 0$}{
		$c \leftarrow $ pick and remove element from $C$\;
		$\varrho_{new} \leftarrow \mathit{UpdateRun}(\varrho,c)$\;
		\uIf{$\mathit{IsAccepting}(\varrho_{new})$}{
			$\mathit{ReportMatch(\varrho_{new})}$\;
			$\mathit{Run_{f}}(T_{s},S_{..k}) \leftarrow \mathit{Run_{f}}(T_{s},S_{..k}) \cup \varrho_{new}$\;
		}
		\Else {
			$\mathit{Run}(T_{s},S_{..k}) \leftarrow \mathit{Run}(T_{s},S_{..k}) \cup \varrho_{new}$\;
		}
		\ForEach{$c \in C$}{
			$\varrho' \leftarrow \mathit{Clone}(\varrho)$\;
			$\varrho_{new} \leftarrow \mathit{UpdateRun}(\varrho',c)$\;
			\uIf{$\mathit{IsAccepting}(\varrho_{new})$}{
				$\mathit{ReportMatch(\varrho_{new})}$\;
				$\mathit{Run_{f}}(T_{s},S_{..k}) \leftarrow \mathit{Run_{f}}(T_{s},S_{..k}) \cup \varrho_{new}$\;
			}
			\Else {
				$\mathit{Run}(T_{s},S_{..k}) \leftarrow \mathit{Run}(T_{s},S_{..k}) \cup \varrho_{new}$\;
			}
		}
	}
}
return $\mathit{Run_{f}}(T_{s},S_{..k})$,$\mathit{Run}(T_{s},S_{..k})$\;
\end{algorithm}

%% file: exp.tex
\section{Experimental results}

We present experimental results by comparing Wayeb against other state-of-the-art CER systems.
Our goal is to test the systems with expressive, relational patterns,
i.e., with patterns which can relate multiple events
(which is the motivation for introducing symbolic register transducers).
For this reason,
we had to exclude systems without the ability to express relational patterns,
such as CORE and Wayeb.
For some other systems,
there is no publicly available implementation or the implementation is no longer maintained (e.g., CRS and Cayuga).
Yet some other systems (e.g., TESLA) suffer from low performance for certain classes of queries,
as mentioned in \cite{DBLP:journals/corr/abs-2111-04635}.

We also considered Flink's implementation of MATCH\_RECOGNIZE \cite{FlinkMR}.
However, though rich with various features,
it is limited in certain crucial respects.
For example,
iteration can only be applied to single events and not to subsequences.
Moreover,
we were not able to reproduce results obtained from other engines,
even with simple sequential patterns when applying the \skipany\ strategy.
Several matches were missing from the output. 
Nevertheless,
we attempted to run some experiments and measure Flink's throughput,
even when it failed to report all matches.
We discovered that its throughput was the lowest of all other engines and comparable to that of FlinkCEP,
Flink's CER engine.
This is not a surprising result, 
as Flink's implementation of MATCH\_RECOGNIZE is based on FlinkCEP.   
For these reasons,
we excluded MATCH\_RECOGNIZE from any further experiments.

Our comparison thus includes
SASE v1.0 \cite{SASE}, Esper v8.7.0 \cite{Esper} and FlinkCEP v1.16.1 \cite{FlinkCEP}.
All these engines are written in Java.
Wayeb is implemented in Scala 2.12.10.
All experiments were run on a 64-bit Linux machine with AMD EPYC 7543 × 126 processors and 400 GB of memory.
We used Java 1.8 for all systems.
All experiments for all systems were run as single-core applications,
without any attempt at parallelization/distribution in order to ensure a level comparison field 
(note that Esper and FlinkCEP support parallelization).
Wayeb is an open-source engine and the experiments presented here are reproducible
\footnote{Available here: \url{https://github.com/ElAlev/cer-srt}.}.

As a basis for our experiments, 
we used the benchmark suite presented in \cite{DBLP:journals/pvldb/BucchiGQRV22}
\footnote{\url{https://github.com/CORE-cer/CORE-experiments}.}.
The suite contains three datasets:
a) stock market data from a single day (224,473 input events);
b) plug measurements from smart homes (1,000,000 input events) and
c) taxi trips from the city of New York (585,762 input events).
For the stock market dataset,
each input event is a BUY or SELL event,
containing the name of the company,
the price of the stock,
the volume of the transaction and its timestamp.
For the smart homes dataset,
each input event is a LOAD event,
containing a load value in Watts,
a household id,
a plug id and a timestamp.
For the taxis dataset,
each input event is a TRIP event,
containing the datetime of the pickup and dropoff,
the zone of the pickup and dropoff,
the trip distance and duration,
the fare amount,
the tip amount,
the total amount, etc.

The suite allows to run the same pattern on multiple engines and with multiple windows.
As explained in the previous section,
the runtime complexity of a CER engine with a given pattern depends heavily on the size of the window and the number of Kleene operators which the pattern contains.
For this reason,
we have used the ability of the suite to run experiments on multiple windows.
We also focused on patterns with (multiple) Kleene operators.
For all patterns,
we fixed the selection policy to \skipany,
since this is the most demanding policy,
both in terms of time and space complexity.

Our results are presented incrementally as we increase the complexity of the tested patterns.
We start with sequential patterns where some of the simple events are related through constraints (Section \ref{section:exp:seq}).
We also study the effect of window size on such patterns (Section \ref{section:exp:seq:varwin}).
We then add Kleene operators on single events to these patterns (Section \ref{section:exp:kleene}).
We additionally test patterns with nested Kleene operators (Section \ref{section:exp:kleene_nested}).
Finally, we present results with patterns containing various, mixed operators (Section \ref{section:exp:other}).
Note that it is not possible to test all systems against all classes of patterns.
Some systems either do not support all classes or have ambiguous/different semantics compared to the ``expected'' ones \cite{DBLP:journals/vldb/GiatrakosAADG20}.
In these cases,
we thus restrict our comparison to the systems which can actually accommodate the target patterns.
For all patterns,
we used windows,
even though Wayeb can accommodate windowless patterns.
Since windows are ubiquitous in CER (for performance issues), 
we decided to focus on windowed \sremo\ in our experiments. 
We also fixed the selection strategy to \skipany,
since this is the most demanding strategy,
both in terms of time and space complexity.
For all experiments described here,
we have made sure that all engines produce the same results for each pattern.

The benchmark suite runs each experiment,
i.e.,
each combination of engine, pattern and window size, 
3 times.
We report the average throughput and memory footprint.
Throughput is measured in terms of (input) events processed per second,
whereas memory is measured in terms of used memory (MB).
For each run,
multiple memory measurements are taken,
one every 10.000 input events.
Before the measurement,
the garbage collector is explicitly called.
We report the average of those memory measurements.
The time we use to calculate throughput includes both the time required to process input events (update the state(s) of the automaton, create new runs, discard old ones, etc.) and the time required to report any complex events. 
However,
we have slightly modified the notion of ``reporting a complex event''.
Typically,
a complex event is ``reported'' by being printed on the standard output,
stored in a file/database or pushed,
via a messaging system (such as Kafka), 
to other ``event consumers''.
However,
such steps are generally expensive and system/architecture dependent,
which would make throughput estimations less robust.
In order to address this issue,
we perform a different step after every complex event detection,
instead of ``reporting'' it.
We scan every simple event contained in a complex event,
we check whether the remainder of the division of an event's timestamp by 10 is 0 and increment a counter if this condition is satisfied.
We thus avoid the cost of accessing the standard output, files and/or databases while ensuring that complex events are not ``ignored'' and undergo a certain, minimal processing. 

We considered using implementation-independent metrics in order to ``properly'' compare the different systems.
However, 
the different implementations vary widely and do not necessarily share common operators which could act as basic measurement blocks.
This is especially true for Esper,
which, 
besides automata,
also employs trees and Allen's interval algebra.
For this reason,
we decided to follow previous work on comparing different CER systems,
where throughput is used as a metric
\cite{DBLP:journals/pvldb/BucchiGQRV22,DBLP:conf/sigmod/ZhangDI14}.
Note, however,
that the compared systems are all JVM-based,
thus significantly limiting the effect of language choice on their performance.  
With respect to complexity,
the publicly available implementation of SASE is very similar to Wayeb.
Thus, they have similar complexities. 
However, their performance might vary significantly due to differences in the constants of Eq. \eqref{eq:cost} concerning the costs of run cloning/updating.
Concerning FlinkCEP,
according to its source code \cite{FlinkCEPNFA},
it closely follows the version of SASE presented in \cite{DBLP:conf/sigmod/AgrawalDGI08}.
It is not clear which optimizations are actually implemented and what their effects on FlinkCEP's complexity are.
Finally,
Esper's documentation discusses the complexity of some operations,
but not those of pattern matching \cite{EsperComplexity}.

%\input{./tex/exp_synthetic_symbolic}

%\subsection{Symbolic}
%\label{section:exp:symbolic}
%\input{./tex/exp_symbolic}

\subsection{Sequential patterns}
\label{section:exp:seq}
\input{exp_seq}

\subsection{Varying window size}
\label{section:exp:seq:varwin}
\input{exp_seq_var_win}

\subsection{Patterns with Kleene operators}
\label{section:exp:kleene}
\input{exp_kleene}

\subsection{Patterns with nested Kleene operators}
\label{section:exp:kleene_nested}

\input{exp_kleene_nested}

\subsection{Patterns with other operators}
\label{section:exp:other}
\input{exp_other}

%\subsubsection{Selection strategies}

%\subsubsection{Pattern length}

%\subsubsection{Window size}

%% file: exp_seq.tex
\begin{figure}[t]
\centering
\begin{subfigure}[t]{0.32\textwidth}
	\includegraphics[width=0.99\textwidth]{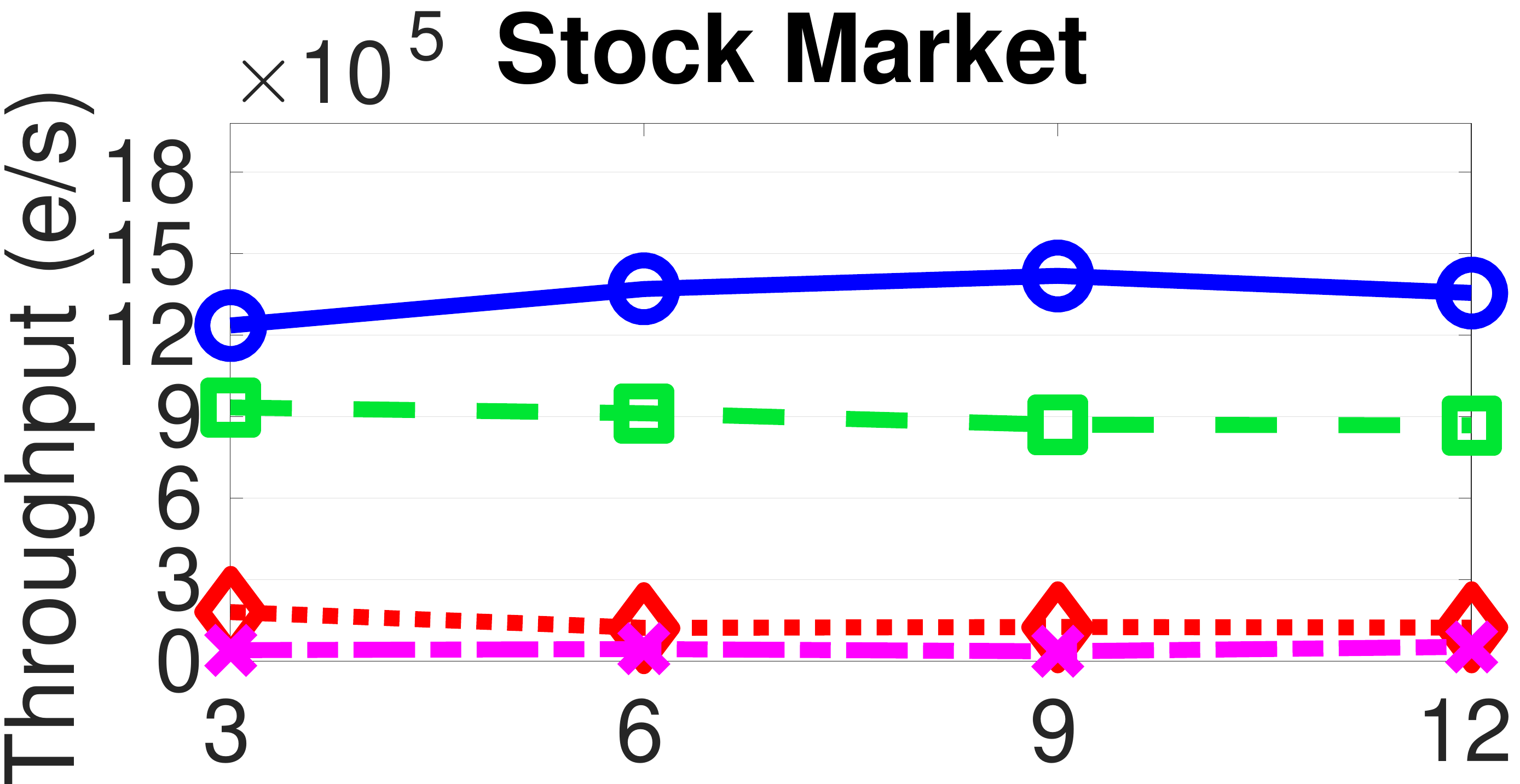}
	\label{fig:exp-seq-throughput-nary-linear-stockw500}
\end{subfigure}
\begin{subfigure}[t]{0.32\textwidth}
	\includegraphics[width=0.99\textwidth]{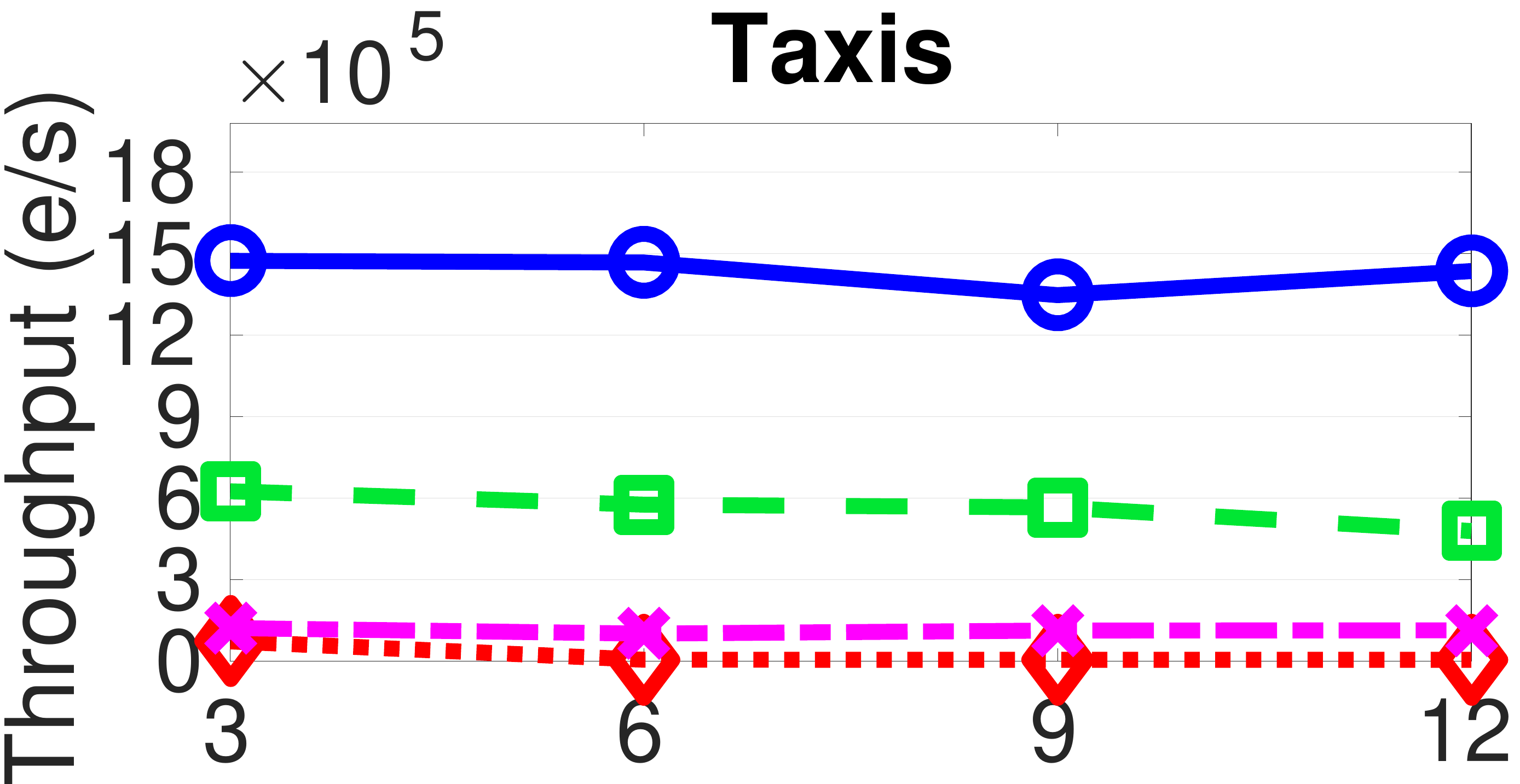}
	\label{fig:exp-seq-throughput-nary-linear-taxiw100}
\end{subfigure}
\begin{subfigure}[t]{0.32\textwidth}
	\includegraphics[width=0.99\textwidth]{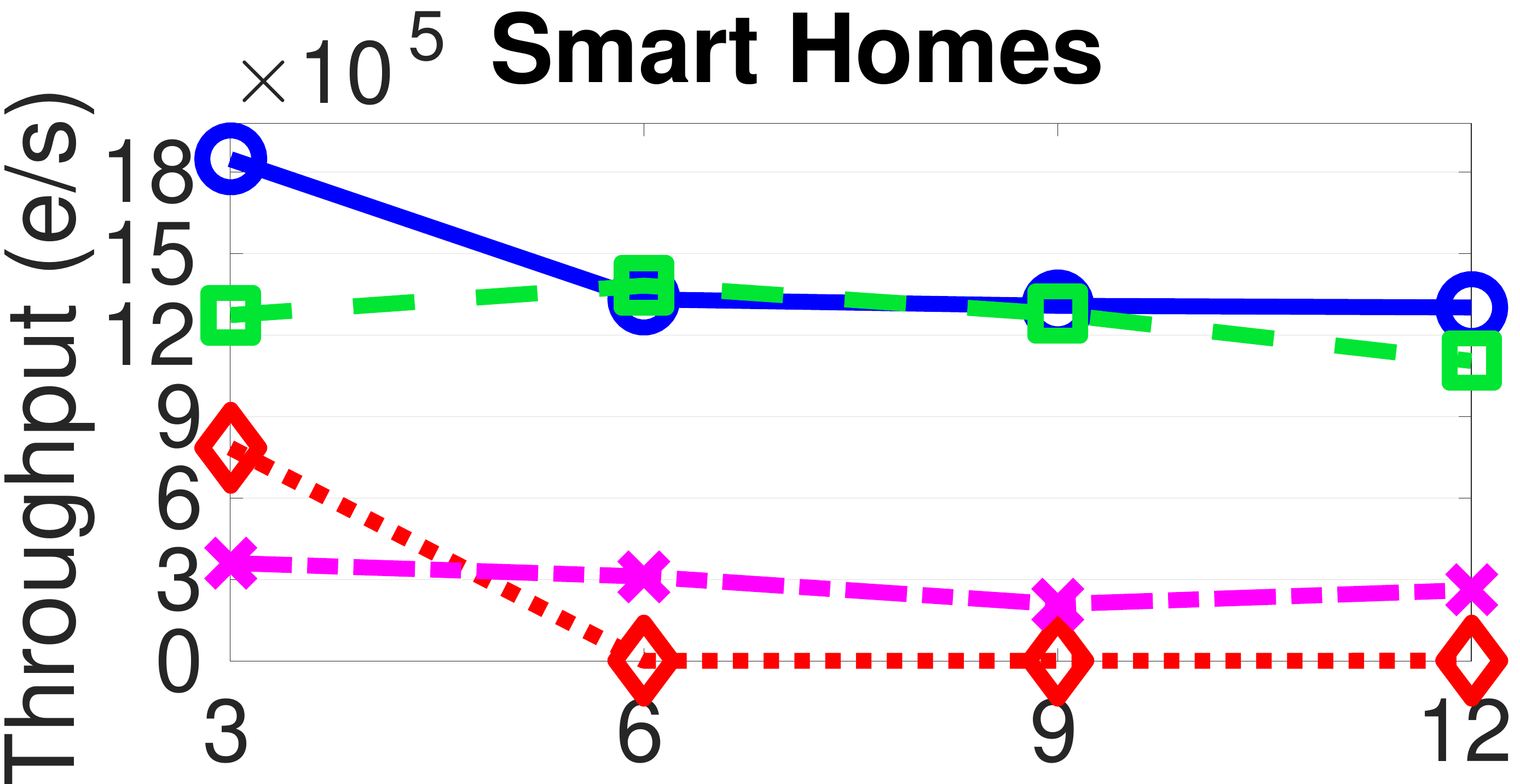}
	\label{fig:exp-seq-throughput-nary-linear-smartw5}
\end{subfigure}
\begin{subfigure}[t]{0.32\textwidth}
	\includegraphics[width=0.99\textwidth]{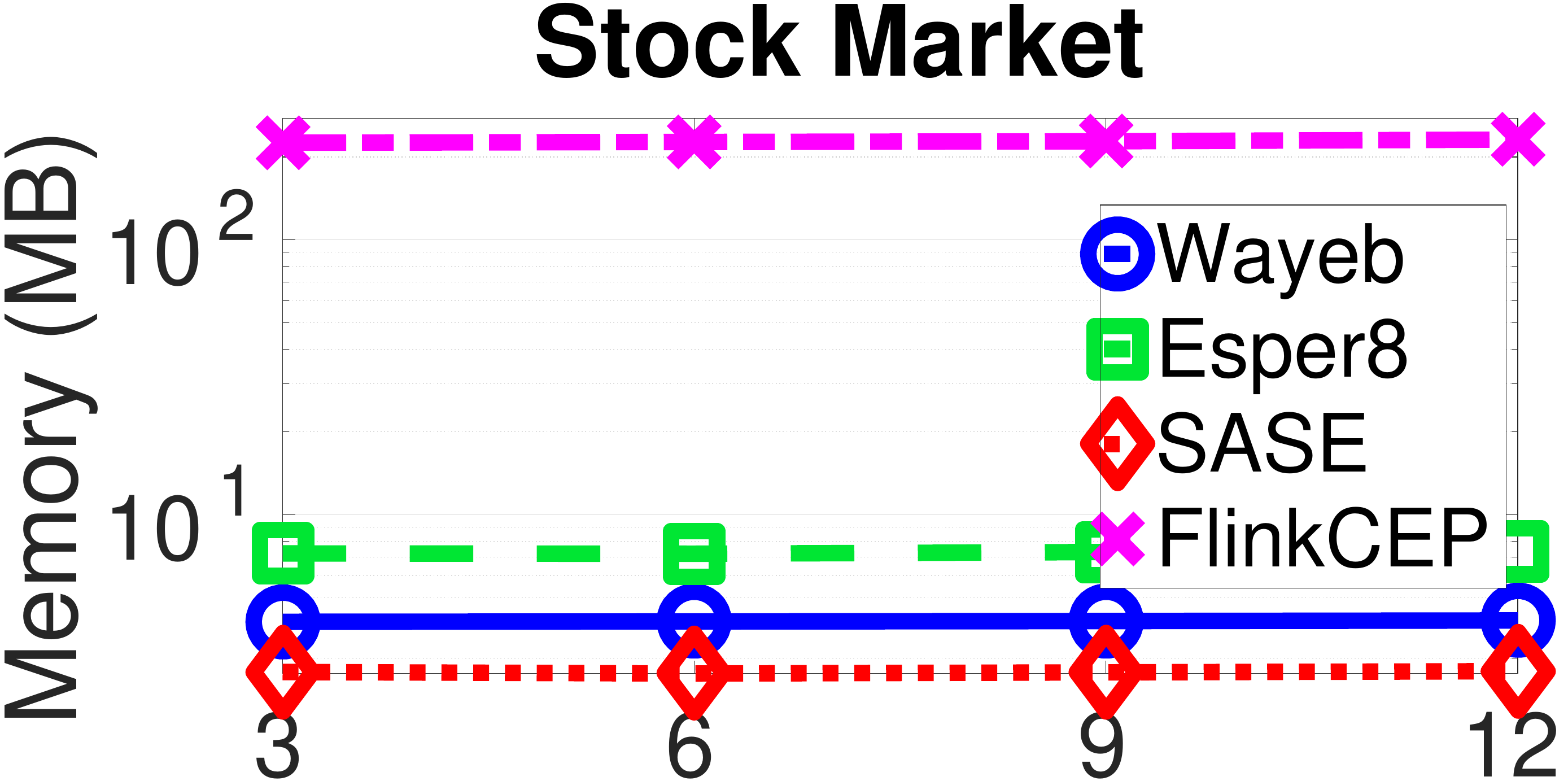}
	\label{fig:exp-seq-memory-nary-semilog-stockw500}
\end{subfigure}
\begin{subfigure}[t]{0.32\textwidth}
	\includegraphics[width=0.99\textwidth]{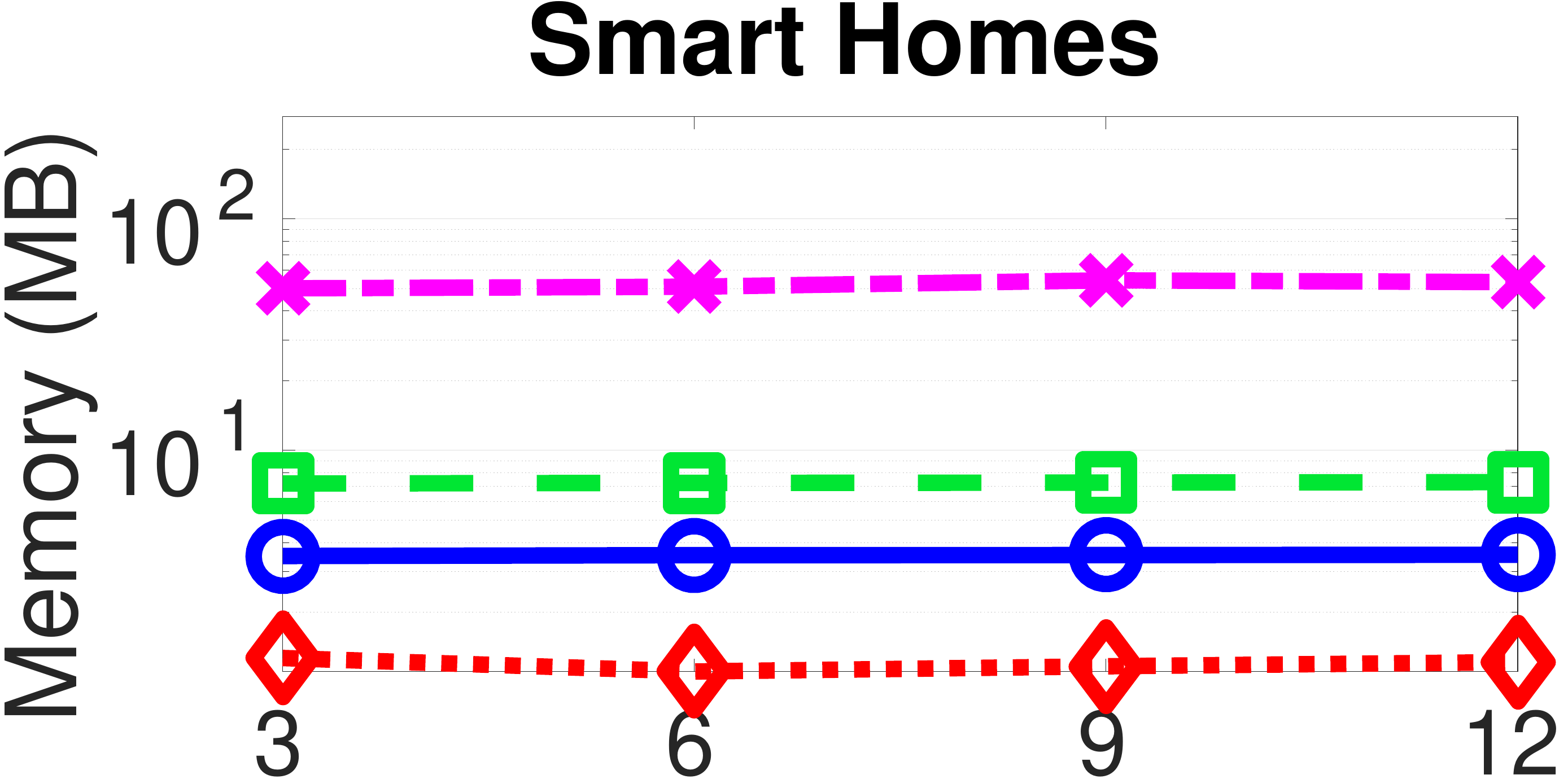}
	\label{fig:exp-seq-memory-nary-semilog-smartw5}
\end{subfigure}
\begin{subfigure}[t]{0.32\textwidth}
	\includegraphics[width=0.99\textwidth]{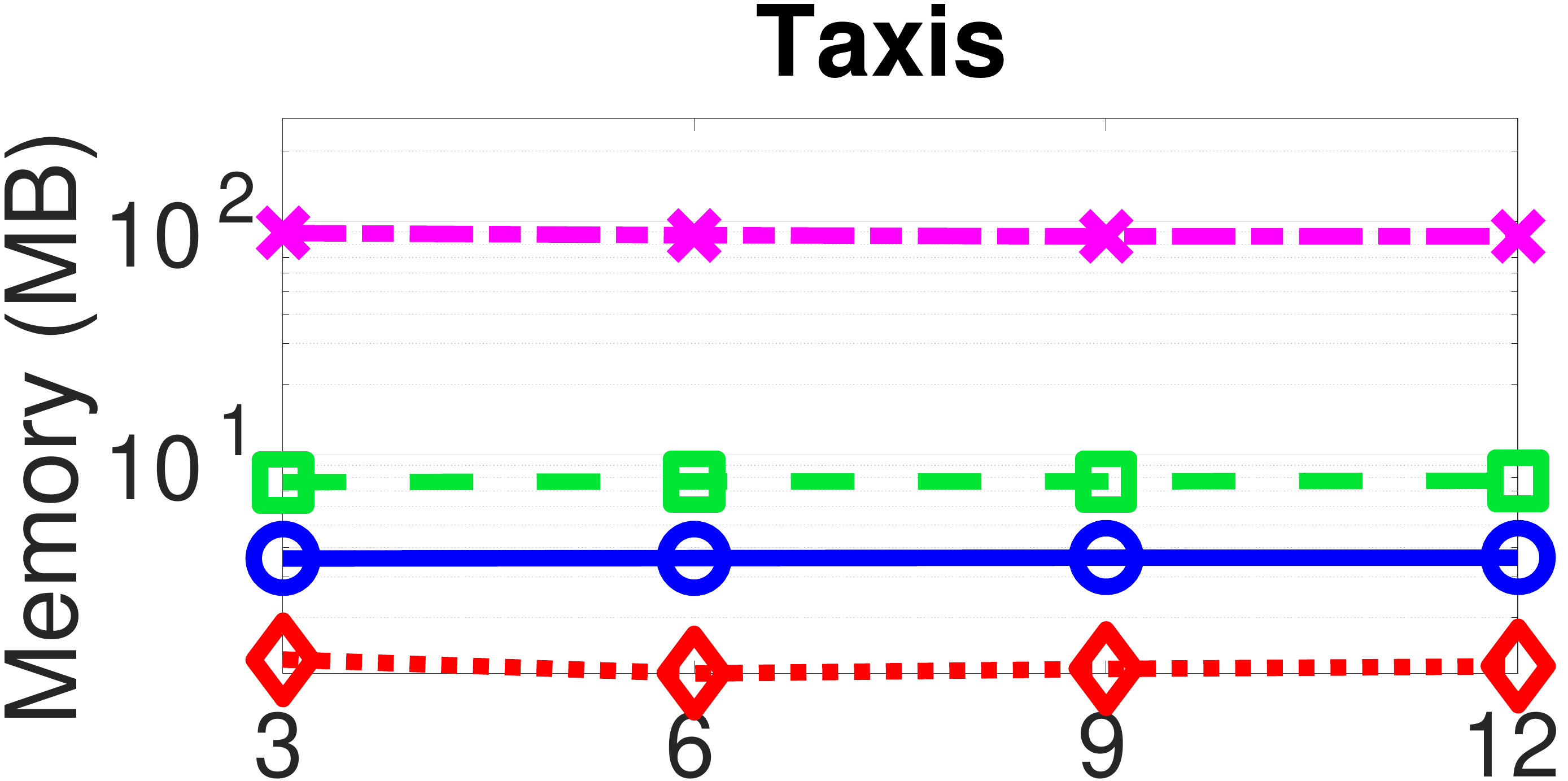}
	\label{fig:exp-seq-memory-nary-semilog-taxiw100}
\end{subfigure}
\caption{Throughput and memory consumption for sequential patterns with $n$-ary predicates as a function of pattern length. Window sizes are $w_{\mathit{stock}}=500$, $w_{\mathit{smart}}=5$, $w_{\mathit{taxi}}=100$.}
\label{fig:exp-seq-nary-linear}
\end{figure}

Our first set of experiments is focused on simple, sequential patterns.
We begin with patterns of the following form:
\begin{equation}
\label{sremo:seq3}
seq_{3} := \circlearrowleft((\phi_{1}(\sim) \uparrow \bullet \downarrow r_{1}) , (\phi_{2}(\sim) \uparrow \bullet) , (\phi_{3}(\sim,r_{1}) \uparrow \bullet))^{[1..w]}
\end{equation}
where $\circlearrowleft$ denotes the \skipany\ selection strategy (see Definition \ref{definition:skipany}),
$w$ is the window size
and $\phi_{1}$, $\phi_{2}$, $\phi_{3}$ all contain ``local'' constraints,
i.e.,
conditions applied to the single, most recently read event.
$\phi_{3}$ also contains a condition relating the most recently read input event with the event that triggered $\phi_{1}$.
For example, 
in the stock market dataset, 
we have:
\begin{equation*}
\label{sremo:seq3predicates}
\begin{aligned}
\phi_{1}(x) := &\ x.\mathit{name}=INTC  \\
\phi_{2}(x) := &\ x.\mathit{name}=RIMM  \\
\phi_{3}(x,y) := &\ (x.\mathit{name}=QQQ \wedge x.\mathit{price} > y.\mathit{price})
\end{aligned}
\end{equation*}
This specific pattern captures a sequence of three stock ticks from three given companies.
The relational constraint is that the stock price of the last event should be greater than the price of the first event.
For each such pattern,
we run experiments for variable pattern ``length''.
We say that the length of the Pattern in eq. \eqref{sremo:seq3} is 3 because it is composed of 3 terminal sub-expressions.
We can increase the length of the pattern by adding more such expressions. 
In our experiments we have used patterns of length 3, 6, 9 and 12.
For example,
the pattern of length 6 has the following form:
\begin{equation}
\label{sremo:seq6}
\begin{aligned}
seq_{6} := & \circlearrowleft((\phi_{1}(\sim) \uparrow \bullet \downarrow r_{1}) , (\phi_{2}(\sim) \uparrow \bullet) , (\phi_{3}(\sim,r_{1}) \uparrow \bullet) , \\
		& (\phi_{4}(\sim) \uparrow \bullet \downarrow r_{2}) , (\phi_{5}(\sim) \uparrow \bullet) , (\phi_{6}(\sim,r_{2}) \uparrow \bullet))^{[1..w]}
\end{aligned}
\end{equation}
The general template remains the same,
i.e.,
$\phi_{4}$, $\phi_{5}$ and $\phi_{6}$ all apply local filters with a given company name.
For every three new sub-expressions we also add a relational constraint
(e.g., between $\phi_{4}$ and $\phi_{6}$ in Pattern \eqref{sremo:seq6}).
The window size is kept constant
(e.g., for the stock market dataset, $w=500$).
The match frequency (ratio of complex to input events) is in the range of $0.36\% - 0\%$ for the stock market dataset (the lengthier the pattern the lower the number of detected matches),
$0.36\% - 0\%$ for the smart homes dataset and $0.05\% - 0\%$ for the taxis dataset.
Note that our purpose in using such patterns is to stress test the systems under controlled conditions.
Some of the patterns may not be intuitive from a practical point of view,
but allow for controlled experiments.
This is the reason why we use a symmetrical and repeatable structure in the patterns when increasing their length.
We aim at testing the effect of length,
without introducing other performance affecting factors.

Figure \ref{fig:exp-seq-nary-linear} presents throughput and memory results for the aforementioned sequential patterns and for all datasets.
Wayeb and Esper stand out clearly as the most efficient engines in terms of throughput. 
Wayeb also has a significant advantage over Esper in most experiments and a slight advantage for the smart homes dataset.
For example,
Wayeb is almost 2.5 times as efficient as Esper for the taxis dataset.
Interestingly,
the length of the pattern does not always have a negative effect on throughput.
A decrease in throughput is significant only for the smart homes dataset.
FlinkCEP has by far the heaviest memory footprint,
while the other systems seem to have a similar performance.
Wayeb has a slightly better performance than Esper,
its main competitor in terms of throughput.
In general,
we see that the performance is relatively stable as a function of pattern length for all systems.
This is especially true for memory.
SASE's low memory footprint can be attributed to its general lightweight construction
(the other systems are designed to perform additional tasks,
besides vanilla, single-core CER) 
and its memory optimization schemes,
such as run recycling.
Throughput exhibits slight variations.
This observation implies that the number of created runs does not vary greatly for the tested sequential patterns.

%% file: exp_seq_var_win.tex
\begin{figure}[t]
\centering
\begin{subfigure}[t]{0.32\textwidth}
	\includegraphics[width=0.99\textwidth]{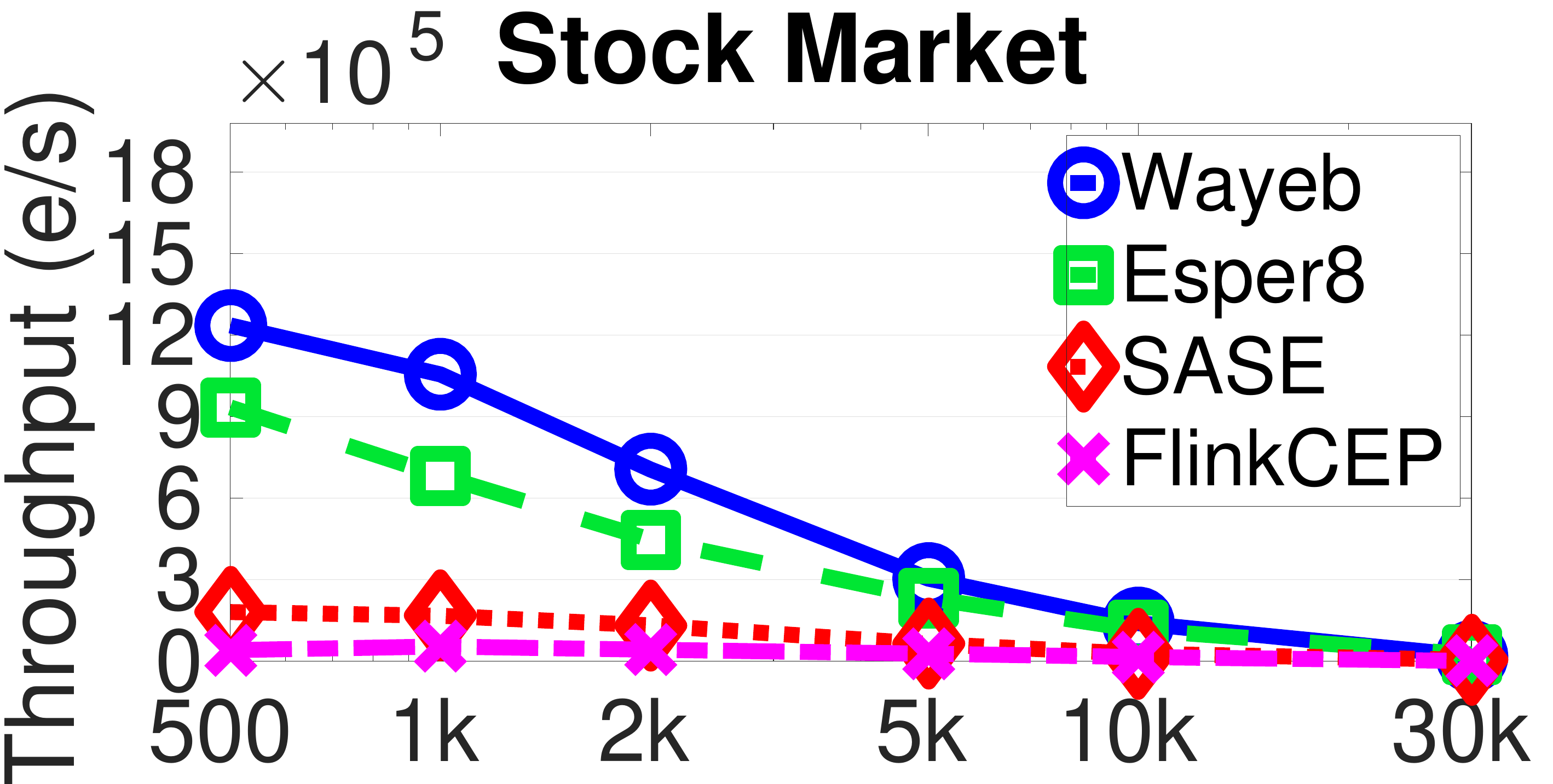}
	\label{fig:exp-seq-var-win-throughput-nary-linear-stockq1}
\end{subfigure}
\begin{subfigure}[t]{0.32\textwidth}
	\includegraphics[width=0.99\textwidth]{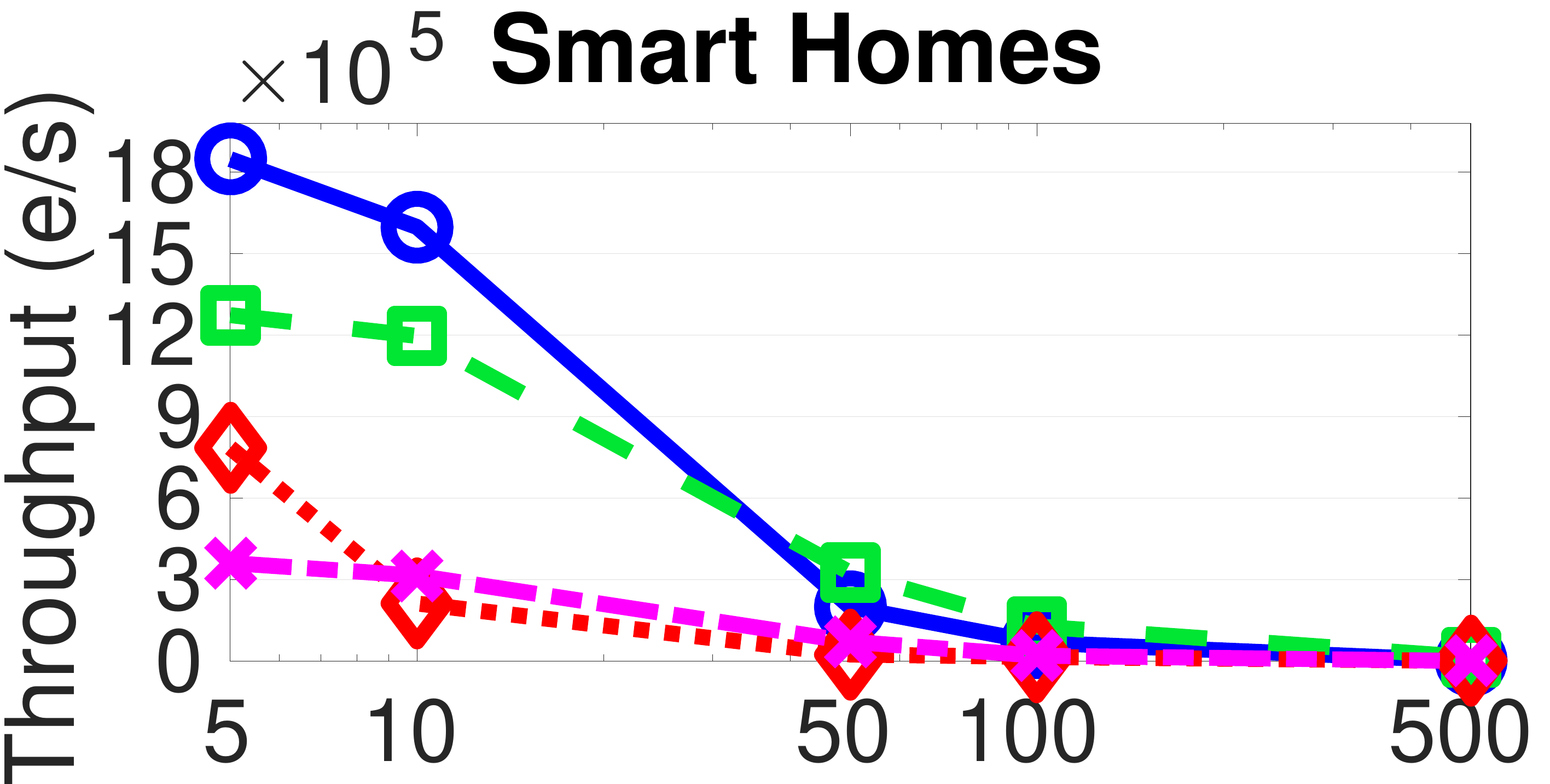}
	\label{fig:exp-seq-var-win-throughput-nary-linear-smartq1}
\end{subfigure}
\begin{subfigure}[t]{0.32\textwidth}
	\includegraphics[width=0.99\textwidth]{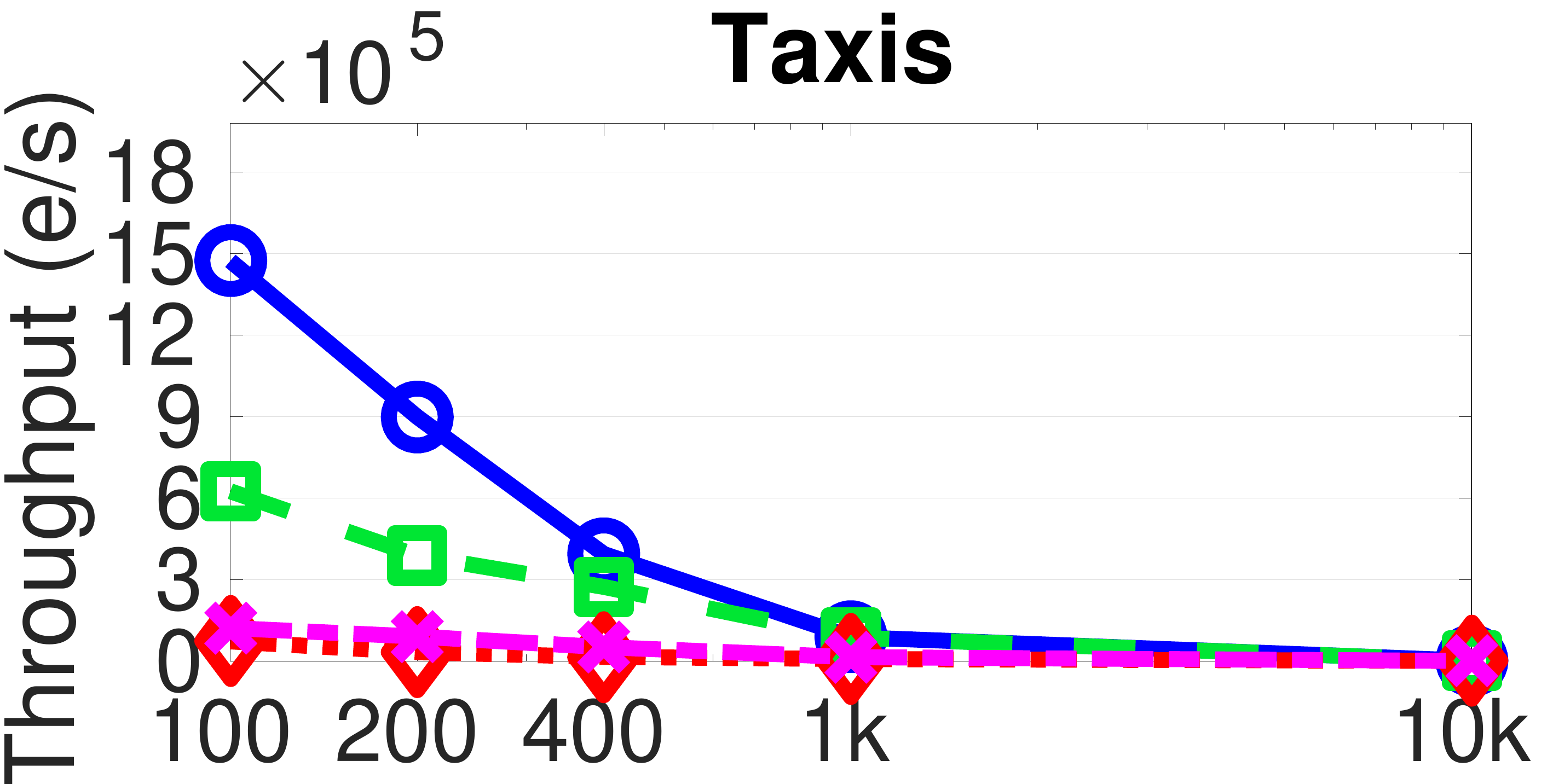}
	\label{fig:exp-seq-var-win-throughput-nary-linear-taxiq1}
\end{subfigure}
\begin{subfigure}[t]{0.32\textwidth}
	\includegraphics[width=0.99\textwidth]{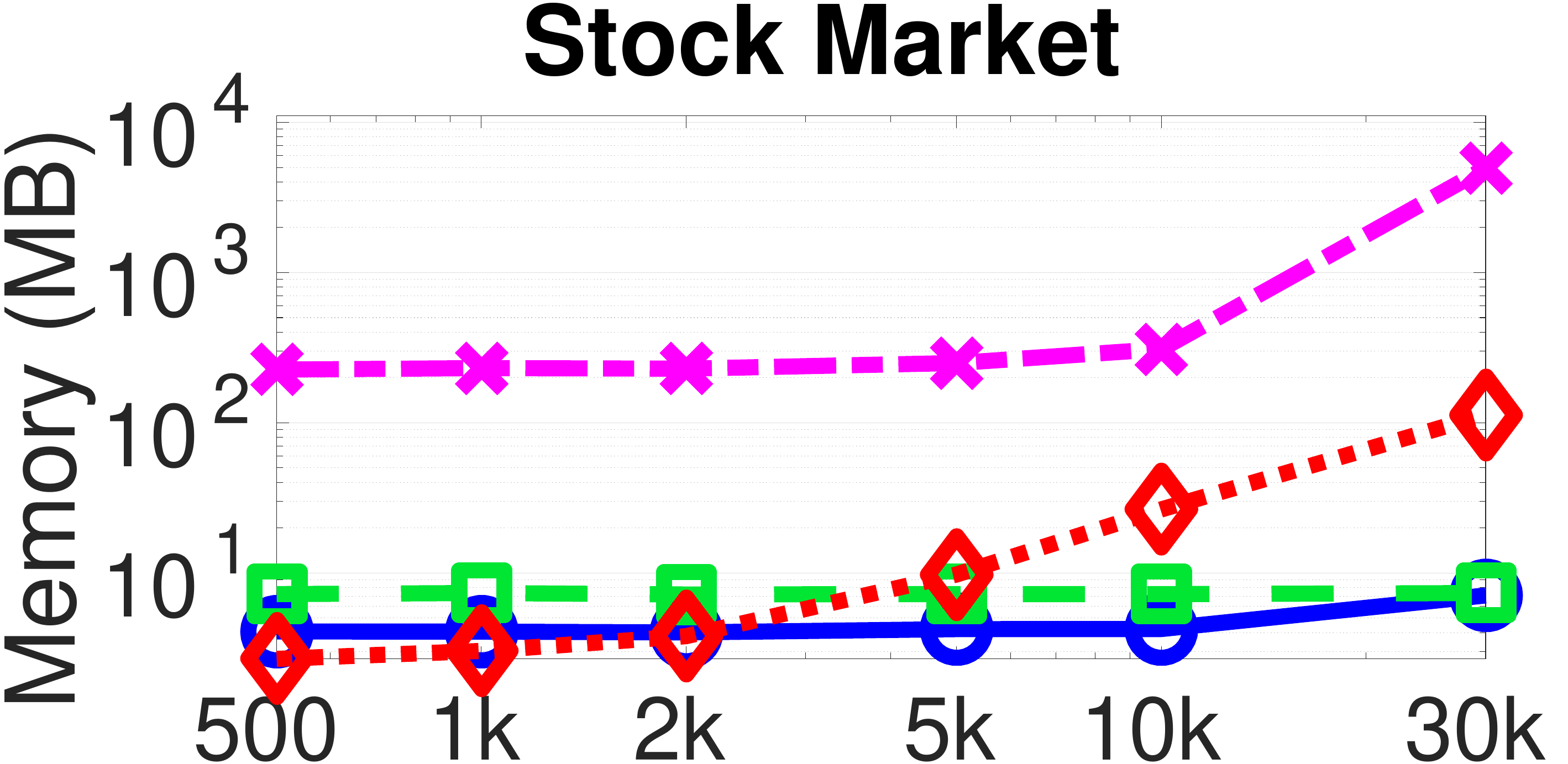}
	\label{fig:exp-seq-var-win-memory-nary-linear-stockq1}
\end{subfigure}
\begin{subfigure}[t]{0.32\textwidth}
	\includegraphics[width=0.99\textwidth]{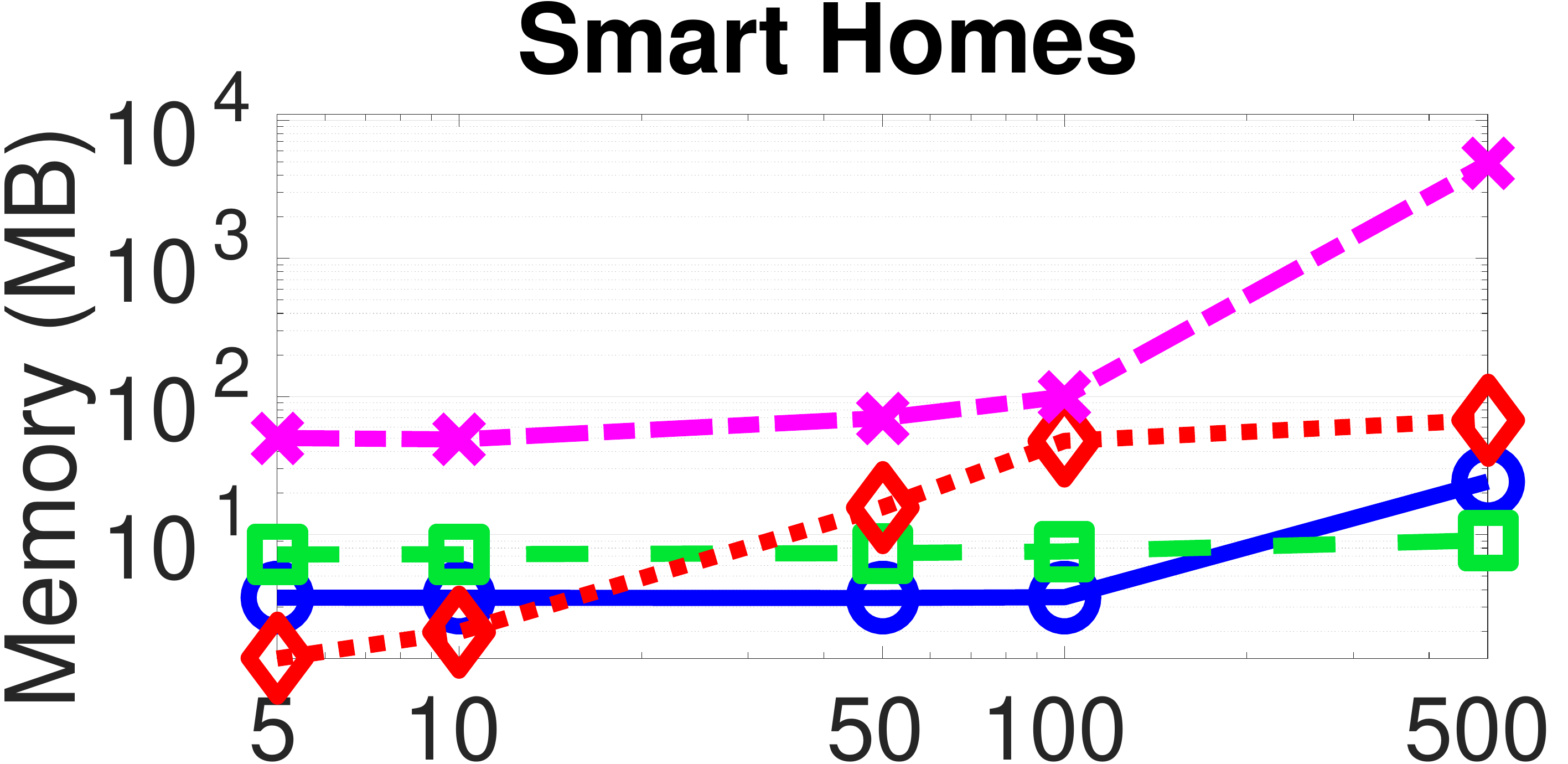}
	\label{fig:exp-seq-var-win-memory-nary-linear-smartq1}
\end{subfigure}
\begin{subfigure}[t]{0.32\textwidth}
	\includegraphics[width=0.99\textwidth]{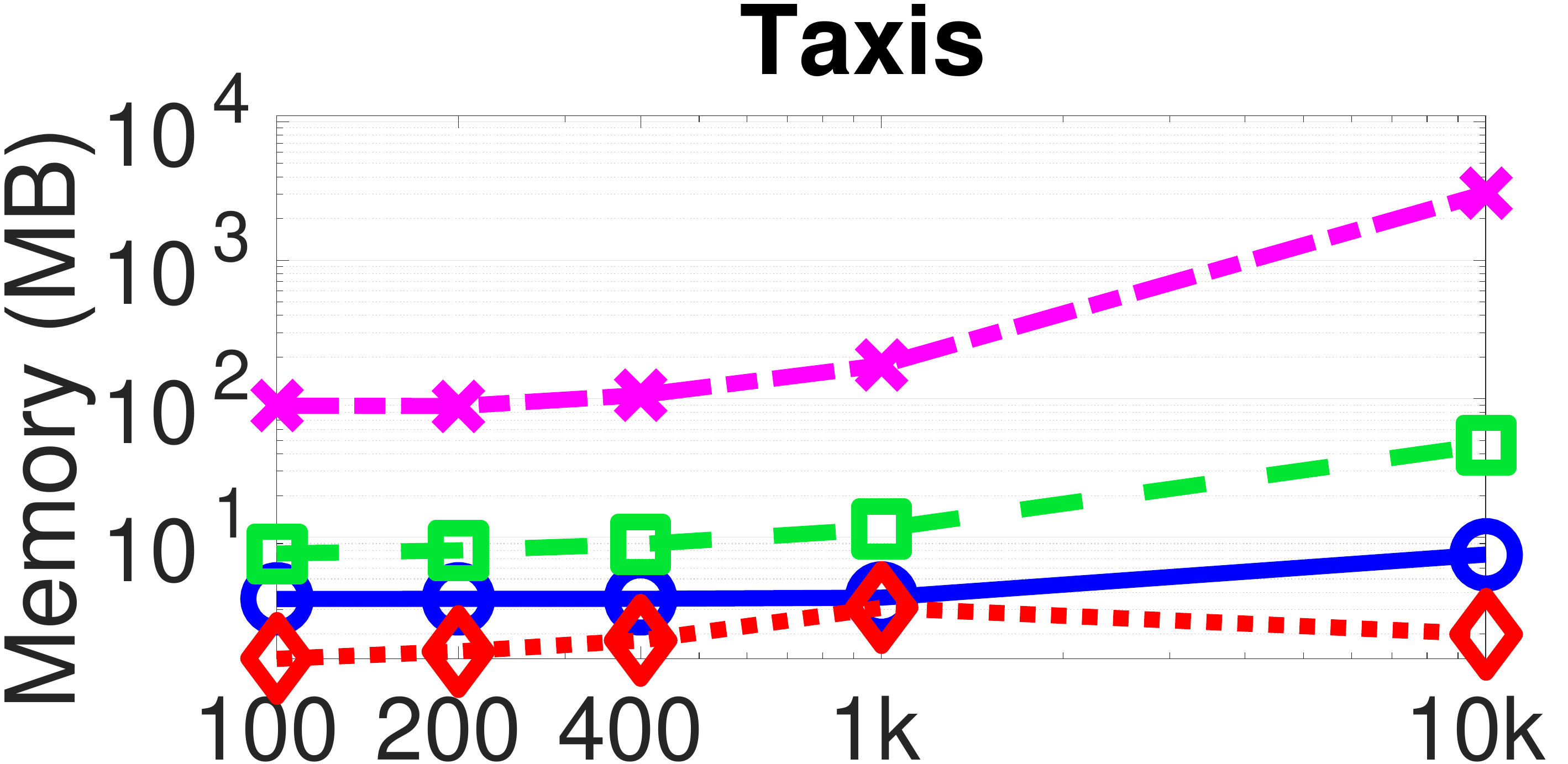}
	\label{fig:exp-seq-var-win-memory-nary-linear-taxiq1}
\end{subfigure}
\caption{Throughput and memory consumption for sequential patterns with $n$-ary predicates as a function of window size. Pattern length is 3.}
\label{fig:exp-seq-var-win-nary-linear}
\end{figure}

In the next set of experiments,
we investigated the behavior of all systems with increasing window sizes.
%We thus focused on the sequential patterns which can be accommodated by all systems.
For each dataset,
we increased the window size up to the point where throughput exhibits a significant drop.
Figure \ref{fig:exp-seq-var-win-nary-linear} shows the relevant results.
Wayeb again exhibits the best performance in terms of throughput, 
followed by Esper.
Moreover, 
Wayeb remains better than Esper and FlinkCEP in terms of memory consumption. 
%FlinkCEP seems to be again very demanding in terms of memory.
%Wayeb and Esper start from high throughput figures for small window sizes and their performance deteriorates as the window size increases.
%On the other hand, SASE and FlinkCEP typically have low throughput for all window sizes. 
All systems exhibit a throughput deterioration as the window size increases.
This implies that window size is more important in determining the number of created runs than pattern length. 
Wayeb and Esper also show a stable memory footprint,
indicating that the memory space reserved for the number of runs is small compared to the total space required by the engines.
This conclusion is reinforced by SASE's memory deterioration.
As a bare-bones CER engine,
its memory consumption is dominated by the number of runs,
which is visible in the presented results.

%% file: exp_kleene.tex
\begin{figure}[t]
\centering
\begin{subfigure}[t]{0.32\textwidth}
	\includegraphics[width=0.99\textwidth]{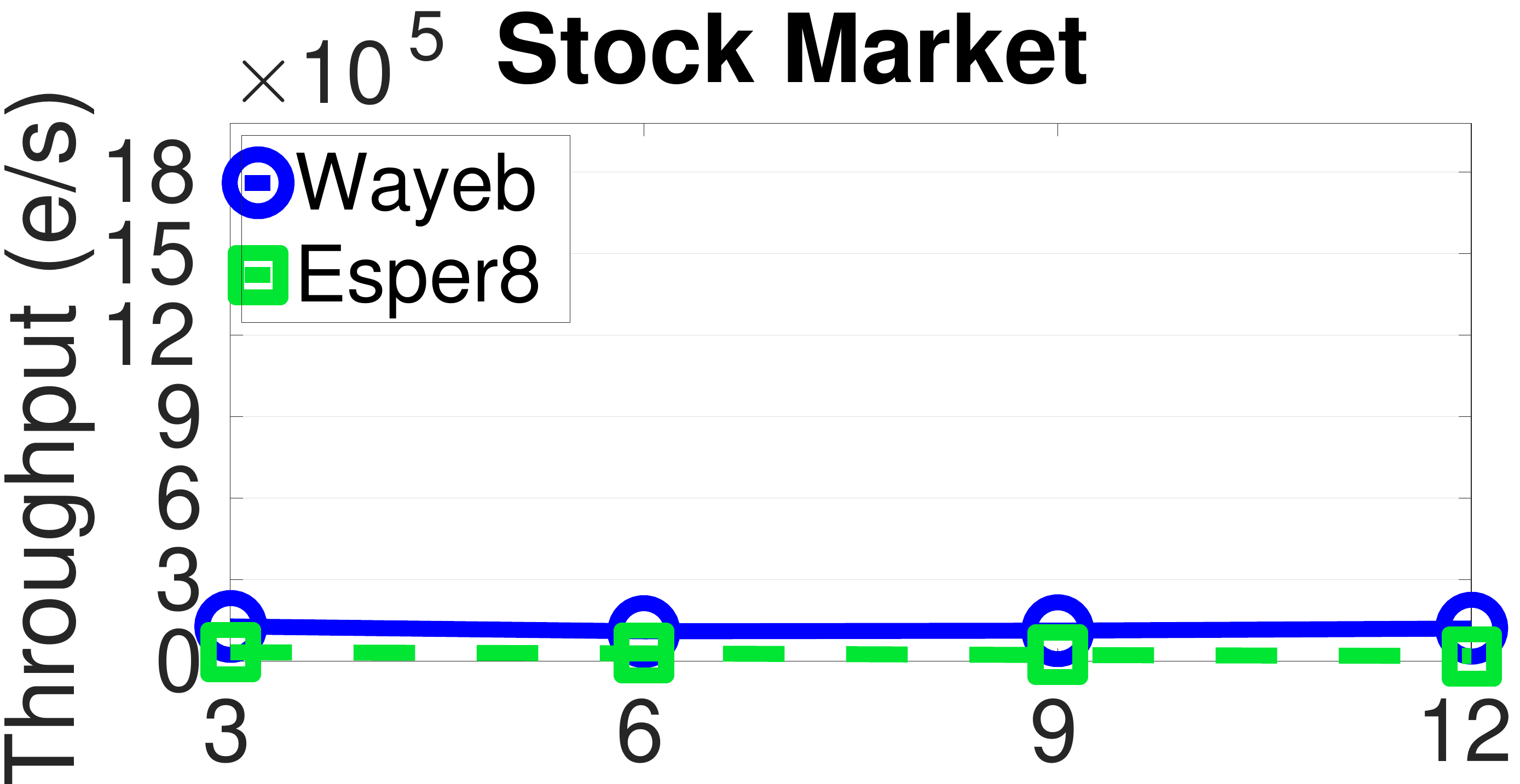}
	\label{fig:exp-kleene-throughput-nary-linear-stockw500}
\end{subfigure}
\begin{subfigure}[t]{0.32\textwidth}
	\includegraphics[width=0.99\textwidth]{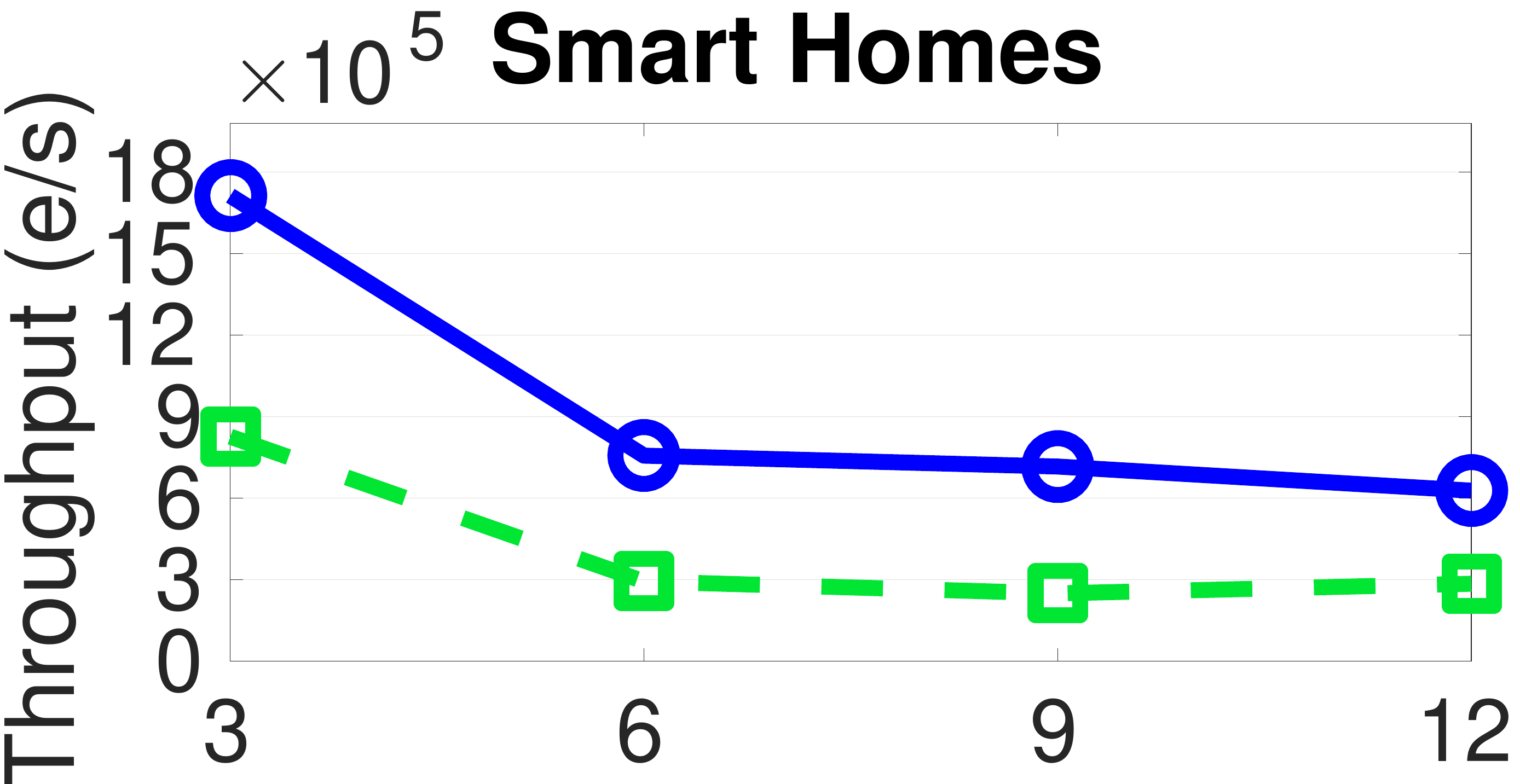}
	\label{fig:exp-kleene-throughput-nary-linear-smartw5}
\end{subfigure}
\begin{subfigure}[t]{0.32\textwidth}
	\includegraphics[width=0.99\textwidth]{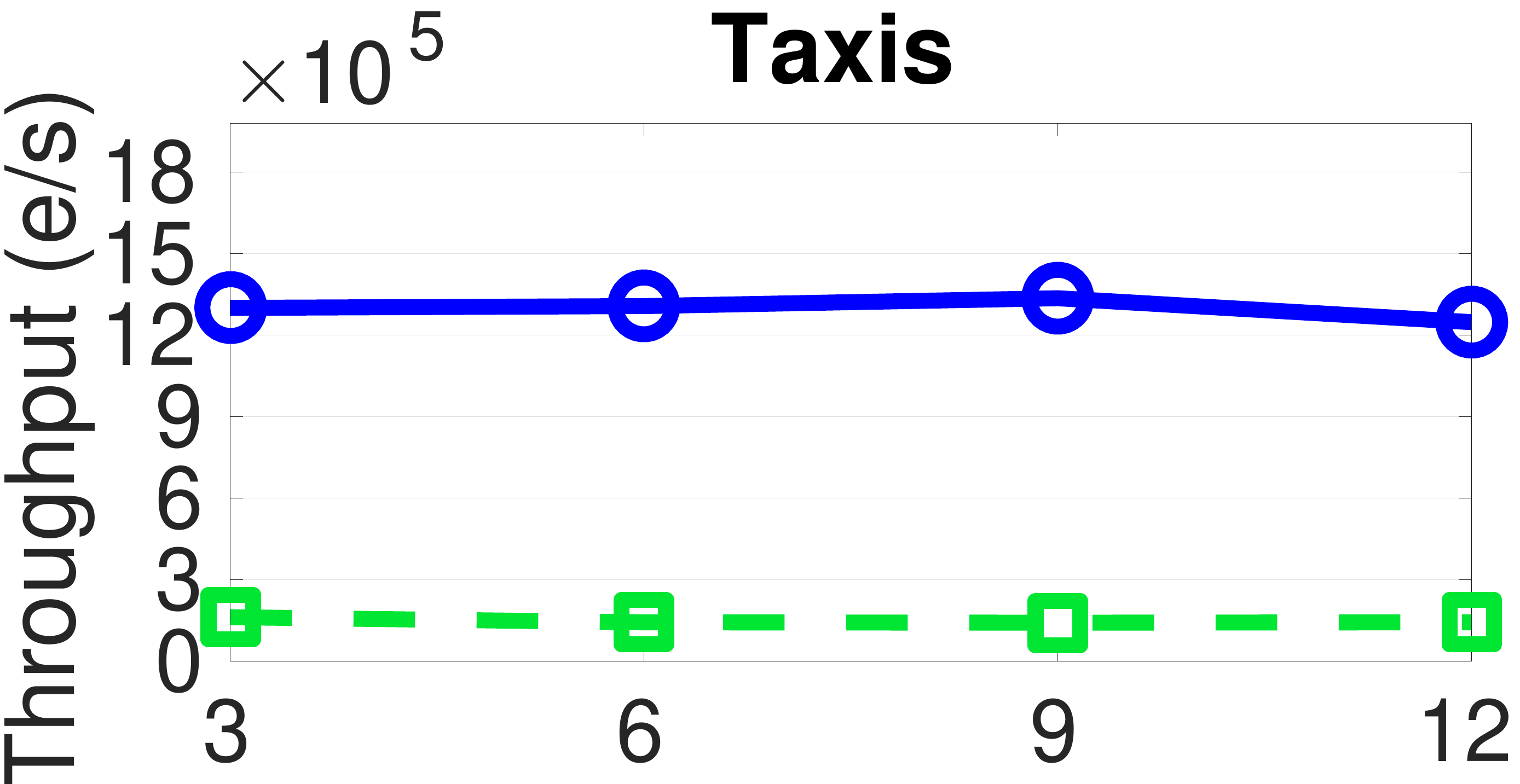}
	\label{fig:exp-kleene-throughput-nary-linear-taxiw100}
\end{subfigure}
\begin{subfigure}[t]{0.32\textwidth}
	\includegraphics[width=0.99\textwidth]{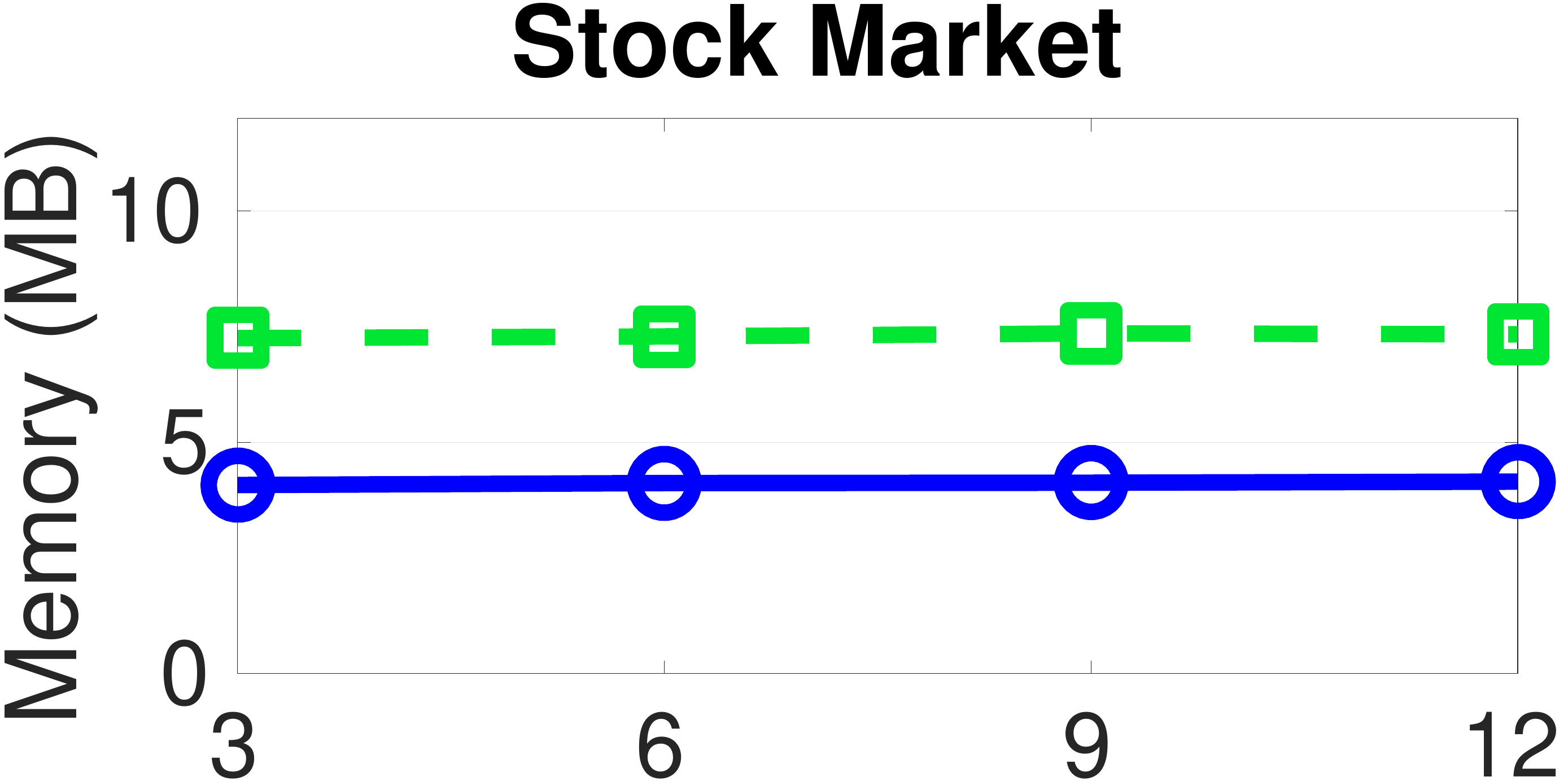}
	\label{fig:exp-kleene-memory-nary-linear-stockw500}
\end{subfigure}
\begin{subfigure}[t]{0.32\textwidth}
	\includegraphics[width=0.99\textwidth]{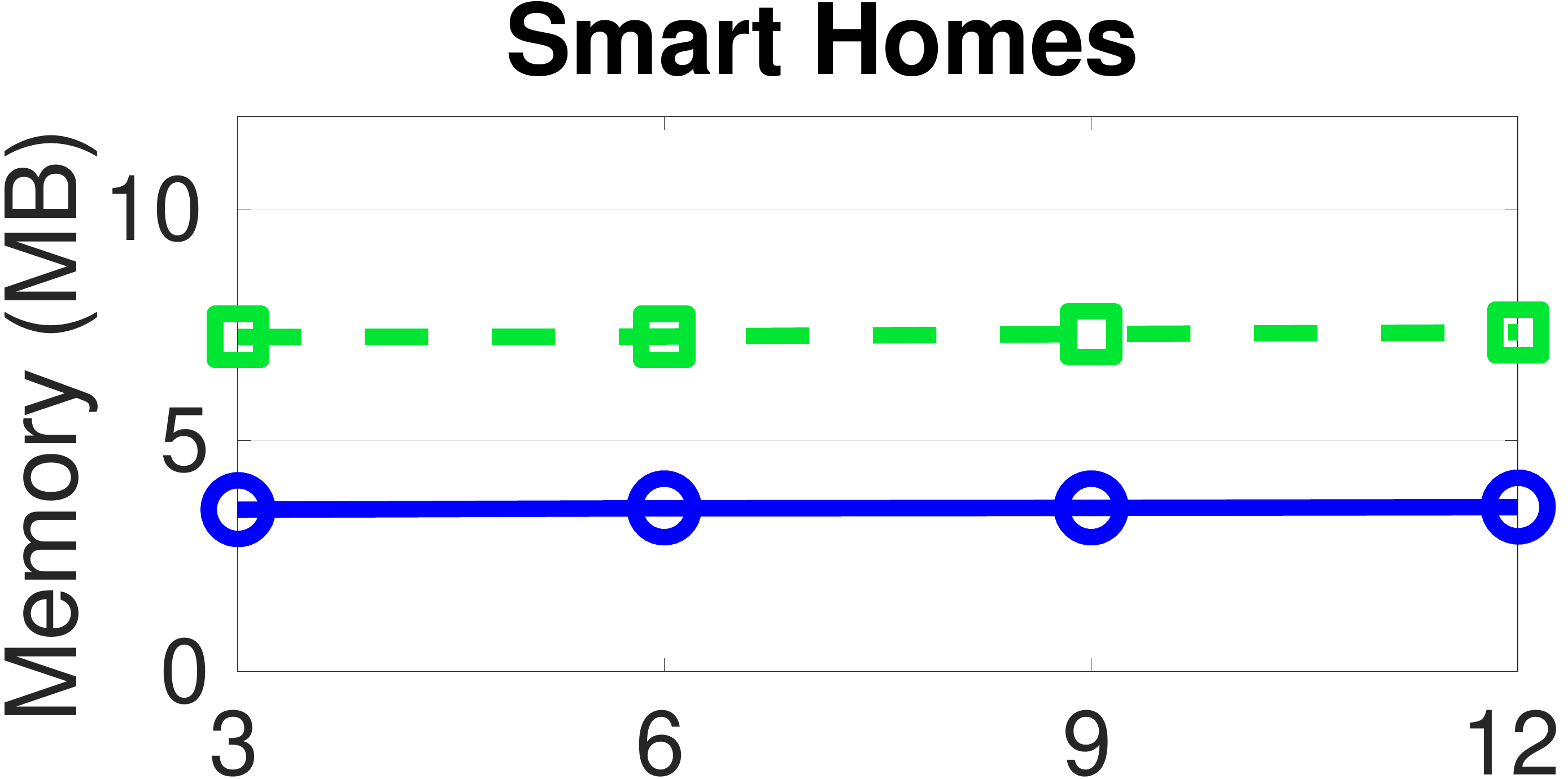}
	\label{fig:exp-kleene-memory-nary-linear-smartw5}
\end{subfigure}
\begin{subfigure}[t]{0.32\textwidth}
	\includegraphics[width=0.99\textwidth]{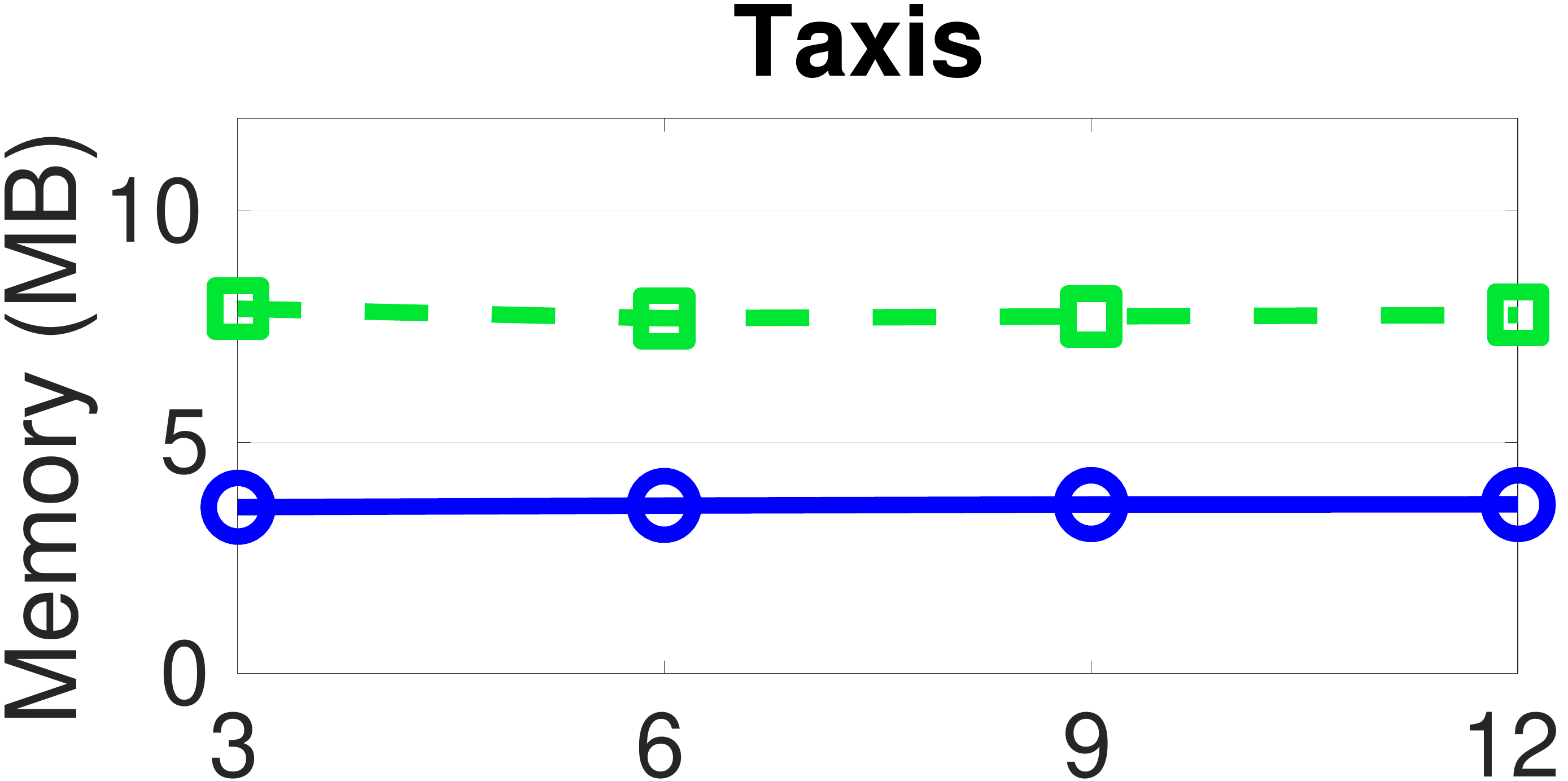}
	\label{fig:exp-kleene-memory-nary-linear-taxiw100}
\end{subfigure}
\caption{Throughput and memory consumption for patterns with $n$-ary predicates and Kleene operators as a function of pattern length. SASE and FlinkCEP are excluded because they do not support patterns with Kleene operators with the expected semantics.}
\label{fig:exp-kleene-nary-linear}
\end{figure}

We now move to patterns containing Kleene operators.
We tested the engines against patterns of the following form:
\begin{equation}
\label{sremo:kleene3}
k_{3} := \circlearrowleft((\phi_{1}(\sim) \uparrow \bullet \downarrow r_{1}) , (\phi_{2}(\sim) \uparrow \bullet)^{+} , (\phi_{3}(\sim,r_{1}) \uparrow \bullet))^{[1..w]}
\end{equation}
Pattern \eqref{sremo:kleene3} is the same as Pattern \eqref{sremo:seq3},
with a single difference. 
$(\phi_{2}(\sim) \uparrow \bullet)^{+}$ is a Kleene-plus operation,
i.e.,
standing for $\phi_{2}(\sim){\cdot}(\phi_{2}(\sim){\uparrow}\bullet)^{*}$,
one or more iterations of $\phi_{2}$
($\phi^{+} := \phi \cdot \phi^{*}$).
Again,
we use patterns of length 3, 6, 9, 12,
gradually increasing the number of Kleene operators
(e.g., patterns of length 6 have 2 such operators). 
The match frequencies are in the range of $0.61\% - 0\%$, $1.35\% - 0\%$, $0.08\% - 0\%$ for the stock market, smart homes and taxis dataset.

Figure \ref{fig:exp-kleene-nary-linear} shows the throughput and memory results.
%Only results for Wayeb and Epser are presented in this figure.
We excluded SASE and FlinkCEP from this set of experiments because they cannot support patterns with Kleene operators with the expected semantics.
As far as SASE is concerned,
although it can accept, compile and run patterns with Kleene operators,
it tends to produce many more matches than those expected from the semantics of \skipany.
This indicates that SASE could possibly suffer from soundness issues,
at least when some operators are used.
FlinkCEP, on the other hand,
has the inverse problem.
Our investigation of FlinkCEP has led us to conclude that this behavior is probably due to the fact that FlinkCEP does not allow the use of \skipany\ within a Kleene operator.
Some matches are thus dropped.
For these reasons,
we focused on Wayeb and Esper which can support patterns with Kleene operators and \skipany.

Wayeb always exhibits higher throughput than Esper.
In some cases
(e.g., for the taxis dataset),
Wayeb's throughput is 6 times that of Esper's. 
Wayeb also has a lower memory footprint.
As expected,
the performance of both Wayeb and Esper for this class of patterns is lower than their performance for sequential patterns.
Due to the presence of Kleene operators,
the engines need to produce many more runs. 
Whenever the stream contains simple events satisfying $\phi_{2}$,
the engines need to keep track of all possible combinations of these events.
This is the reason why more runs are created.

%% file: exp_kleene_nested.tex
At the next level of pattern complexity,
we have patterns with nested Kleene operators.
In order to run experiments with such patterns,
we used expressions of the following form:
\begin{equation}
\label{sremo:kleeneNested4}
kn_{4} := \circlearrowleft((\phi_{1}(\sim) \uparrow \bullet \downarrow r_{1}) , ((\phi_{2}(\sim) \uparrow \bullet) , (\phi_{3}(\sim) \uparrow \bullet)^{+})^{+} , (\phi_{4}(\sim,r_{1}) \uparrow \bullet))^{[1..w]}
\end{equation}
Note that $\phi_{2}$ is under a single Kleene-plus operators whereas $\phi_{3}$ under two.
This expression has 4 terminal sub-expressions.
We also used patterns with 8, 12 and 16 terminal sub-expressions,
i.e., patterns with multiple (2, 3 and 4) nested Kleene operators.

SASE's language does not support patterns with nested Kleene operators.
FlinkCEP has the issues mentioned in Section \ref{section:exp:kleene} regarding the semantics of \skipany\ with iteration.
Esper's language is also not able to support such patterns.
Thus,
Wayeb is the only engine which can properly support nested Kleene operators.

The experimental results are shown in Figure \ref{fig:exp-kleene-nested-nary-linear}.
In order to gain a more complete understanding of Wayeb's behavior,
we show results for multiple values of the window size. 
Wayeb maintains high throughput, 
in the order of millions or hundreds of thousands of events per second,
for most combinations of window size and pattern length.
As the window size increases,
Wayeb's performance deteriorates,
since larger window sizes always lead to more runs being created as Wayeb consumes a stream.
%We can also see that the effect of the pattern length becomes more pronounced as the window size increases.
The combined variation of window size and pattern length in this figure illustrates also the more pronounced combined effect on throughput.
The window size still remains the most important factor for performance.
For large windows though,
the pattern length starts having an impact as well,
since larger windows give a chance to longer patterns to create additional runs.

\begin{figure}[t]
\centering
\begin{subfigure}[t]{0.32\textwidth}
	\includegraphics[width=0.99\textwidth]{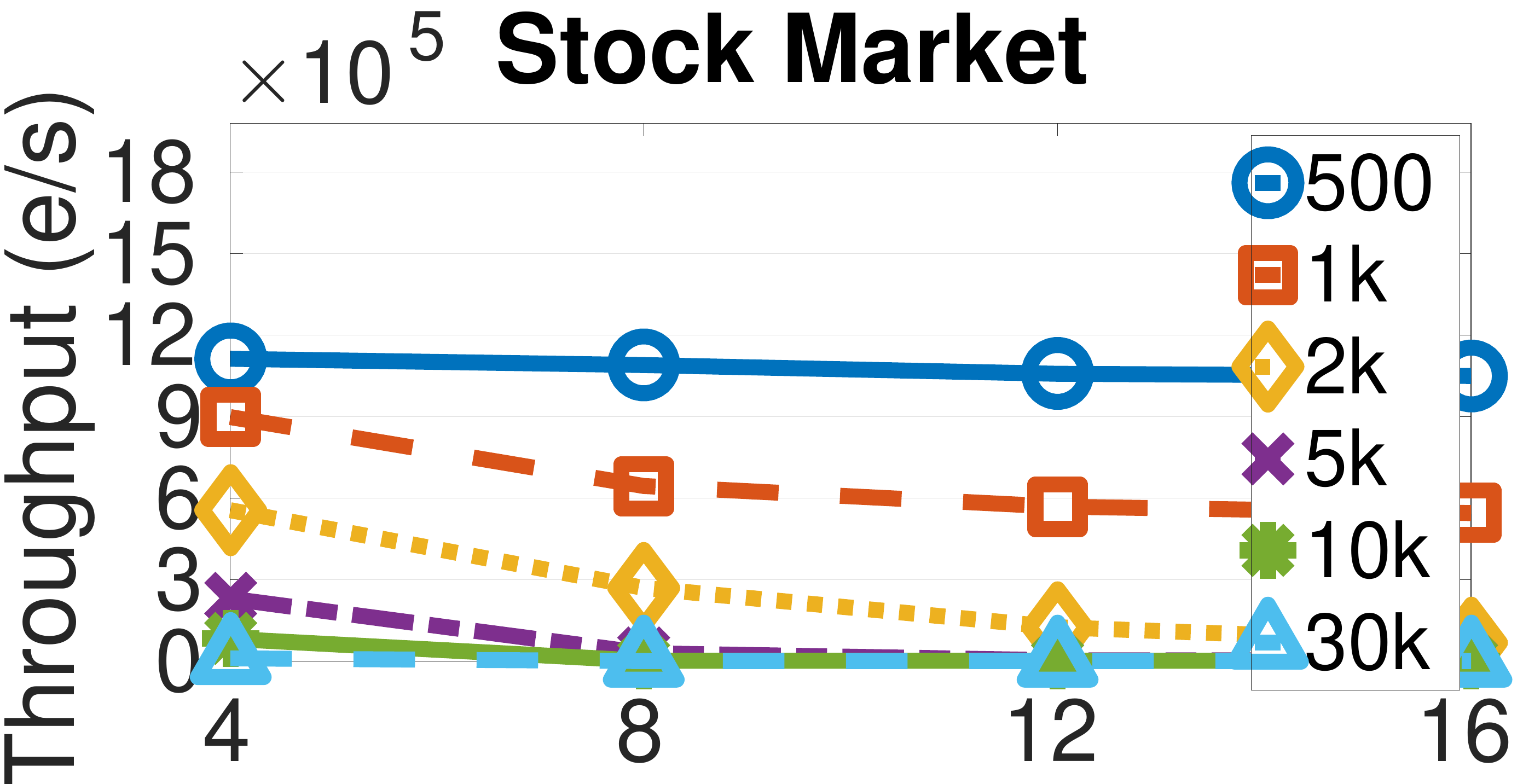}
	\label{fig:exp-kleene-nested-throughput-nary-linear-stock-allwindows}
\end{subfigure}
\begin{subfigure}[t]{0.32\textwidth}
	\includegraphics[width=0.99\textwidth]{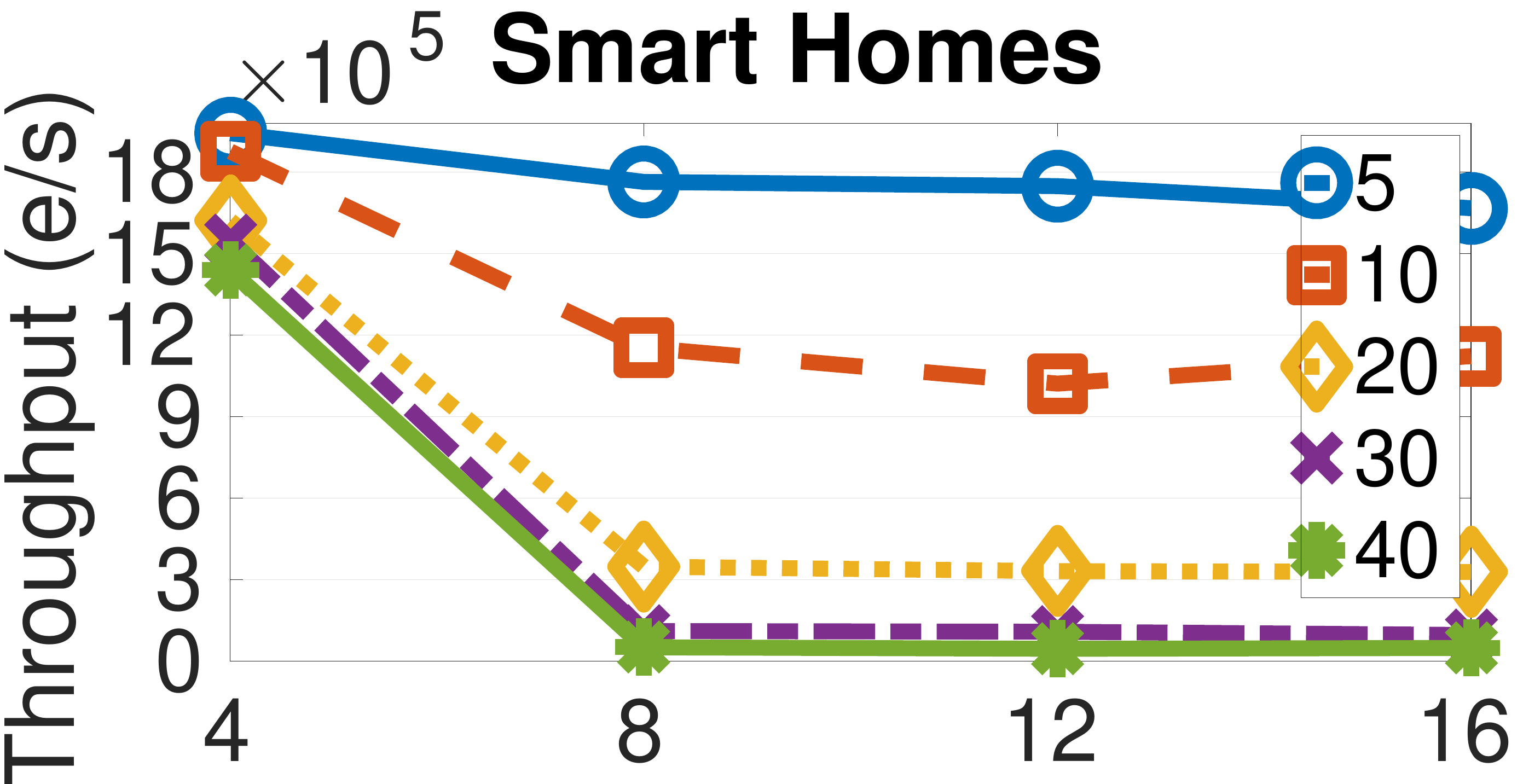}
	\label{fig:exp-kleene-nested-throughput-nary-linear-smart-allwindows}
\end{subfigure}
\begin{subfigure}[t]{0.32\textwidth}
	\includegraphics[width=0.99\textwidth]{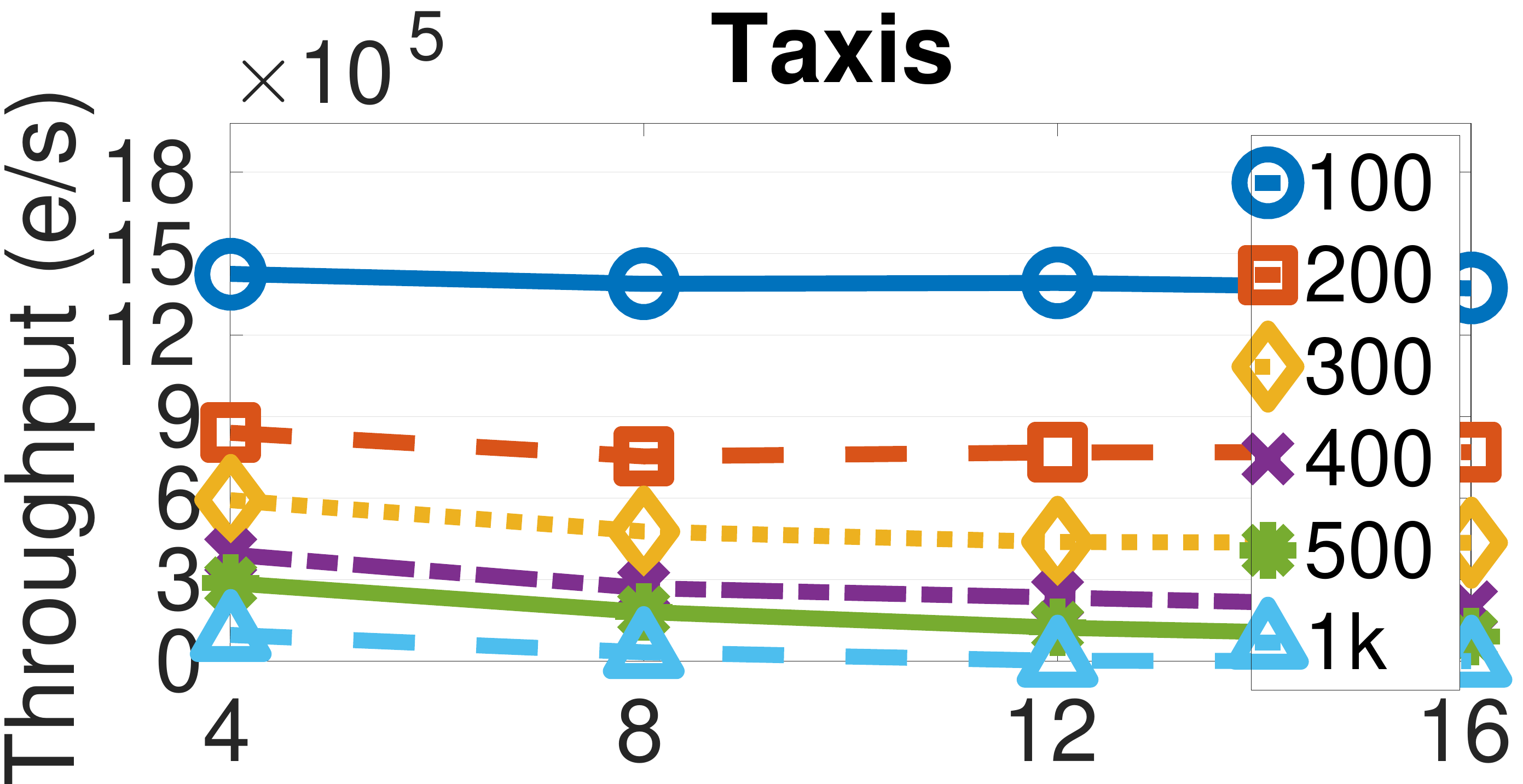}
	\label{fig:exp-kleene-nested-throughput-nary-linear-taxi-allwindows}
\end{subfigure}
\begin{subfigure}[t]{0.32\textwidth}
	\includegraphics[width=0.99\textwidth]{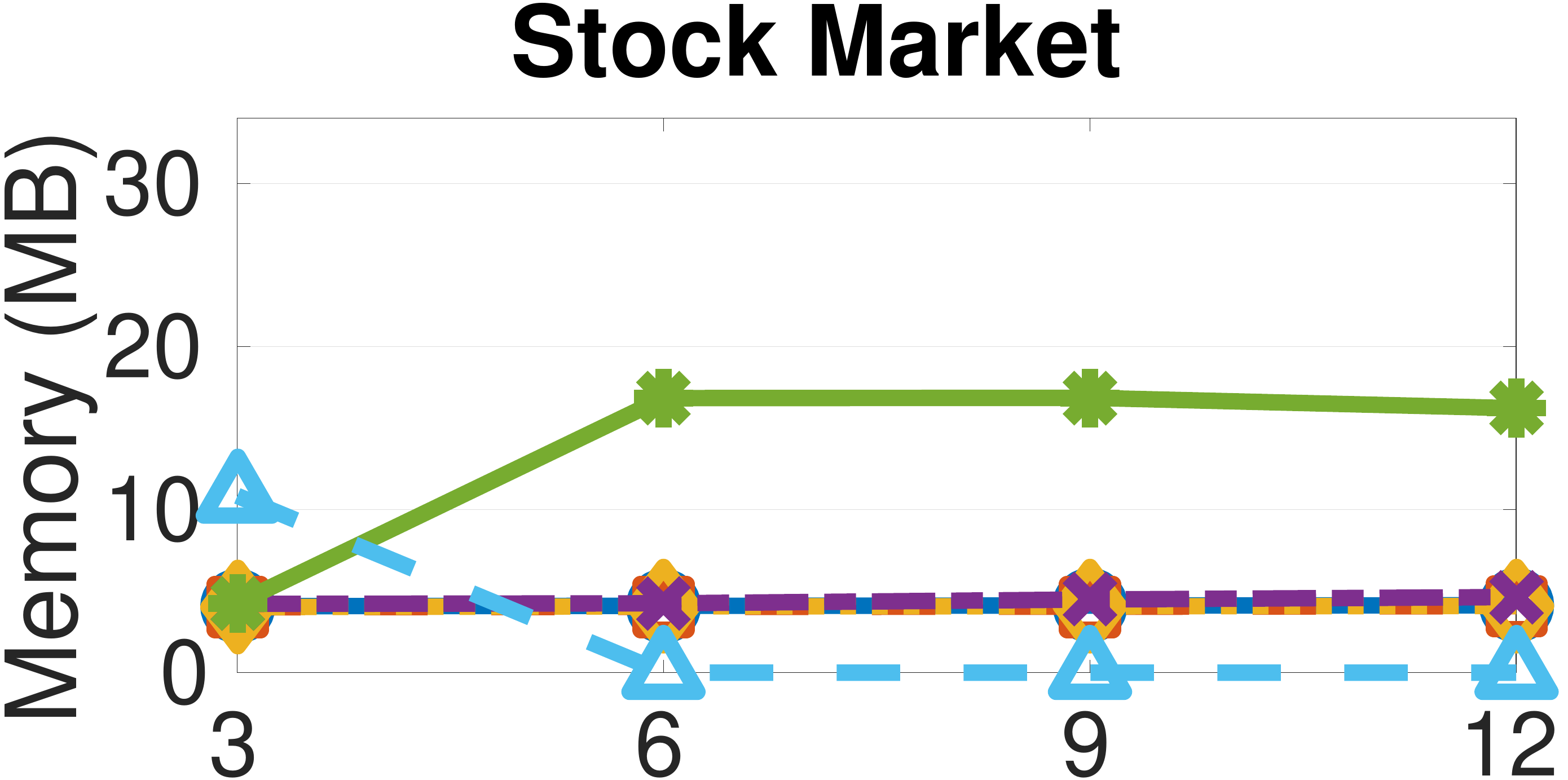}
	\label{fig:exp-kleene-nested-memory-nary-linear-stock-allwindows}
\end{subfigure}
\begin{subfigure}[t]{0.32\textwidth}
	\includegraphics[width=0.99\textwidth]{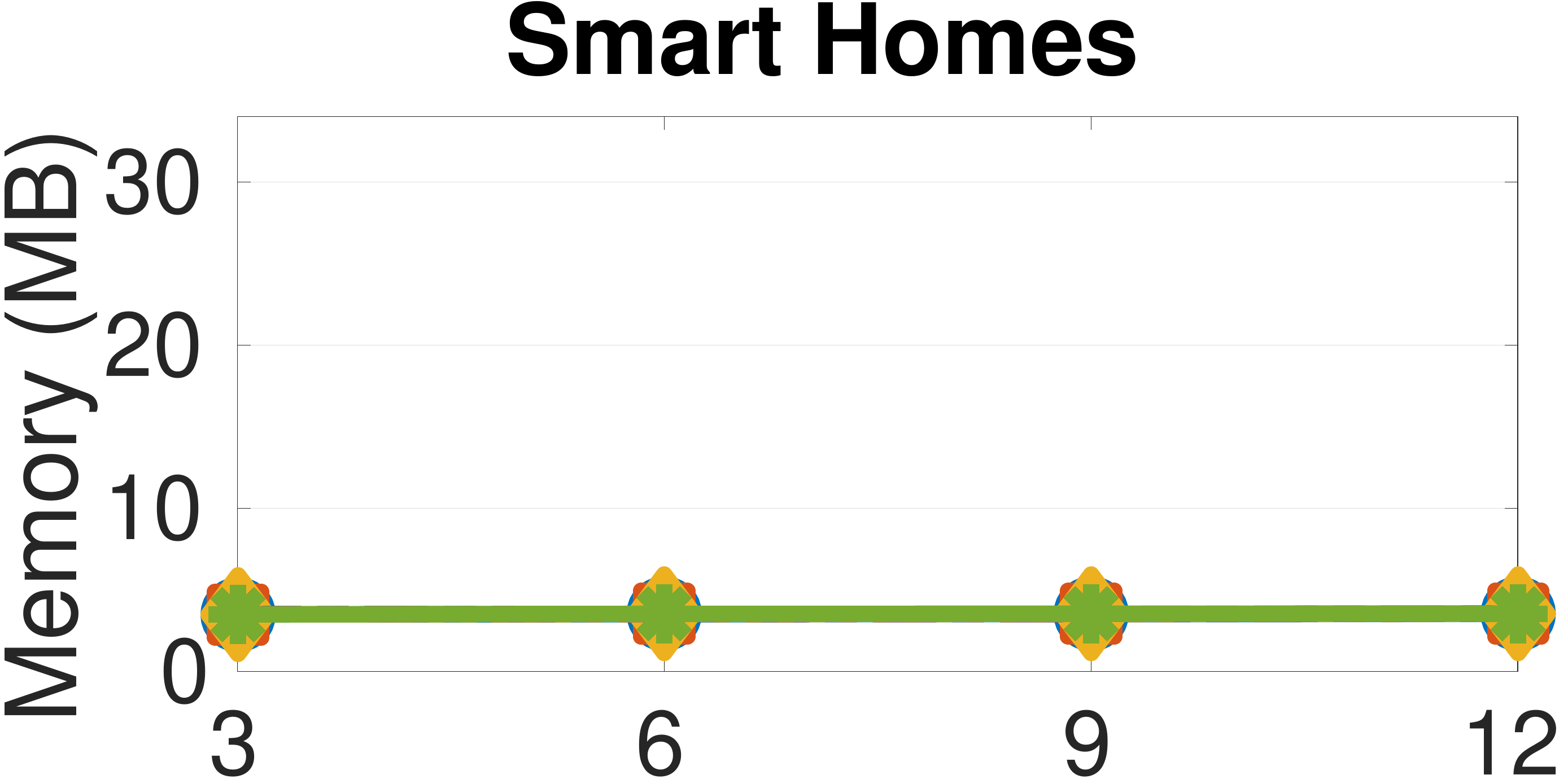}
	\label{fig:exp-kleene-nested-memory-nary-linear-smart-allwindows}
\end{subfigure}
\begin{subfigure}[t]{0.32\textwidth}
	\includegraphics[width=0.99\textwidth]{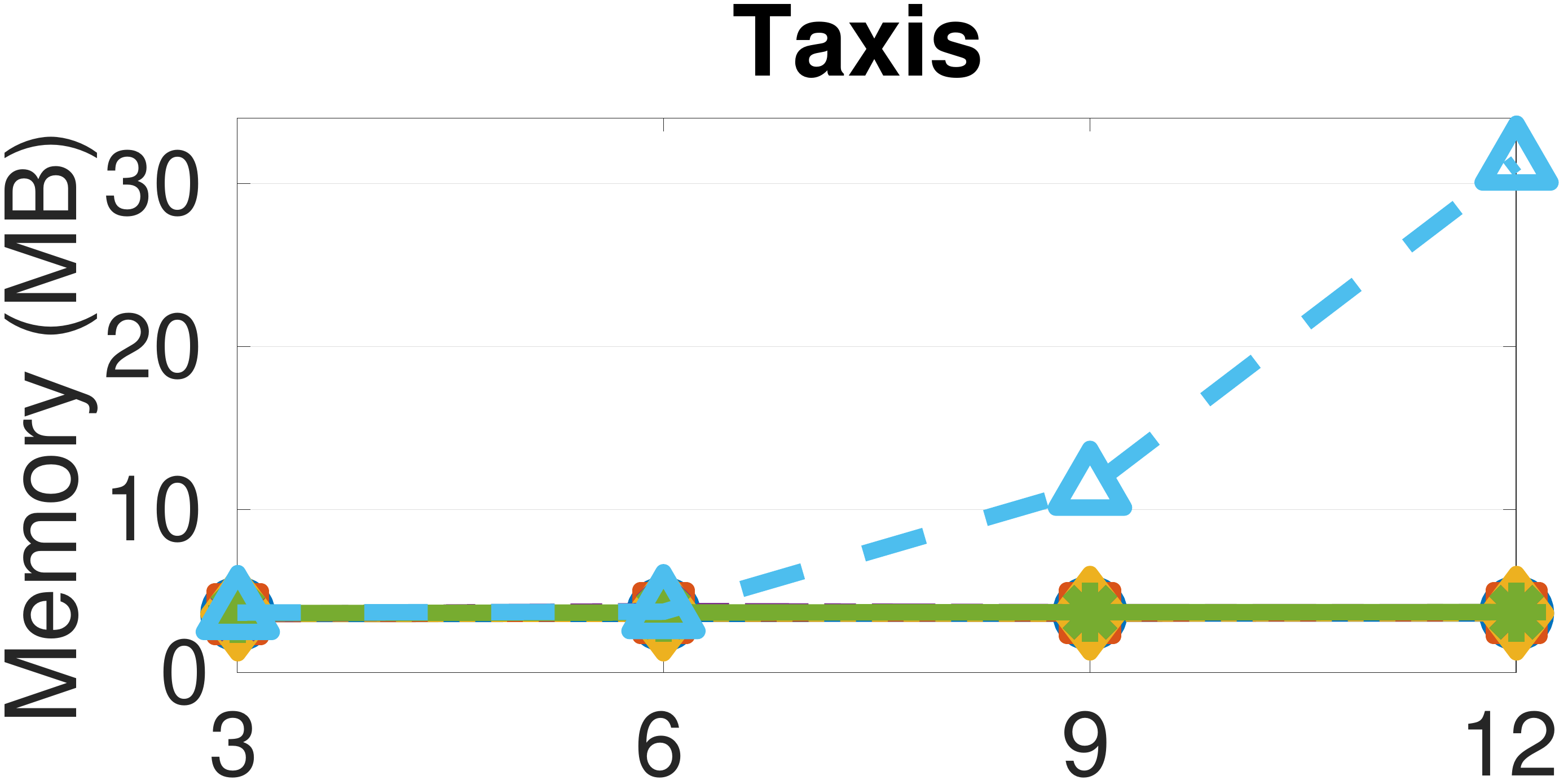}
	\label{fig:exp-kleene-nested-memory-nary-linear-taxi-allwindows}
\end{subfigure}
\caption{Throughput and memory consumption for patterns with $n$-ary predicates and nested Kleene operators as a function of pattern length for various windows. SASE, FlinkCEP and Esper are excluded because they do not support patterns with nested Kleene operators.}
\label{fig:exp-kleene-nested-nary-linear}
\end{figure}

%% file: exp_other.tex
In the last set of experiments, 
we used the stock market dataset and tested all engines against patterns with various operators.
We considered a diverse range of patterns, 
where other operators like disjunction, iteration and their combination were employed.
These operators include simple filters, disjunction and combinations of iteration and disjunction.
In particular, 
we tested 5 patterns:
\begin{enumerate}
	\item A sequential pattern starting and ending with a SELL event, and with two BUY events in between.
\begin{equation}
\label{sremo:q1}
q_{1} := \circlearrowleft((\phi_{1}(\sim) \uparrow \bullet \downarrow r_{1}) , (\phi_{2}(\sim) \uparrow \bullet) , (\phi_{3}(\sim) \uparrow \bullet) , (\phi_{4}(\sim,r_{1}) \uparrow \bullet))^{[1..w]}
\end{equation}
where 
\begin{equation*}
\begin{aligned}
\phi_{1}(x) := &\ x.\mathit{type}=SELL \wedge x.\mathit{name}=MSFT  \\
\phi_{2}(x) := &\ x.\mathit{type}=BUY \wedge x.\mathit{name}=ORCL \\
\phi_{3}(x) := &\ x.\mathit{type}=BUY \wedge x.\mathit{name}=CSCO \\
\phi_{3}(x,y) := &\ x.\mathit{type}=SELL \wedge x.\mathit{name}=AMAT \wedge x.\mathit{price} < y.\mathit{price}
\end{aligned}
\end{equation*}
	\item $q_{2}$: same as $q_{1}$, but with local thresholds on price.
	\item $q_{3}$: same as $q_{1}$, but $\phi_{2}$ now includes disjunction: $\phi_{2}(x) := (x.\mathit{type}=BUY \vee x.\mathit{type}=SELL) \wedge x.\mathit{name}=ORCL$. We also applied the same modification to $\phi_{3}$.
	\item $q_{4}$: same as $q_{3}$, but with local thresholds on price.
	\item Combining iteration and disjunction:
\begin{equation}
\label{sremo:q5}
q_{5} := \circlearrowleft((\phi_{1}(\sim) \uparrow \bullet \downarrow r_{1}) , (\phi_{2}(\sim) \uparrow \bullet)^{+} , (\phi_{3}(\sim,r_{1}) \uparrow \bullet))^{[1..w]}
\end{equation}
where 
\begin{equation*}
\begin{aligned}
\phi_{2}(x) := &\ (x.\mathit{type}=BUY \vee x.\mathit{type}=SELL) \wedge  \\
			   &\ x.\mathit{name}=QQQ \wedge x.\mathit{volume} = 4000 
\end{aligned}
\end{equation*}
\end{enumerate}

SASE can only support \sremo\ $q_{1}$ and $q_{2}$.
Therefore,
we do not show SASE results for \sremo\ $q_{3}$, $q_{4}$ and $q_{5}$.
FlinkCEP supports all 5 patterns,
but its semantics of the iteration operator are ambiguous and its results when using iteration do not match those of the other systems.
Therefore,
we do not show FlinkCEP results for \sremo\ $q_{5}$.

\begin{figure}[t]
\centering
%\begin{subfigure}[t]{0.20\textwidth}
	\includegraphics[width=0.73\textwidth]{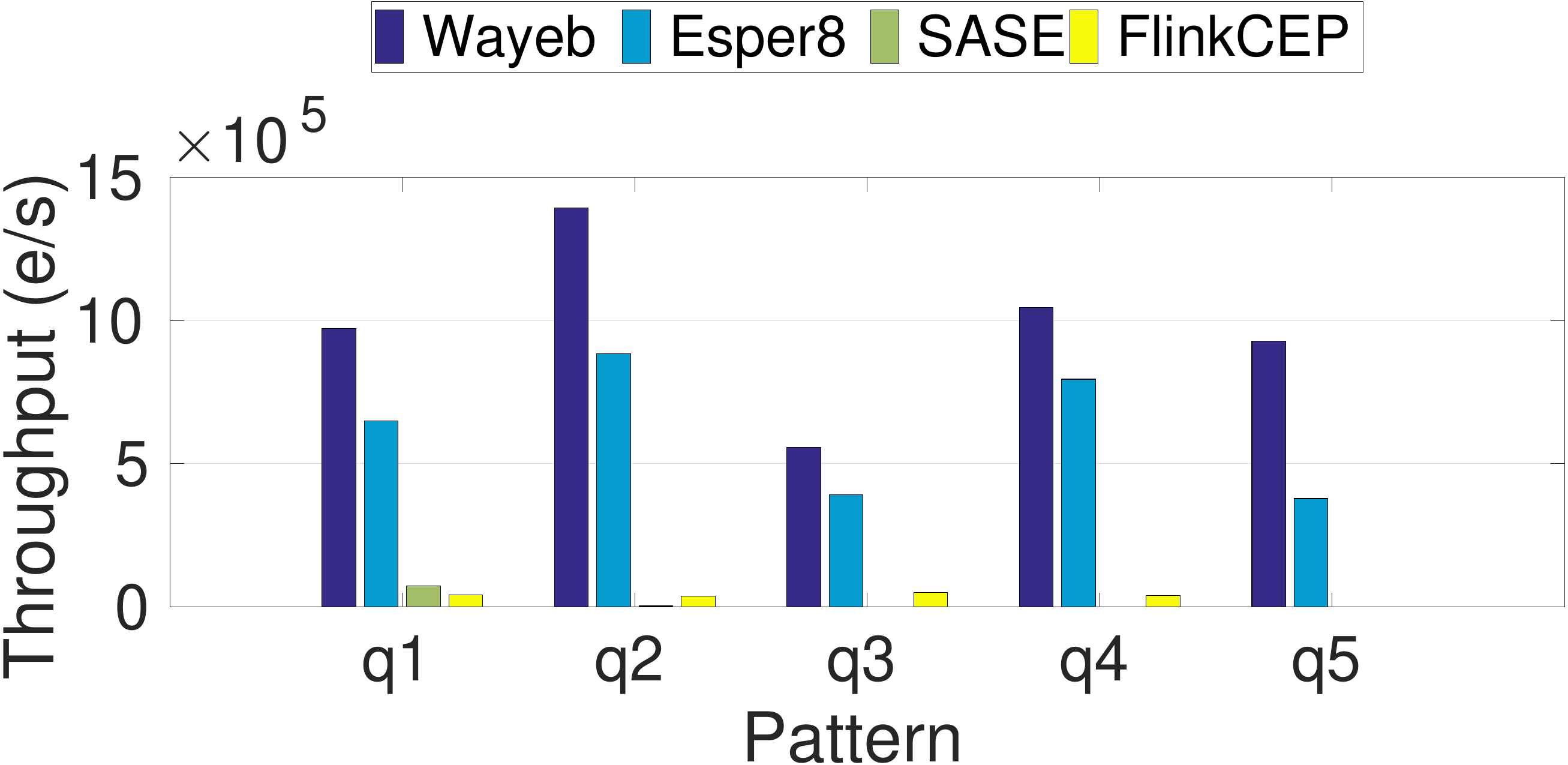}
%	\label{fig:exp-kleene-nested-throughput-nary-linear-stock-allwindows}
%\end{subfigure}
%\begin{subfigure}[t]{0.20\textwidth}
%	\includegraphics[width=0.99\textwidth]{./figures/experiments/smartHomesKleeneNested/smartHomesKleeneNested236LinearAllWindowsTimeNoConsumeNoPrintNoLimitPostProcess.pdf}
%	\label{fig:exp-kleene-nested-throughput-nary-linear-smart-allwindows}
%\end{subfigure}
\caption{Throughput for patterns with $n$-ary predicates and various operators. $w = 1000$.}
\label{fig:exp-other-nary-linear}
\end{figure}

The relevant results are shown in Figure \ref{fig:exp-other-nary-linear}.
Wayeb has the highest throughput for all patterns,
followed by Esper.
The performance for $q_{2}$ is higher than that for $q_{1}$, 
due to the presence of extra threshold filters which prune several runs.
On the other hand,
$q_{3}$ is the most demanding one, 
because it does not have any threshold filters and it includes disjunction,
thus leading to more runs being created.
$q_{4}$ rebounds to higher throughput figures,
due to the inclusion of filters.
For $q_{5}$,
Esper has its lowest performance and Wayeb its second lowest.
This is due to the presence of both iteration and disjunction.
%The combination of iteration and the \skipany\ strategy 

Finally,
we experimentally tested Wayeb's performance on the above patterns when there is no requirement for it to produce an output, 
i.e., to completely enumerate each complex event.
For this purpose,
we modified Wayeb's behavior in two ways.
First,
we disallowed any post-processing/reporting of the detected complex events.
Second,
we completely switched off Wayeb's functionality of gradually creating partial matches.
We only retained its functionality of tracking the runs to determine whether they have reached a final state.
In both cases,
Wayeb's performance remained almost unaffected.
The reason for this behavior is that we already represent runs in a very minimal way,
even when they need to carry partial matches.
This result also indicates that the main bottleneck for Wayeb lies in the actual evaluation and maintenance of the various runs and not in the production of their output.

%% file: outro.tex
\section{Discussion and Conclusions}

We presented a system for CER based on an automaton model, \srt, that can act as a computational model for patterns with $n$-ary conditions ($n \geq 1$),
which are quintessential for CER applications. 
\srt\ have nice compositional properties, 
without imposing severe restrictions on the use of operators.
Most of the standard operators in CER,
such as concatenation/sequence, union/disjunction, intersection/conjunction and Kleene-star/iteration,
may be used freely.
%This is not the case though for complement/negation.
We showed that complement may also be used and determinization is possible,
if a window operator is used,
a very common feature in CER.
We briefly discussed the complexity of the problems of non-emptiness, membership and universality.
Although the problem of membership in general is at least NP-complete,
in cases where we can use windowed, deterministic \srt,
the cost of updating the state of such an automaton after reading a single element is linear in the number of registers and conditions.
With non-deterministic \srt,
the runtime complexity becomes exponential in the size of the window and the number of Kleene-star operators.
We presented experimental results showing that our framework with \srt\ is highly expressive,
with the ability to support complex patterns with nested operators and relational constraints.
For instance, it is the only system that may express in practice nested Kleene operators.
At the same time,
we do not need to sacrifice performance for this increased expressive power.
It also outperforms other state-of-the-art engines for most patterns and workloads.

It is interesting that our system can achieve this even without any algorithmic optimizations.
Our aim for the future is to investigate our engine's optimization potential.
For example,
CORE exploits structural and computational commonalities in order to speed up the processing of matches \cite{DBLP:journals/pvldb/BucchiGQRV22}.
It employs a graph structure to represent all matches compactly and to avoid redundant predicate evaluations.
However, CORE's patterns carry only unary conditions,
i.e., relational constraints are not allowed.
This means that some graph-based optimization techniques cannot be directly transferred to \srt.
Other optimization techniques include lazy evaluation of runs \cite{DBLP:conf/ijcai/DoussonM07} and various distribution methods (see \cite{DBLP:journals/vldb/GiatrakosAADG20} for an overview).
We also aim to extend our framework towards Complex Event Forecasting.
This could potentially allow us to investigate optimization techniques based on ``branch prediction'',
i.e. based on predictions on how a run might evolve in the future.

\subsection{Variable binding and aggregates}

Another line of importance concerns the ability of \sremo\ and \srt\ to capture aggregates,
which is related to how the current semantics of \sremo\ determine (register) variable binding.
As an example,
consider the following \sremo:
\begin{equation}
\label{srem:b_star_seq_s_filter_eq_id}
e := (\mathit{TypeIsB}(\sim) \uparrow \bullet \downarrow r_{1})^{+} \cdot
	 ((\mathit{TypeIsS}(\sim) \wedge \mathit{EqualId}(\sim,r_{1})) \uparrow \bullet)
\end{equation}
One question is the following:
how exactly is variable $r_{1}$ mapped to input events?
The sub-expression inside the Kleene-plus operator requires the storage of one or more $B$ events in register $r_{1}$,
which will be compared later for identifier equality with a $S$ event. 
According to the semantics of \sremo, 
if there are multiple $B$ events before the $S$ event,
register $r_{1}$ will be overwritten.
When the $S$ event arrives,
it will be compared only with the last (stored) $B$ event.

On the other hand,
it is often useful and more intuitive to expect different semantics:
for the pattern to be satisfied, 
all $B$ events must be compared to the $S$ event.
This type of semantics would also be useful for aggregates.
For example, instead of equality,
we might require that the sum of the volumes of all $B$ events exceeds a given threshold.
Essentially,
in both cases,
we would require each $B$ event to be stored in a different register,
something which is not currently allowed by the semantics of \sremo.

In order to support aggregates,
the first step would be to modify the semantics of \sremo\ so as to have a stack of registers where an iteration operator can store events.
If \srt\ are restricted to windowed \sremo\ 
(aggregates for unbounded or windowless iteration are in any case not very meaningful),
then \srt\ could support aggregates.
This could be achieved by ``unrolling'' iterations (similarly to the way we unroll them for determinization, see technical report). 
This method would be computationally sub-optimal, 
since it would potentially require a large number of registers for storing multiple events inside an iteration operator,
but it is sufficient to show that \srt\ are expressive enough to capture aggregates of single, non-nested iterations.
In practice, appropriate optimizations would need to be employed,
which we intend to explore in the future.

For nested iterations,
we would probably need to move beyond \srt.
Consider the following \sremo: 
\begin{equation*}
\label{srem:b_star_seq_s_filter_eq_id_star}
e' := ((\mathit{TypeIsB}(\sim) \uparrow \bullet \downarrow r_{1})^{+} \cdot
	 ((\mathit{TypeIsS}(\sim) \wedge \mathit{EqualId}(\sim,r_{1})) \uparrow \bullet))^{+}
\end{equation*}
It is the same as \sremo\ \eqref{srem:b_star_seq_s_filter_eq_id},
but enclosed in an extra Kleene-star operator.
Assume we apply this pattern on the stream of Table \ref{table:example_stream_aggregates}.
When we reach the sixth event in the stream ($\mathit{index = 6}$),
should we compare its id with the id of the fifth event only,
or with the ids of events 1, 2 and 3 as well?
If the intended semantics is that of the former option 
(i.e., the comparison should not include all $B$ events,
but only those of the ``current'' repetition),
then we would probably need to extend \srt,
possibly with special flags indicating which registers belong to which repetition of an iteration operator.
We intend to explore in the future more precisely the necessary modifications of \srt\ required for handling more flexible variable binding schemes and multiple semantics for aggregates.
\begin{table}[t]
\centering
\caption{Example stream.}
\begin{tabular}{cccccccc} 
\toprule
type & B & B & B & S & B & S & ... \\ 
\midrule
id & 1 & 1 & 1 & 1 & 1 & 1 & ... \\
index & 1 & 2 & 3 & 4 & 5 & 6 & ... \\
\bottomrule
\end{tabular}
\label{table:example_stream_aggregates}
\end{table}

\subsection{Hierarchies of complex events}

Finally, we comment briefly on the issue of if and how \sremo\ and \srt\ may be used to construct complex event hierarchies,
i.e.,
automata which can process not only simple input events but the output of other automata as well.
Defining complex events in terms of other complex events is feasible at the language level,
given the compositionality of \sremo\ operators.
Each instance of a complex event definition $e$ within another definition $e'$ could be replaced by the initial definition of $e$ and then compile $e'$ into a single \srt.
However, this would be a sub-optimal solution in cases where $e$ appears in multiple other definitions,
since the sub-automaton corresponding to $e$ would be constructed multiple times and each new input event would need to be processed repeatedly by all these copies. 
This solution might not even be possible in a distributed CER setting where the results of a sub-automaton need to be sent to another automaton in a different location which does not have access to the original input events.
Constructing hierarchies properly is non-trivial and raises several issues concerning the semantics of operators,
e.g., the associativity of sequence \cite{DBLP:conf/pods/WhiteRGD07}.
Hence,
although we have not currently implemented a mechanism for constructing hierarchies,
it is worth noting that: 
a) our language is compositional, 
a property which paves the way for a proper treatment of hierarchies, and
b) our system currently produces complex events as sets of indices, 
i.e., the complete history of a match. 
As shown in \cite{DBLP:conf/pods/WhiteRGD07}, 
this is the only model with the necessary properties to avoid all temporal issues with hierarchies.
In the future, 
we will investigate how we can extend \srt\ so that they can handle complete histories as input
(presumably at the upper levels of a hierarchy).

%% file: appendix.tex
\section{Appendix}

\subsection{Proof of Theorem \ref{theorem:matchesInduceLanguage}}
\label{sec:proof:matchesInduceLanguage}
\input{proofs_matchesInduceLanguage}

\subsection{Proof of Theorem \ref{theorem:sremo2srt}}
\label{sec:proof:sremo2srt}
\input{proofs_sremo2srt}

\subsection{Proof of Lemma \ref{lemma:epsilon}}
\label{sec:proof:epsilon}
\input{proofs_epsilon}

\subsection{Proof of Theorem \ref{theorem:closure}}
\label{sec:proof:closure}
\input{proofs_closure}

\subsection{Proof of Theorem \ref{theorem:complement}}
\label{sec:proof:complement}
\input{proofs_complement}

\subsection{Proof of Theorem \ref{theorem:determinization}}
\label{sec:proof:determinization}
\input{proofs_determinization}

\subsection{Proof of Theorem \ref{theorem:wsremo2dsrt}}
\label{sec:proof:wsremo2dsrt}
\input{proofs_wsrem2dsra}

\subsection{Proof of Corollary \ref{corollary:wsra_complement}}
\label{sec:proof:wsra_complement}
\input{proofs_wsra_complement}

%% file: proofs_matchesInduceLanguage.tex
\begin{theorem*}
Let $e,e'$ be two \sremo. 
If, for every string $S$, $\mathit{Match}(e,S) = \mathit{Match}(e',S)$,
then $\mathit{Lang}(e) = \mathit{Lang}(e')$.
\end{theorem*}

\begin{proof}
We need to show that $\forall S, S \in \mathit{Lang}(e) \Leftrightarrow S \in \mathit{Lang}(e')$.
First, assume that $S \in \mathit{Lang}(e)$.
We also know that $\mathit{Match}(e,S) = \mathit{Match}(e',S)$.
Thus, from the definition of matches (see Definition \ref{definition:language_matches_sremo}),
it follows that $(e,S,M,\sharp) \vdash v$ for some $M$ and some $v$.
It also holds that $(e',S,M,\sharp) \vdash v'$.
Then, by definition (see again Definition \ref{definition:language_matches_sremo}),
$S \in Lang(e')$.
We have thus proven that $S \in \mathit{Lang}(e) \Rightarrow S \in \mathit{Lang}(e')$.
With a similar reasoning,
we can also prove that $S \in \mathit{Lang}(e') \Rightarrow S \in \mathit{Lang}(e)$.
\end{proof}

%% file: proofs_sremo2srt.tex
\begin{theorem*}
For every \sremo\ $e$ there exists an equivalent \srt\ $T$, i.e., a \srt\ such that $\mathit{Lang}(e) = \mathit{Lang}(T)$ and $\mathit{Match}(e,S)=\mathit{Match}(T,S)$ for every string $S$.
\end{theorem*}

\begin{proof}

We only need to prove that $\mathit{Match}(e,S)=\mathit{Match}(T,S)$ for every string $S$.
$\mathit{Lang}(e) = \mathit{Lang}(T)$ then follows immediately from Theorem \ref{theorem:matchesInduceLanguage}.

For a \sremo\ $e$, a string $S$ and valuations $v$, $v'$,
let $\mathcal{M}(e,S,v,v')$ denote all matches $M$ such that $(e,S,M,v) \vdash v'$.
Similarly, for a \srt\ $T$,
let $\mathcal{M}(T,S,v,v')$ denote all matches $M$ such that
$M \in Match(\varrho)$ where $\varrho \in Run_{f}(T,S)$
and 
$\varrho = [1,q_{1},v_{1}] \overset{\delta_{1}}{\rightarrow} [2,q_{2},v_{2}] \overset{\delta_{2}}{\rightarrow} \cdots \overset{\delta_{n}}{\rightarrow} [n,q_{n+1},v_{n+1}]$,
with $v_{1} = v$ and $v_{n+1} = v'$. 
For every possible \sremo\ $e$,
we will construct a corresponding \srt\ $T$ and then prove either that $\mathit{Match}(e,S) = \mathit{Match}(T,S)$ or that $\mathcal{M}(e,S,v,v') = \mathcal{M}(T,S,v,v')$ for every string $S$.
The latter implies that $\mathcal{M}(e,S,\sharp,v'') = \mathcal{M}(T,S,\sharp,v'')$ for some valuation $v''$ or equivalently $\mathit{Match}(e,S) = \mathit{Match}(T,S)$,
which is our goal.
The proof is inductive.
We prove directly the base cases for the simple expressions $e := \emptyset$,
$e := \epsilon$, $e := \phi = R(x_{1}, \cdots, x_{n})$ and $e := \phi = R(x_{1}, \cdots, x_{n}) \downarrow w$.
For the complex expression $e := e_{1} \cdot e_{2}$, $e := e_{1} + e_{2}$ and $e' = e^{*}$,
we use as an inductive hypothesis that our target result hods for the sub-expressions and then prove that it also holds for the top expression.
For example, 
for $e := e_{1} \cdot e_{2}$,
we assume that $\mathcal{M}(e_{1},S_{1},v,v'') = \mathcal{M}(T_{1},S_{1},v,v'')$ and that $\mathcal{M}(e_{2},S_{2},v'',v') = \mathcal{M}(T_{2},S_{2},v'',v')$.

We must be careful, however, with the valuations.
If, for example, $v$ applies to the \srt\ $T$,
does it also apply to the sub-automaton $T_{1}$,
if $T$ and $T_{1}$ have different registers?
We can avoid this problem and make all valuations compatible 
(i.e., having the same domain as functions)
by fixing the registers for all expressions and sub-expressions.
We can estimate the registers that we need for a top expression $e$ by scanning its conditions and write operations.
Let $\mathit{reg}(e)$ be a function applied to a \srem\ $e$.
We define it as follows:
\begin{equation}
\mathit{reg}(e) =  
  \begin{cases}
    \emptyset & \quad \text{if } e = \emptyset   \\
    \emptyset & \quad \text{if } e = \epsilon \\
    \{ x_{1} \} \cup \cdots \cup \{ x_{n} \} \cup \{ w \} & \text{if } e = R(x_{1},\cdots, x_{n}) \downarrow w \\
    \mathit{reg}(e_{1}) \cup \mathit{reg}(e_{2}) & \text{if } e = e_{1} \cdot e_{2} \\
    \mathit{reg}(e_{1}) \cup \mathit{reg}(e_{2}) & \text{if } e = e_{1} + e_{2} \\
    \mathit{reg}(e_{1}) & \text{if } e = (e_{1})^{*} 
  \end{cases}
\end{equation}
For our proofs that follow,
we first apply this function to the top expression $e$ to obtain $R_{top} = \mathit{reg}(e)$ and we use $R_{top}$ as the set of registers for all automata and sub-automata.
All valuations can thus be compared without any difficulties,
since they will have the same domain $R_{top}$.

\begin{figure}
\centering
\begin{subfigure}[t]{0.45\textwidth}
	\includegraphics[width=0.95\textwidth]{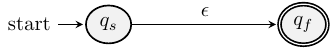}
	\caption{Base case of a single $\epsilon$ condition, $e := \epsilon$.}
	\label{fig:sremo2srt:epsilon}
\end{subfigure}
\begin{subfigure}[t]{0.45\textwidth}
	\includegraphics[width=0.95\textwidth]{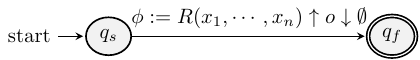}
	\caption{Base case of a single condition, $e := \phi \uparrow o = R(x_{1}, \cdots, x_{n}) \uparrow o$.}
	\label{fig:sremo2srt:phi}
\end{subfigure}\\
\begin{subfigure}[t]{0.55\textwidth}
	\includegraphics[width=0.95\textwidth]{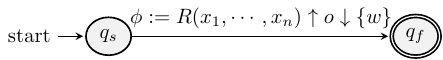}
	\caption{Base case of a single condition with a write register, $e := \phi \uparrow o \downarrow W = R(x_{1}, \cdots, x_{n}) \uparrow o \downarrow \{ w \}$.}
	\label{fig:sremo2srt:phiW}
\end{subfigure}\\
\begin{subfigure}[t]{0.65\textwidth}
	\includegraphics[width=0.95\textwidth]{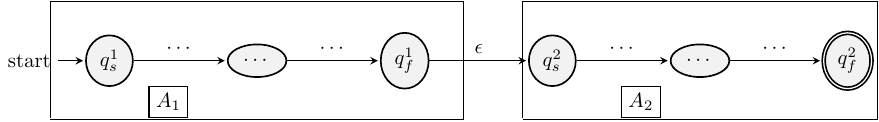}
	\caption{Concatenation. $e = e_{1} \cdot e_{2}$.}
	\label{fig:sremo2srt:seq}
\end{subfigure}\\
\begin{subfigure}[t]{0.65\textwidth}
	\includegraphics[width=0.95\textwidth]{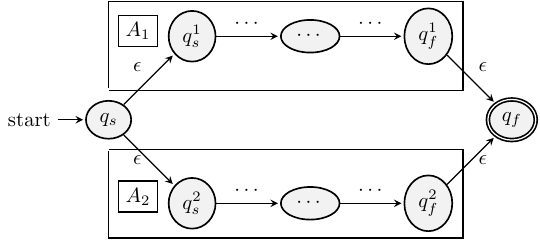}
	\caption{OR. $e = e_{1} + e_{2}$.}
	\label{fig:sremo2srt:or}
\end{subfigure}\\
\begin{subfigure}[t]{0.65\textwidth}
	\includegraphics[width=0.95\textwidth]{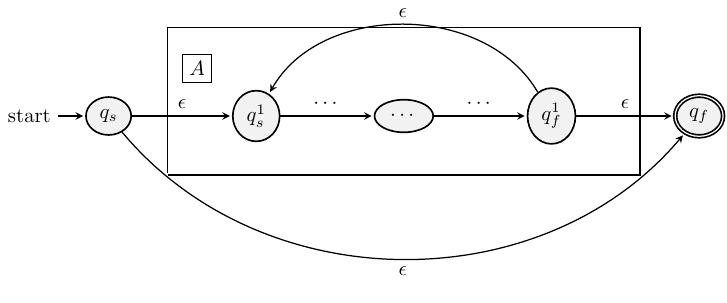}
	\caption{Iteration. $e^{'} = e^{*}$.}
	\label{fig:sremo2srt:iter}
\end{subfigure}
\caption{The cases for constructing a \srt\ from a \sremo.}
\label{fig:sremo2srt}
\end{figure}

\textbf{Assume $e := \epsilon$.}
We know that $\mathit{Match}(e,S)=\emptyset$.
We can then construct a \srt\ $T = (Q,q_{s},Q_{f},R,\Delta)$
where $Q = \{q_{s}, q_{f}\}$, $Q_{f}=\{q_{f}\}$, $R=R_{top}$, $\Delta=\{ \delta \}$ and $\delta = q_{s},\epsilon \uparrow \otimes \downarrow \emptyset \rightarrow q_{f}$.
See Figure \ref{fig:sremo2srt:epsilon}.
It is obvious that $T$ accepts only the empty string since there is only one path that leads to the final state and this path goes through an $\epsilon$ transition.
No elements are marked.
Thus $\mathit{Match}(T,S)=\emptyset$.

\textbf{Assume $e := \phi = R(x_{1}, \cdots, x_{n}) \uparrow o$, where $\phi$ is a condition and all $x_{i}$ belong to a set of register variables $\{r_{1},\cdots,r_{k}\}$.}
We construct the following \srt\ $T=(Q,q_{s},Q_{f},R,\Delta)$,
where $Q = \{q_{s}, q_{f}\}$, $Q_{f}=\{q_{f}\}$, $R=R_{top}$, $\Delta=\{ \delta \}$ and $\delta = q_{s},\phi \uparrow o \downarrow \emptyset \rightarrow q_{f}$.
See Figure \ref{fig:sremo2srt:phi}.

We first prove $M \in \mathcal{M}(e,S,v,v') \Rightarrow M \in \mathcal{M}(T,S,v,v')$ for a match $M$ and any string $S$.
We also first assume that $o = \bullet$.
It is obvious that $S$ must be composed of a single element, i.e., $S=t_{1}$.
Also, $M = \{1\}$. 
%If $\phi$ is a condition with no arguments,
%this means that it corresponds to a unary relation $R(x_{1})$.
%If $\phi$ has argument,
%then it corresponds to an n-ary relation $R(x_{1},\cdots,x_{n})$.
%In both cases,
Since $S=t_{1}$ is accepted by $e$ starting from the valuation $v$,
this means that $(\phi,S,M,v) \vdash v'$, with $v'=v$,
according to the second case of Definition \ref{definition:sremo_semantics}.
Thus $(t_{1},v) \models \phi$.
%(equivalently $t_{1} \in R(t_{1})$ or $t_{1} \in R(t_{1},v(x_{1}),\cdots,v(x_{n}))$ depending on the number of arguments) .
This then implies that the second case in the definition of a successor configuration 
(see Definition \ref{definition:configuration}) holds for our constructed automaton $T$.
As a result,
$T$, upon reading $S$, moves to its final state $q_{f}$, marks $t_{1}$ and accepts $S$.
This move does not change the valuation, thus $v'=v$.
We have thus proven that $M = \{1\} \in \mathcal{M}(T,S,v,v')$.
Similarly, 
if $o = \otimes$,
we can prove that $M = \emptyset \in \mathcal{M}(T,S,v,v')$.

The inverse direction,
$M \in \mathcal{M}(T,S,v,v') \Rightarrow M \in \mathcal{M}(e,S,v,v')$,
can be proven in a similar manner.

\textbf{Assume $e := \phi = R(x_{1}, \cdots, x_{n}) \uparrow o \downarrow w$, 
where $\phi$ is a condition, 
all $x_{i}$ belong to a set of register variables $\{r_{1},\cdots,r_{k}\}$ and $w$ a write register (not necessarily one of $r_{i}$).}
We construct the following \srt\ $T=(Q,q_{s},Q_{f},R,\Delta)$,
where $Q = \{q_{s}, q_{f}\}$, $Q_{f}=\{q_{f}\}$, $R=R_{top}$, $\Delta=\{ \delta \}$ and $\delta = q_{s},\phi \uparrow o \downarrow \{ w \} \rightarrow q_{f}$.
See Figure \ref{fig:sremo2srt:phiW}.

The proof is essentially the same as that for the previous case.
The only difference is that we need to use the third case from the definition of successor configurations (Definition \ref{definition:configuration}).
This means that $v' = v[w \leftarrow t_{1}]$.
If $w \in R$, 
then $t_{1}$ is stored in $w$ and $v'(w) = t_{1}$.
Otherwise,
$v'$ remains the same as $v$.

\textbf{Assume $e := e_{1} \cdot e_{2}$, where $e_{1}$ and $e_{2}$ are \sremo.}
We first construct $T_{1}$ and $T_{2}$, 
the \srt\ for $e_{1}$ and $e_{2}$ respectively.
We construct the following \srt\ $A=(Q,q_{s},Q_{f},R,\Delta)$,
where $Q = T_{1}.Q \cup T_{2}.Q$, $q_{s}=T_{1}.q_{s}$, $Q_{f}=\{T_{2}.q_{f}\}$, $R=R_{top}$, $\Delta= T_{1}.\Delta \cup T_{2}.\Delta \cup \{ \delta \}$ and $\delta = T_{1}.q_{f},\epsilon \rightarrow T_{2}.q_{s}$.
See Figure \ref{fig:sremo2srt:seq}.
We thus simply connect $T_{1}$ and $T_{2}$ with an $\epsilon$ transition.
Notice that $T_{1}.R$ and $T_{2}.R$ may overlap.
Their union retains only one copy of each register,
if a register appears in both of them.

We first prove $M \in \mathcal{M}(e,S,v,v') \Rightarrow M \in \mathcal{M}(T,S,v,v')$ for a match $M$ and any string $S$.
Since $M \in \mathcal{M}(e,S,v,v')$,
$S$ can be broken into two sub-strings $S_{1}$ and $S_{2}$ and $M$ into two subsets $M_{1}$ and $M_{2}$ such that
$S = S_{1} \cdot S_{2}$, $M = M_{1} \cdot M_{2}$, $(e_{1},S_{1},M_{1},v) \vdash v''$ and $(e_{2},S_{2},M_{2},v'') \vdash v'$.
This is equivalent to $M_{1} \in \mathcal{M}(e_{1},S_{1},v,v'')$ and $M_{2} \in \mathcal{L}(e_{2},S_{2},v'',v')$.
From the induction hypothesis 
(i.e., that what we want to prove holds for the sub-expressions $e_{1}$, $e_{2}$ and their automata $T_{1}$, $T_{2}$)
it follows that $M_{1} \in \mathcal{M}(T_{1},S_{1},v,v'')$ and $M_{2} \in \mathcal{M}(T_{2},S_{2},v'',v')$.
Notice that if $T_{1}$ and $t_{2}$ have different sets of registers,
we can always expand $T_{1}.R$ and $T_{2}.R$ to their union,
without affecting in any way the behavior of the automata.
Now, let $l_{1} = \lvert S_{1} \rvert$ and $l_{2} = \lvert S_{2} \rvert$.
From $M_{1} \in \mathcal{M}(T_{1},S_{1},v,v'')$ it follows that there exists an accepting run $\varrho_{1}$ of $T_{1}$ over $S_{1}$ such that
$\varrho_{1}=[1,T_{1}.q_{s},v] \rightarrow \cdots \rightarrow [l_{1}+1,T_{1}.q_{f},v'']$.
Similarly,
from $M_{2} \in \mathcal{M}(T_{2},S_{2},v'',v')$ it follows that there exists an accepting run $\varrho_{2}$ of $T_{2}$ over $S_{2}$ such that
$\varrho_{2}=[1,T_{2}.q_{s},v''] \rightarrow \cdots \rightarrow [l_{2}+1,T_{2}.q_{f},v']$.
Let's construct a run by connecting $\varrho_{1}$ and $\varrho_{2}$ with an $\epsilon$ transition:
$\varrho = [1,T_{1}.q_{s},v] \rightarrow \cdots \rightarrow [l_{1}+1,T_{1}.q_{f},v''] \overset{T_{1}.q_{f},\epsilon \rightarrow T_{2}.q_{s}}{\rightarrow} [l_{1}+2,T_{2}.q_{s},v''] \rightarrow \cdots \rightarrow [l_{1}+l_{2}+1,T_{2}.q_{f},v']$.
We can see that this is indeed an accepting run of $T$ and that $Match(\varrho) = M$.
Thus $M \in \mathcal{M}(T,S,v,v')$.

The inverse direction,
$M \in \mathcal{M}(T,S,v,v') \Rightarrow M \in \mathcal{M}(e,S,v,v')$,
can be proven in a similar manner.
Since $M \in \mathcal{M}(T,S,v,v')$,
there exists an accepting run $\varrho$ of $T$ over $S$.
By the construction of $T$, however,
this run must be in the form $\varrho = \varrho_{1} \overset{\epsilon}{\rightarrow} \varrho_{2}$, 
with $\varrho_{1}$ being an accepting run of $T_{1}$ over a string $S_{1}$ and $\varrho_{2}$ an accepting run of $T_{2}$ over $S_{2}$, 
where $S = S_{1} \cdot S_{2}$.
We then use the induction hypothesis to prove that
$M_{1} = Match(\varrho_{1}) \in \mathcal{M}(e_{1},S_{1},v,v'')$ and $M_{2} = Match(\varrho_{2}) \in \mathcal{L}(e_{2},S_{2},v'',v')$
and finally that $M = M_{1} \cdot M_{2} \in \mathcal{M}(e,S,v,v')$.

\textbf{Assume $e := e_{1} + e_{2}$, where $e_{1}$ and $e_{2}$ are \sremo.}
We first construct $T_{1}$ and $T_{2}$, 
the \srt\ for $e_{1}$ and $e_{2}$ respectively.
We construct the following \srt\ $T=(Q,q_{s},Q_{f},R,\Delta)$,
where $Q = T_{1}.Q \cup T_{2}.Q \cup \{ q_{s}, q_{f} \}$, $Q_{f}=\{ q_{f} \}$, $R=R_{top}$, $\Delta= T_{1}.\Delta \cup T_{2}.\Delta \cup \{ \delta_{s,1}, \delta_{s,2}, \delta_{1,f}, \delta_{2,f} \}$ and $\delta_{s,1} = q_{s},\epsilon \rightarrow T_{1}.q_{s}$, $\delta_{s,2} = q_{s},\epsilon \rightarrow T_{2}.q_{s}$, $\delta_{1,f} = T_{1}.q_{f},\epsilon \rightarrow q_{f}$, $\delta_{2,f} = T_{2}.q_{f},\epsilon \rightarrow q_{f}$.
See Figure \ref{fig:sremo2srt:or}.
We thus create a new state,
$q_{s}$,
acting as the start state and connect it through $\epsilon$ transitions to the start states of $T_{1}$ and $T_{2}$.
We also create a new final state and connect to it the final states of $T_{1}$ and $T_{2}$.
Again, $T_{1}.R$ and $T_{2}.R$ may overlap.
Their union retains only one copy of each register,
if a register appears in both of them.

It is easy to prove that $M \in \mathcal{M}(e,S,v,v') \Rightarrow M \in \mathcal{M}(T,S,v,v')$ for a match $M$ and any string $S$.
If $(e_{1},S,M,v) \vdash v'$,
this implies that $e_{1}$ is accepted by $T_{1}$ and $M$ is a match of $e_{1}$ and $T_{1}$.
$M$ is thus also a match of $T$.
Similarly if $(e_{2},S,M,v) \vdash v'$ for $T_{2}$.
The inverse direction has a similar proof.

\textbf{Assume $e' := e^{*}$, where $e$ is a \srem.}
We construct a new \sra\ $T'$ as shown in Figure \ref{fig:sremo2srt:iter}.
We first construct the \sra\ for $e$, $T$.
We create a new final and a new start state.
We connect the new start state to the old start and to the new final.
We connect the old final to the new final and the old start.
$R$ is again $R_{top}$.

We first prove that $M \in \mathcal{M}(e',S,v,v') \Rightarrow M \in \mathcal{M}(T',S,v,v')$ for a match $M$ and any string $S$.
Since $M \in \mathcal{M}(e',S,v,v')$,
$S = S_{1} \cdot S'$ and $M = M_{1} \cdot M'$ such that $(e,S_{1},M_{1},v) \vdash v''$ and $(e^{*},S',M',v'') \vdash v'$.
Equivalently,
this implies that
\begin{itemize}
	\item $(e,S_{1},M_{1},v) \vdash v_{1}$ and $M_{1} = Match(\varrho_{1}) \in \mathcal{M}(T,S_{1},v,v_{1})$
	\item $(e,S_{2},M_{2},v_{1}) \vdash v_{2}$ and $M_{2} = Match(\varrho_{2}) \in \mathcal{M}(T,S_{2},v_{1},v_{2})$
	\item $(e,S_{3},M_{3},v_{2}) \vdash v_{3}$ and $M_{3} = Match(\varrho_{3})\in \mathcal{M}(T,S_{3},v_{2},v_{3})$ etc
	\item  until $(e,S_{n},M_{n},v_{n-1}) \vdash v_{n}$ and $M_{n} = Match(\varrho_{n}) \in \mathcal{M}(T,S_{n},v_{n-1},v_{n})$,
\end{itemize}
where $v_{n} = v'$.
We can then construct the run $\varrho = \varrho_{1} \overset{\epsilon}{\rightarrow} \varrho_{2} \overset{\epsilon}{\rightarrow} \cdots \overset{\epsilon}{\rightarrow} \varrho_{n}$.
It is easy to see that $\varrho$ is an accepting run of $T'$ and that $M=M_{1} \cdot M_{2} \cdot \cdots \cdot M_{n}$ is a match of $T'$.
Similarly for the inverse direction.
\end{proof}

%% file: proofs_epsilon.tex
\begin{lemma*}
For every \srt\ $T_{\epsilon}$ with $\epsilon$ transitions there exists an equivalent  \srt\ $T_{\notin}$ without $\epsilon$ transitions, i.e., a \srt\ such that $\mathit{Match}(T_{\epsilon},S)=\mathit{Match}(T_{\notin},S)$ for every string $S$.
\end{lemma*}

\begin{proof}
\input{algorithms_epsilon}
We first give the algorithm.
See Algorithm \ref{algorithm:epsilon}.
Note that in this algorithm,
the function $\mathit{Enclose}$ is the usual function for $\epsilon$-enclosure in standard automata theory and we will not repeat it here (see \cite{DBLP:books/daglib/0016921}).
Suffice it to say that,
when applied to a state $q$ (or set of states $\{q_{i}\}$),
it returns all the states we can reach from $q$ (or all $q_{i}$)
by following only $\epsilon$-transitions.
It is also worth noting that the algorithm does not create the power-set of states and then connects them through transitions.
It creates those subsets it needs by ``forward-looking'' for what is necessary,
but it is equivalent to the power-set construction algorithm.
We will prove that $\mathit{Match}(T_{\epsilon},S)=\mathit{Match}(T_{\notin},S)$ for every string $S$ or, equivalently, 
that
$M \in \mathcal{M}(T_{\epsilon},S,\sharp,v') \Leftrightarrow M \in \mathcal{M}(T_{\notin},S,\sharp,v')$ for a match $M$ and any string $S$.

We first prove the direction $M \in \mathcal{M}(T_{\epsilon},S,v,v') \Rightarrow M \in \mathcal{M}(T_{\notin},S,v,v')$.
The other direction can be proven similarly.
Let $\varrho_{\epsilon}$ denote an accepting run of $A_{\epsilon}$ over $S$,
where $k = \lvert S \rvert$ is the length of $S$.
\begin{equation}
\label{run:epsilon}
\begin{aligned}
\varrho_{\epsilon} = & [1,q_{\epsilon,1}=q_{\epsilon,s},v_{\epsilon,1}=\sharp] \overset{\epsilon}{\rightarrow} [\cdots] \overset{\epsilon}{\rightarrow} \cdots & \text{sub-run 1}  \\
 & \overset{\delta_{\epsilon,1}}{\rightarrow} [2,q_{\epsilon,2},v_{\epsilon,2}] \overset{\epsilon}{\rightarrow} [\cdots] \overset{\epsilon}{\rightarrow} \cdots  & \text{sub-run 2}  \\
 & \cdots & \\
 & \overset{\delta_{\epsilon,i-1}}{\rightarrow} [i,q_{\epsilon,i},v_{\epsilon,i}] \overset{\epsilon}{\rightarrow} [\cdots] \overset{\epsilon}{\rightarrow} [i',q_{\epsilon,i'},v_{\epsilon,i'}]  & \text{sub-run i}  \\
 & \overset{\delta_{\epsilon,i}}{\rightarrow} [i+1,q_{\epsilon,i+1},v_{\epsilon,i+1}] \overset{\epsilon}{\rightarrow} [\cdots] \overset{\epsilon}{\rightarrow} \cdots  & \text{sub-run i+1}  \\
 & \cdots & \\
 & \overset{\delta_{\epsilon,k}}{\rightarrow} [k+1,q_{\epsilon,k+1} \in Q_{\epsilon,f},v_{\epsilon,k+1}]  & \text{sub-run k+1}  
\end{aligned}
\end{equation}
Let $\varrho_{\notin}$ denote a run of $A_{\notin}$ over $S$.
\begin{equation}
\label{run:noepsilon}
\begin{aligned}
\varrho_{\notin} = & [1,q_{\notin,1}=q_{\notin,s},v_{\notin,1}=\sharp] & \text{sub-run 1} \\
 & \overset{\delta_{\notin,1}}{\rightarrow} [2,q_{\notin,2},v_{\notin,2}] & \text{sub-run 2} \\
 & \cdots & \\
 & \overset{\delta_{\notin,i-1}}{\rightarrow} [i,q_{\notin,i},v_{\notin,i}] & \text{sub-run i}\\
 & \overset{\delta_{\notin,i}}{\rightarrow} [i+1,q_{\notin,i+1},v_{\notin,i+1}] & \text{sub-run i+1}\\
 & \cdots & \\
 & \overset{\delta_{\notin,k}}{\rightarrow} [k+1,q_{\notin,k+1},v_{\notin,k+1}] & \text{sub-run k+1}
\end{aligned}
\end{equation}
$\varrho^{\notin}$ necessarily follows $k$ transitions,
since it does not have any $\epsilon$-transitions.
On the other hand, 
$\varrho^{\epsilon}$ may follow more than $k$ transitions ($j \geq k$),
because several $\epsilon$ transitions may intervene between ``actual'', non-$\epsilon$ transitions, 
as shown in Run \eqref{run:epsilon}.
The number of non-$\epsilon$ transitions is still $k$.
$\varrho_{\epsilon}$ is thus necessarily composed of $k$ ``sub-runs'',
where the first configuration of each sub-run is reached via a non-$\epsilon$ transition, followed by a sequence of 0 or more $\epsilon$ transitions.
Each line in Run \eqref{run:epsilon} is such a sub-run.
We can also split $\varrho_{\notin}$ in sub-runs,
but in this case each such sub-run will be simply composed of a single configuration.
See Run \eqref{run:noepsilon}.

We will prove the following.
For each sub-run $i$ of $\varrho_{\epsilon}$,
it holds that:
\begin{enumerate}
	\item $q_{\epsilon,i} \in q_{\notin,i}$. 
	\item $v_{\epsilon,i} = v_{\notin,i}$, i.e., $T_{\epsilon}$ and $T_{\notin}$ have the same register contents at each $i$.
	\item $\delta_{\epsilon,i}.o = \delta_{\notin,i}.o$.
\end{enumerate} 
We can prove this inductively.
We assume that the above claims hold for $i$ and then we can show that they must necessarily hold for $i+1$.
Since they are obviously true for $i=1$,
they are then true for all other $i$ as well.
Thus, $q_{\notin,k+1} \in Q_{\notin,f}$,
$\varrho_{\notin}$ is an accepting run as well and
$\mathit{Match}(\varrho_{\epsilon}) = \mathit{Match}(\varrho_{\notin})$.

First, notice that within each sub-run $i$ of $\varrho_{\epsilon}$,
$v_{\epsilon,i}$ remains the same,
since $\epsilon$ transitions never modify the contents of the registers.
Thus, in $\varrho_{\epsilon}$, $v_{\epsilon,i'}=v_{\epsilon,i}$.
It is also obviously true that $i'=i$,
since $\epsilon$ transitions do not read elements from $S$ and thus the automaton's head does not move.
The only thing that could possibly change is $q_{\epsilon,i'}$,
so that, in general, $q_{\epsilon,i'} \neq q_{\epsilon,i}$.
Therefore,
in Run \eqref{run:epsilon},
we move from sub-run $i$ to sub-run $i+1$ by jumping from $q_{\epsilon,i'}$ to $q_{\epsilon,i+1}$.
This implies that $\delta_{\epsilon,i}$,
connecting $q_{\epsilon,i'}$ to $q_{\epsilon,i+1}$,
is triggered when the contents of the register are those of $v_{\epsilon,i'}=v_{\epsilon,i}$.

Now, $q_{\epsilon,i'}$ belongs to the enclosure of $q_{\epsilon,i}$.
Otherwise, it would be impossible to reach it from $q_{\epsilon,i}$ by following only $\epsilon$ transitions.  
From the induction hypothesis we know that $q_{\epsilon,i} \in q_{\notin,i}$.
From the construction algorithm for $T_{\notin}$ (Algorithm \ref{algorithm:epsilon}) 
we also know that the transition $\delta_{\epsilon,i}$ also exists in $T_{\notin}$,
with $q_{\notin,i}$ as its source.
$\delta_{\notin,i}$ has the same condition/output and references the same registers as $\delta_{\epsilon,i}$.
Since $\delta_{\epsilon,i}$ is triggered with $v_{\epsilon,i'}$,
$\delta_{\notin,i}$ must also be triggered because $v_{\notin,i} = v_{\epsilon,i}$ (by the induction hypothesis) and thus $v_{\notin,i} = v_{\epsilon,i'}$.
From the construction algorithm,
we can see that $q_{\epsilon,i+1} \in q_{\notin,i+1}$.
The state $q_{\notin}$ in Algorithm \ref{algorithm:epsilon} is $q_{\notin,i}$ in Run \eqref{run:noepsilon} whereas $p_{\notin}$ is $q_{\notin,i+1}$.
$p_{\notin}$ contains all states that can be reached from $q_{\epsilon,i}$ when $\delta_{\epsilon,i}$ is triggered. 
Thus, it also contains $q_{\epsilon,i+1}$. 

The second part of the induction hypothesis is obviously true for $i+1$,
i.e., $v_{\epsilon,i+1} = v_{\notin,i+1}$,
since exactly the same registers are modified in exactly the same way by $\delta_{\epsilon,i}$ and $\delta_{\notin,i}$.

The third part is also true.
The construction algorithm ensures that $\delta_{\epsilon,i}.o = \delta_{\notin,i}.o$.
 
Therefore, $q_{\epsilon,k+1} \in q_{\notin,k+1}$ which implies that $q_{\notin,k+1} \in Q_{\notin,f}$ and thus $\varrho_{\notin}$ is an accepting run of $T_{\notin}$ over $S$.
Moreover, 
the transitions between sub-runs in Runs \eqref{run:epsilon} and \eqref{run:noepsilon} have marked exactly the same input elements.
Therefore,
if $M \in \mathcal{M}(T_{\epsilon},S,\sharp,v')$,
then $M \in \mathcal{M}(T_{\notin},S,\sharp,v')$.
\end{proof}

%% file: algorithms_epsilon.tex
\begin{algorithm}
%\SetAlgoNoLine
\KwIn{\srt\ $T_{\epsilon}$, possibly with $\epsilon$ transitions}
\KwOut{\srt\ $T_{\notin}$ without $\epsilon$-transitions}
$q_{\notin,s} \leftarrow Enclose(T_{\epsilon}.q_{s})$; $Q_{\notin} \leftarrow \{ q_{\notin,s} \}$; $\Delta_{\notin} \leftarrow \emptyset$\;
\eIf{$\exists q \in q_{\notin,s}: q \in T_{\epsilon}.Q_{f}$}{
	$Q_{\notin,f} \leftarrow \{ q_{\notin,s} \}$\;
}
{
	$Q_{\notin,f} \leftarrow \emptyset$\;
}
$\mathit{frontier} \leftarrow \{q_{\notin,s}\}$\;
\While{$\mathit{frontier} \neq \emptyset$}{
	$q_{\notin} \leftarrow$ pick an element from $\mathit{frontier}$\;
	$t \leftarrow$ gather all outgoing transitions (except epsilon) of all $q_{\epsilon} \in q_{\notin}$\;
	$cos \leftarrow$ find all distinct pairs of conditions and outputs from $t$\;
	\ForEach{$co \in cos$}{
		$W \leftarrow$ gather all write registers from all $t$ whose condition and output match that of $co$\;
		$p_{\notin} \leftarrow$ gather and enclose all target states from all $t$ whose condition and output match that of $co$\;
		$Q_{\notin} \leftarrow Q_{\notin} \cup \{ p_{\notin} \}$\;
		\If{$\exists q \in p_{\notin}: q \in T_{\epsilon}.Q_{f}$}{
			$Q_{\notin,f} \leftarrow Q_{\notin,f} \cup \{ p_{\notin} \}$\;
		}
		$\delta_{\notin} \leftarrow \mathit{CreateNewTransition}(q_{\notin},co.\phi \uparrow co.o \downarrow W \rightarrow p_{\notin})$\;
		$\Delta_{\notin} \leftarrow \Delta_{\notin} \cup \{\delta_{\notin}\}$\;
		$\mathit{frontier} \leftarrow \mathit{frontier} \cup \{p_{\notin}\}$\;
	}
	$\mathit{frontier} \leftarrow \mathit{frontier} \setminus \{q_{\notin}\}$\;
}
$T_{\notin} \leftarrow (Q_{\notin}, q_{\notin,s}, Q_{\notin,f}, T_{\epsilon}.R,\Delta_{\notin})$\;
$\mathtt{return}\ T_{\notin}$;
\caption{Eliminating $\epsilon$-transitions ($\mathit{EliminateEpsilon}$).}
\label{algorithm:epsilon}
\end{algorithm}

%% file: proofs_closure.tex
\begin{theorem*}
\srt\ are closed under union, intersection, concatenation and Kleene-star.
\end{theorem*}

\begin{proof}
For union, concatenation and Kleene-star the proof is essentially the proof for converting \sremo\ to \srt.
For concatenation,
if we have \srt\ $T_{1}$ and $T_{2}$ we construct $T$ as in Figure \ref{fig:sremo2srt:seq}.
For union,
we construct the \sra\ as in Figure \ref{fig:sremo2srt:or}.
For Kleene-star, 
we construct the \sra\ as in Figure \ref{fig:sremo2srt:iter}.
The only difference in these constructions is that we now assume,
without loss of generality,
that $T_{1}.R \cap T_{2}.R = \emptyset$,
i.e., that $T_{1}$ and $T_{2}$ have different sets of registers and that the automaton $T$ constructed from $T_{1}$ and $T_{2}$ retains all registers of both $T_{1}$ and $T_{2}$.
For example, 
if we have two \srt\ $T_{1}$ and $T_{2}$ and we want to construct a \srt\ $T$ such that $\mathit{Match}(T,S) = \mathit{Match}(T_{1},S) \cdot \mathit{Match}(T_{2},S)$ then we connect $T_{1}$'s final state to $T_{2}$'s start state via an $\epsilon$ transition.
It is easy to see that if $M_{1} \in \mathit{Match}(T_{1},S)$ and $M_{2} \in \mathit{Match}(T_{2},S)$ then $M = M_{1} \cdot M_{2} \in \mathit{Match}(T,S)$.
$S_{1}$ will force $T$ to move to $T_{1}'s$ final state 
(both $A$ and $T_{1}$ start with empty registers).
Subsequently, $T$ will jump to $T_{2}$'s start state and then $S_{2}$ will force $T$ to go to $T_{2}$'s final state which is $T$'s final state,
since $T_{2}$'s registers in $T$ are empty when $T_{2}$ starts reading $S_{2}$. 

We will now prove closure under intersection.
Let $T_{1} = (Q_{1},q_{1,s},Q_{1,f},R_{1},\Delta_{1})$ and $T_{2} = (Q_{2},q_{2,s},Q_{2,f},R_{2},\Delta_{2})$ be two \srt.
We wan to construct a \srt\ $T = (Q,q_{s},Q_{f},R,\Delta)$ such that 
$\mathit{Match}(T,S) = \mathit{Match}(T_{1},S) \cap \mathit{Match}(T_{2},S)$.
We construct $T$ as follows:
\begin{itemize}
	\item $Q = Q_{1} \times Q_{2}$.
	\item $q_{s} = (q_{1,s},q_{2,s})$.
	\item $Q_{f} = (q_{1},q_{2})$, where $q_{1} \in Q_{1,f}$ and $q_{2} \in Q_{2,f}$,
	i.e., $Q_{f} = Q_{1,f} \times Q_{2,f}$.
	\item $R = R_{1} \cup R_{2}$, assuming, without loss of generality, that $R_{1} \cap R_{2} = \emptyset$.
	\item For each $q = (q_{1},q_{2}) \in Q$ we add a transition $\delta$ to $q' = (q_{1}',q_{2}') \in Q$ if there exists a transition $\delta_{1}$ from $q_{1}$ to $q_{1}'$ in $T_{1}$ and a transition $\delta_{2}$ from $q_{2}$ to $q_{2}'$ in $T_{2}$.
	The condition of $\delta$ is $\phi = \delta_{1}.\phi \wedge \delta_{2}.\phi$.
	The output of $\delta$ is $\bullet$ if the outputs of $\delta_{1}$ and $\delta_{2}$ are both $\bullet$. 
	Otherwise, it is $\otimes$.
	The write registers of $\delta$ are $W = \delta_{1}.W \cup \delta_{2}.W$ 
	(notice that, 
	if $\delta_{1}.W \neq \emptyset$ and $\delta_{2}.W \neq \emptyset$, 
	this creates a multi-register \srt,
	even if $T_{1}$ and $T_{2}$ are single-register).
	Thus, $\delta = (q_{1},q_{2}),(\delta_{1}.\phi \wedge \delta_{2}.\phi) \downarrow (\delta_{1}.W \cup \delta_{2}.W) \rightarrow (q_{1}',q_{2}')$.
\end{itemize}
It is evident that, 
if a match $M$ is produced by both $T_{1}$ and $T_{2}$ on a string $S$,
it is also produced by $T$.
If $M$ is not produced either by $T_{1}$ or $T_{2}$,
then it is not produced by $T$.
Therefore,
$\mathit{Match}(T,S) = \mathit{Match}(T_{1},S) \cap \mathit{Match}(T_{2},S)$.

\end{proof}

%% file: proofs_complement.tex
\begin{theorem*}
\srt\ are not closed under complement.
\end{theorem*}

\begin{proof}
\begin{figure}[t]
\centering
\includegraphics[width=0.75\textwidth]{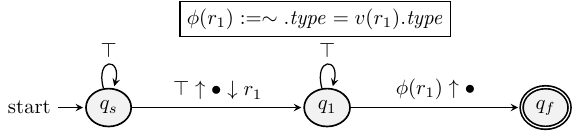}
\caption{\srt\ accepting strings which have the same type in two elements. Notice that $\sim$ denotes the current event (last event read from the string).}
\label{fig:complement}
\end{figure}
The proof is by a counter example.
Let $T$ denote the \srt\ of Figure \ref{fig:complement}.
This \srt\ reads strings composed of tuples.
Each tuple contains an attribute called $\mathit{type}$, 
taking values from a finite or infinite alphabet.
The symbol $\sim$ simply denotes the current element of the string,
i.e., the last element read from it.
Therefore, $T$ accepts strings in which there are two elements with the same type,
regardless of the length of $S$.
Assume that there exists a \srt\ $T_{c}$ which accepts only when $T$ does not accept.
In other words, 
$T_{c}$ accepts all strings $S$ whose elements all have a different type.
Let $k = \lvert T_{c}.R \rvert$ be the number of registers of $T_{c}$.
Let $\lvert S \rvert =  k+m$,
where $m > 1$,
be the length of a string $S$ whose elements all have different types.
However, $T_{c}$ cannot possibly exist.
At the end of $S$,
as $T_{c}$ is ready to read the last element of $S$,
it must have stored all of the previous $k+m-1$ elements of $S$.
But $T$ has only $k$ registers,
whereas $k+m-1 > k$,
since $m>1$.
Thus, $T_{c}$ cannot exist.
\end{proof}

%% file: proofs_determinization.tex
\begin{theorem*}
Not every \srt\ is output-agnostic determinizable.
\end{theorem*}

\begin{proof}
\begin{figure}[t]
\centering
\includegraphics[width=0.75\textwidth]{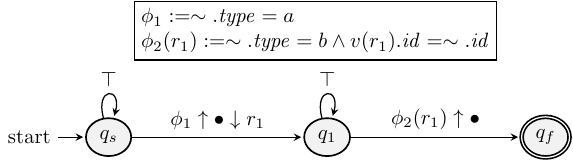}
\caption{\srt\ accepting all strings containing an $a$ element followed by a $b$ element, whose identifiers are the same.}
\label{fig:determinization}
\end{figure}
The proof is again by a counter example.
Let $T$ denote the \srt\ of Figure \ref{fig:determinization}.
This \srt\ reads strings composed of tuples.
Each tuple contains an attribute, 
called $\mathit{type}$, 
taking values from a finite or infinite alphabet.
It also contains another tuple, 
called $\mathit{id}$,
taking integer values.
$T$ thus accepts strings $S$ that contain an $a$ followed by a $b$,
whose ids are equal,
regardless of the length of $S$. 

Assume there exist a deterministic \srt\ $T_{d}$ with $k$ registers which is equivalent to $T$.
Let 
\begin{equation*}
S = (a, 1) (b, 2) 
\end{equation*}
be a string given to $T_{d}$.
After reading $S_{1}=(a,1)$,
$A_{d}$ must store it in a register $r_{1}$ in order to be able to compare it when $(b,2)$ arrives.
Let 
\begin{equation*}
S' = (a, 1) (a, 3) (b, 2) 
\end{equation*}
After reading $S_{1}'=(a,1)$,
$T_{d}$ must store it in the register $r_{1}$,
since $T_{d}$ is deterministic and follows a single run.
Thus, it must have the exact same behavior after reading $S_{1}$ and $S_{1}'$.
But we must also store $S_{2}'=(a,3)$ after reading it.
Additionally,
$S_{2}'$ must be stored in a different register $r_{2}$.
We cannot overwrite $r_{1}$.
If we did this and $S_{1}'$ were $(a,2)$,
then we would not be able to match $(a,2)$ to $S_{3}'=(b,2)$ and $S'=(a,2)(a,3)(b,2)$ would not be accepted.
Now, let
\begin{equation*}
S'' = \underbrace{(a, \cdots) (a, \cdots) \cdots (a, \cdots)}_{k+1 \text{ elements}}  (b, 2) 
\end{equation*}
With a similar reasoning,
all of the first $k+1$ elements of $S''$ must be stored after reading them.
But this is a contradiction,
as $T_{d}$ can store at most $k$ different elements.
Therefore, there does not exist a deterministic \srt\ which is equivalent to $T$.
\end{proof}

%% file: proofs_wsrem2dsra.tex
\begin{theorem*}
For every windowed \sremo\ there exists an equivalent output-agnostic deterministic \srt.
\end{theorem*}

\begin{proof}

In what follows,
we omit referring to the output of transitions,
since we will be focusing on output-agnostic determinism.
We will thus treat \srt\ as if they were automata without output.
We call such automata \emph{Symbolic Register Automata} (\sra)
(very similar to Symbolic Register Automata presented in \cite{DBLP:conf/cav/DAntoniFS019}).
Equivalence between deterministic and non-deterministic \sra\ will be shown at the level of languages,
not that of matches.

We first show how we can construct a so-called ``unrolled \sra''\ from a windowed expression:
\begin{lemma}
\label{lemma:windowed_srem}
For every windowed \sremo\ there exists an equivalent unrolled \sra\ without any loops, 
i.e., a \sra\ where each state may be visited at most once.
\end{lemma}

\begin{proof}

\input{algorithms_unroll_cycles}

Let $e_{w} := e^{[1..w]}$.
Algorithm \ref{algorithm:wsrem2sra} shows how we can construct $A_{e_{w}}$.
The basic idea is that we first construct as usual the \sra\ $A_{e}$ for the sub-expression $e$
(and eliminate $\epsilon$-transitions).
We can then use $A_{e}$ to enumerate all the possible walks of $A_{e}$ of length up to $w$ and then join them in a single \sra\ through disjunction.
A walk $w$ over a \sra\ $A$ is a sequence of transitions \linebreak $w=<\delta_{1},\cdots,\delta_{k}>$,
such that:
\begin{itemize}
	\item $\forall \delta_{i}\ \delta_{i} \in A.\Delta$
	\item $\delta_{1}.\mathit{source} = A.q_{s}$
	\item $\forall \delta_{i},\delta_{i+1}\ \delta_{i}.\mathit{target}=\delta_{i+1}.\mathit{source}$
\end{itemize}
We say that such a walk is of length $k$.
Essentially,
we need to remove cycles from every walk of $A_{e}$ by ``unrolling'' them as many times as necessary,
without the length of the walk exceeding $w$.
This ``unrolling'' operation is performed by the (recursive) Algorithm \ref{algorithm:unroll_cycles}.
Because of this ``unrolling'',
a state of $A_{e}$ may appear multiple times as a state in $A_{e_{w}}$.
We keep track of which states of $A_{e_{w}}$ correspond to states of $A_{e}$ through the function $\mathit{CopyOfQ}$ in the algorithm.
For example, if $q_{e}$ is a state of $A_{e}$, $q_{e_{w}}$ a state of $A_{e_{w}}$ and
$\mathit{CopyOfQ}(q_{e_{w}}) = q_{e}$,
this means that $q_{e_{w}}$ was created as a copy of $q_{e}$
(and multiple states of $A_{e_{w}}$ may be copies of the same state of $A_{e}$,
i.e., $\mathit{CopyOfQ}$ is a surjective but not an injective function).
We do the same for the registers as well, 
through the function $\mathit{CopyOfR}$.
The algorithm avoids an explicit enumeration,
by gradually building the automaton as needed,
through an incremental expansion.
Of course, walks that do not end in a final state may be removed,
either after the construction or online,
whenever a non-final state cannot be expanded.

The lemma is a direct consequence of the construction algorithm.
First, note that,
by the construction algorithm,
there is a one-to-one mapping (bijective function) between the walks/runs of $A_{e_{w}}$
and the walks/runs of $A_{e}$ of length up to $w$.
We can show that if $\varrho_{e}$ is a run of $A_{e}$ of length up to $w$ over a string $S$
($\varrho_{e}$ has at most $w$ transitions),
then the corresponding run $\varrho_{e_{w}}$ of $A_{e_{w}}$ is indeed a run and if $\varrho_{e}$ is accepting so is $\varrho_{e_{w}}$.
By definition, 
since the runs have no $\epsilon$-transitions and are at most of length $w$,
$\lvert S \rvert \leq w$.

We first prove the following proposition:

\begin{proposition*}
There exists a run of $A_{e}$ over a string $S$ of length up to $w$
\begin{equation*}
\varrho_{e}=[1,q_{e,1}=A_{e}.q_{s},v_{e,1}] \overset{\delta_{e,1}}{\rightarrow}  \cdots \overset{\delta_{e,i-1}}{\rightarrow} [n,q_{e,i},v_{e,i}] \overset{\delta_{e,i}}{\rightarrow} \cdots \overset{\delta_{e,n-1}}{\rightarrow} [n,q_{e,n},v_{e,n}]
\end{equation*}
iff
there exists a run $\varrho_{e_{w}}$ of $A_{e_{w}}$ 
\begin{equation*}
\varrho_{e_{w}}=[1,q_{e_{w},1}=A_{e_{w}}.q_{s},v_{e_{w},1}] \overset{\delta_{e_{w},1}}{\rightarrow}  \cdots \overset{\delta_{e_{w},i-1}}{\rightarrow} [n,q_{e_{w},i},v_{e_{w},i}] \overset{\delta_{e_{w},i}}{\rightarrow} \cdots \overset{\delta_{e_{w},n-1}}{\rightarrow} [n,q_{e_{w},n},v_{e_{w},n}]
\end{equation*}
such that:
\begin{itemize}
	\item $\mathit{CopyOfQ}(q_{e_{w},i}) = q_{e,i}$
	\item $v_{e,i}(r_{e})=v_{e_{w},i}(r_{e_{w}})$, 
	if
	$\mathit{CopyOfR}(r_{e_{w}})=r_{e}$
	and
	$r_{e_{w}}$ appears last among the registers that are copies of $r_{e}$ in $\varrho_{e_{w}}$.
	%\item $match(\varrho_{\psi})=match(\varrho_{\phi})=M$ if runs are accepting
\end{itemize}
\end{proposition*}

We say that a register $r$ appears in a run at position $i$ if $r \in \delta_{i}.W$,
i.e.,
if the $i^{th}$ transition writes to $r$.
We say that a register $r_{e_{w}}$,
where $\mathit{CopyOfR}(r_{e_{w}})=r_{e}$, 
appears last if no other copies of $r_{e}$ appear after $r_{e_{w}}$ in a run.
The notion of a register's (last) appearance also applies for walks of $A_{e_{w}}$,
since $A_{e_{w}}$ is a directed acyclic graph,
as can be seen by Algorithms \ref{algorithm:unroll_cycles0} and \ref{algorithm:unroll_cyclesk}
(they always expand ``forward'' the \sra, 
without creating any cycles and without merging any paths).

\begin{proof}
The proof is by induction on the length of the runs $k$, with $k \leq w$.
We prove only one direction (assume a run $\varrho_{e}$ exists).
The other is similar.

\textbf{Base case: $k=0$.}
For both \sra, 
only the start state and the initial configuration with all registers empty is possible.
Thus, $v_{e,i}=v_{e_{w},i}=\sharp$ for all registers. 
By Algorithm \ref{algorithm:unroll_cycles0} (line \ref{line:unroll_cycles:qs_descendent}),
we know that $\mathit{CopyOf}(q_{e_{w},s}) = q_{e,s}$.

\textbf{Case for $0 < k+1 \leq w$, assuming the proposition holds for $k$.}
Let
\begin{equation*}
\varrho_{e,k+1} = \cdots [k,q_{e,k},v_{e,k}] \overset{\delta_{e,k}}{\rightarrow} [k+1,q_{e,k+1},v_{e,k+1}]
\end{equation*}
and
\begin{equation*}
\varrho_{e_{w},k+1} = \cdots [k,q_{e_{w},k},v_{e_{w},k}] \overset{\delta_{e_{w},k}}{\rightarrow} [k+1,q_{e_{w},k+1},v_{e_{w},k+1}]
\end{equation*}
be the runs of $A_{e}$ and $A_{e_{w}}$ respectively of length $k+1$ over the same $k+1$ elements of a string $S$.
We know that $\varrho_{e,k+1}$ is an actual run and we need to construct $\varrho_{e_{w},k+1}$, 
knowing, by the induction hypothesis,
that there is an actual run up to $q_{e_{w},i+k}$.
Now, by the construction algorithm,
we can see that if $\delta_{e,k}$ is a transition of $A_{e}$ from $q_{e,k}$ to $q_{e,k+1}$,
there exists a transition $\delta_{e_{w},k}$ with the same condition
from $q_{e_{w},k}$ to a $q_{e_{w},k+1}$ such that $\mathit{CopyOfQ}(q_{e_{w},k+1})=q_{e,k+1}$.
Moreover, if $\delta_{e,k}$ is triggered,
so does $\delta_{e_{w},k}$,
because the registers in the register selection of $\delta_{e_{w},k}$
are copies of the corresponding registers in $\delta_{e,k}.\phi.rs$.
By the induction hypothesis,
we know that the contents of the registers in $\delta_{e,k}.\phi.rs$ will be equal to the contents of their corresponding registers in $\varrho_{e_{w}}$ that appear last.
But these are exactly the registers in $\delta_{e_{w},k}.\phi.rs$
(see line \ref{line:unroll_cycles:latest_appearance} in Algorithm \ref{algorithm:unroll_cyclesk}).
We can also see that the part of the proposition concerning the valuations $v$ also holds.
If $\delta_{e,k}.W = \{ r_{e} \}$ and $\delta_{e_{w},k}.W = \{ r_{e_{w}} \}$,
then we know,
by the construction algorithm
(line \ref{line:unroll_cycles:rn_descendent}),
that $\mathit{CopyOfR}(r_{e_{w}}) = r_{e}$ and $r_{e_{w}}$ will be the last appearance of a copy of $r_{e}$ in $\varrho_{e_{w},k+1}$.
Thus the proposition holds for $0 < k+1 \leq w$ as well.
\end{proof}

The above proposition must necessarily hold for accepting runs as well.
Therefore,
$A_{e}$ accepts the same language as $A_{e_{w}}$.
\end{proof}

We also note that $w$ must be a number greater than (or equal to) 
the minimum length of the walks induced by the accepting runs of $A_{e}$
(which is something that can be computed by the structure of the expression).
Although this is not a formal requirement,
if it is not satisfied,
then $A_{e_{w}}$ won't detect any matches.

The process for constructing a deterministic \sra\ (\dsra) from a windowed \sremo\ is shown in Algorithm \ref{algorithm:determinization}.
It first constructs a non-deterministic \sra\ (\nsra) and then uses the power set of this \nsra's states to construct the \dsra.
For each state $q_{d}$ of the \dsra,
it gathers all the conditions from the outgoing transitions of the states of the \nsra\ $q_{n}$ ($q_{n} \in q_{d}$),
it creates the (mutually exclusive) \emph{minterms} of these conditions, 
i.e., the set of maximal satisfiable Boolean combinations of the conditions.
It then creates transitions, 
based on these minterms.
Please, note that we use the ability of a transition to write to more than one registers.
So, from now on,
$\delta.W$ will be a set that is not necessarily a singleton.
This allows us to retain the same set of registers, i.e.,
the set of registers $R$ will be the same for the \nsra\ and the \dsra.
A new transition created for the \dsra\ may write to multiple registers,
if it ``encodes'' multiple transitions of the \nsra,
which may write to different registers.
It is also obvious that the resulting \sra\ is deterministic,
since the various minterms out of every state are mutually exclusive,
i.e., at most one may be triggered.
Intuitively,
having a windowed \sra\ allows us to construct a deterministic \sra\ with as many registers as necessary.
Therefore,
it is always possible to have available all past $w$ elements.
This is not possible in the counter-example of Section \ref{sec:proof:determinization},
where we showed that \sra\ are not in general determinizable.

First, we will prove the following proposition:
\begin{proposition*}
There exists a run $\varrho_{n}$ over a string $S$ which $A_{n}$ can follow by reading the first $k$ tuples of $S$,
iff there exists a run $\varrho_{d}$ 
that $A_{d}$ can follow by reading the same first $k$ tuples,
such that, if
\begin{equation*}
\varrho_{n} = [1,q_{n,1},v_{n,1}] \overset{\delta_{n,1}}{\rightarrow} \cdots \overset{\delta_{n,i-1}}{\rightarrow} [i,q_{n,k},v_{n,i}] \overset{\delta_{n,i}}{\rightarrow} \cdots \overset{\delta_{n,k-1}}{\rightarrow} [k,q_{n,k},v_{n,k}]
\end{equation*}
and
\begin{equation*}
\varrho_{d} = [1,q_{d,1},v_{d,1}] \overset{\delta_{d,1}}{\rightarrow}  \cdots \overset{\delta_{d,i-1}}{\rightarrow} [i,q_{d,i},v_{d,i}] \overset{\delta_{d,i}}{\rightarrow} \cdots \overset{\delta_{d,k-1}}{\rightarrow} [k,q_{d,k},v_{d,k}]
\end{equation*}
are the runs of $A_{n}$ and $A_{d}$ respectively, then,
\begin{itemize}
	\item $q_{n,i} \in q_{d,i}\ \forall i: 1 \leq i \leq k$
	\item if $r \in A_{d}.R$ appears in $\varrho_{n}$, then it appears in $\varrho_{d}$
	\item $v_{n,i}(r) = v_{d,i}(r)$  for every $r$ that appears in $\varrho_{n}$ (and $\varrho_{d})$
\end{itemize}
\end{proposition*}

We say that a register $r$ appears in a run at position $i$ if $r \in \delta_{i}.W$.

\input{algorithms_determinization}

\begin{proof}
We will prove only direction (the other is similar).
Assume there exists a run $\varrho_{n}$.
We will prove that there exists a run $\varrho_{d}$ by induction on the length $k$ of the run.

\textbf{Base case: $k=0$.}
Then $\varrho_{n}=[1,q_{n,1},\sharp]=[1,q_{n,s},\sharp]$.
The run $\varrho_{d}=[1,q_{d,s},\sharp]$ is indeed a run of the \dsra\
that satisfies the proposition,
since $q_{n,s} \in q_{d,s} = \{q_{n,s}\}$ 
(by the construction algorithm, line \ref{line:determinization:start_state}),
all registers are empty
and no registers appear in the runs. 

\textbf{Case $k>0$.}
Assume the proposition holds for $k$.
We will prove it holds for $k+1$ as well.
Let 
\begin{equation}
\label{run:n}
\varrho_{n,k+1} = \cdots [k,q_{n,k},v_{n,k}] \begin{cases}
\overset{\delta_{n,k}^{1}}{\rightarrow} [k+1,q_{n,k+1}^{1},v_{n,k+1}^{1}] \\
\overset{\delta_{n,k}^{2}}{\rightarrow} [k+1,q_{n,k+1}^{2},v_{n,k+1}^{2}] \\
\cdots \\
\overset{\delta_{n,k}^{m}}{\rightarrow} [k+1,q_{n,k+1}^{m},v_{n,k+1}^{m}]
\end{cases}
\end{equation}
be the possible runs that can follow a run $\varrho_{n,k}$ after the \nsra\ reads the $(k+1)^{th}$ tuple.
Notice that,
typically,
since $A_{n}$ is non-deterministic,
there might be multiple runs $\varrho_{n,k}$ and each such run can spawn its own multiple runs $\varrho_{n,k+1}$.
The same reasoning that we present below applies to all these $\varrho_{n,k}$.

We need to find a run of the \dsra\ like:
\begin{equation*}
\varrho_{d,k+1} = \cdots [k,q_{d,k},v_{d,k}] \overset{\delta_{d,k}}{\rightarrow} [k+1,q_{d,k+1},v_{d,k+1}]
\end{equation*}

\begin{figure}[t]
\centering
\begin{subfigure}[t]{0.25\textwidth}
	\includegraphics[width=0.99\textwidth]{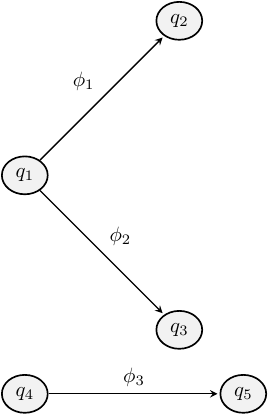}
	\caption{\nsra.}
	\label{fig:determinization_example:nsra}
\end{subfigure}
\begin{subfigure}[t]{0.73\textwidth}
	\includegraphics[width=0.99\textwidth]{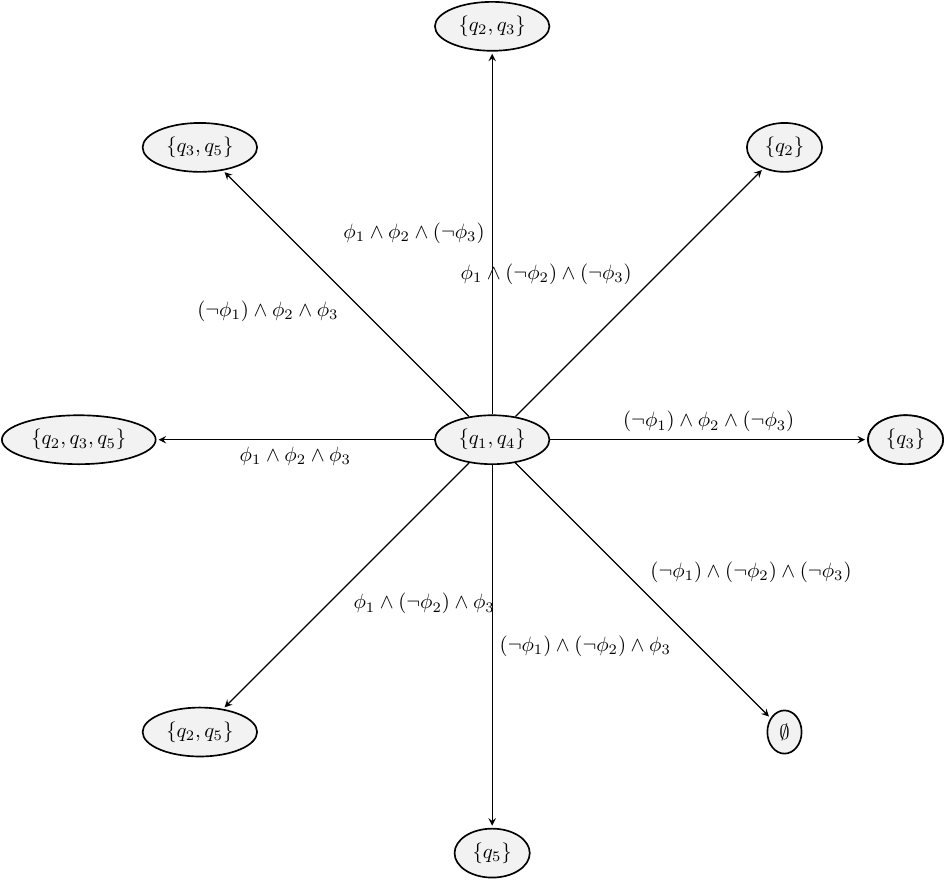}
	\caption{\dsra.}
	\label{fig:determinization_example:dsra}
\end{subfigure}
\caption{Example of converting a \nsra\ to a \dsra.}
\label{fig:determinization_example}
\end{figure}

By the induction hypothesis,
we know that $q_{n,k} \in q_{d,k}$.
By the construction Algorithm \ref{algorithm:determinization},
we then know that,
if $\phi_{n,k}^{j}=\delta_{n,k}^{j}.\phi$ is the condition of a transition that takes the non-deterministic run to $q_{n,k+1}^{j}$,
then there exists a transition $\delta_{d,k}$ in the \dsra\ from $q_{d,k}$ whose condition will be a minterm,
containing all the $\phi_{n,k}$ in their positive form
and all other possible conditions in their negated form.
Moreover, 
the target of that transition, 
$q_{d,k+1}$, 
contains all $q_{n,k+1}^{j}$.
More formally,  $q_{d,k+1} = \bigcup\limits_{j=1}^{m}{q_{n,k+1}^{j}}$.

As an example,
see Figure \ref{fig:determinization_example}.
Figure \ref{fig:determinization_example:nsra} depicts part of a \nsra.
Figure \ref{fig:determinization_example:dsra} depicts part of the \dsra\ that woyld be constructed from that of Figure \ref{fig:determinization_example:nsra}.
The construction algorithm would create the state $\{q_{2},q_{4}\}$,
the minterms from the conditions of all the outgoing transitions of $q_{2}$ and $q_{4}$ and then attempt to determine which minterm would move the \dsra\ to which subset of $\{q_{2},q_{3},q_{5}\}$.
The results is shown in Figure \ref{fig:determinization_example:dsra}.
Now, assume that a run of the \nsra\ has reached $q_{1}$ via one run and $q_{4}$ via another run,
i.e. $q_{n,k} = q_{1}$ in Eq. \eqref{run:n} for the first of these runs and $q_{n,k} = q_{4}$ for the second.
Assume also that both $\phi_{1}$ and $\phi_{2}$ are triggered after reading the $(k+1)^{th}$ element,
but not $\phi_{3}$.
This means that the \nsra\ would move to $q_{2}$ and $q_{3}$.
In Eq. \eqref{run:n},
this would mean that $m=2$ and that $\delta_{n,k}^{1}.\phi = \phi_{1}$ and $\delta_{n,k}^{2}.\phi = \phi_{2}$.
But in the \dsra\ there is a transition that simulates this move of the \nsra.
The minterm $\phi_{1} \wedge \phi_{2} \wedge (\neg \phi_{3})$ moves the \dsra\ to $\{q_{2},q_{3}\}$.
It contains $\delta_{n,k}^{1}.\phi$ and $\delta_{n,k}^{2}.\phi$ in their positive form and all other conditions (here only $\phi_{3}$) in their negated form.
With a similar reasoning,
we see that the \dsra\ can simulate the \nsra\ for every other possible combination of $\{\phi_{1},\phi_{2},\phi_{3}\}$.

What we have proven thus far is a structural similarity between \nsra\ and \dsra.
We also need to prove that $\delta_{d,k}$ applies as well,
i.e., that the minterm on this transition is triggered exactly when its positive conjuncts are triggered.
To prove this,
we need to show that the contents of the registers that a condition $\phi$ of the \nsra\ accesses are the same that this $\phi$ accesses in the \dsra\ when participating in a minterm.

As we said,
the condition on $\delta_{d,k}$ is a conjunct (minterm),
where all $\phi_{n,k}^{j}$ appear in their positive form and all other conditions in their negated form.
But note that the conditions in negated form are those that were not triggered in $\varrho_{n,k+1}$ when reading the $(k+1)^{th}$ tuple.
Additionally,
the arguments passed to each of the conditions of the minterm are the same (registers) as those passed to them in the non-deterministic run
(by the construction algorithm, line \ref{line:determinization:register_selection}). 
To make this point clearer,
consider the following simple example of a minterm:
\begin{equation*}
\phi = \phi_{1}(r_{1,1},\cdots,r_{1,k}) \wedge \neg \phi_{2}(r_{2,1},\cdots,r_{2,l}) \wedge \phi_{3}(r_{3,1},\cdots,r_{3,m})
\end{equation*}
This means that $\phi_{1}(r_{1,1},\cdots,r_{1,k})$,
with the exact same registers as arguments,
will be the formula of a transition of the \nsra that was triggered.
Similarly for $\phi_{3}$.
With respect to $\phi_{2}$,
it will be the condition of a transition that was not triggered.
If we can show that the contents of those registers are the same in the runs of the \nsra\ and \dsra\ when reading the last tuple,
then this will mean that $\delta_{d,k}.\phi$ is indeed triggered.
But this is the case by the induction hypothesis
($v_{n,k}(r) = v_{d,k}(r)$),
since all these registers appear in the run $\varrho_{n,k}$ up to $q_{n,k}$.

The second part of the proposition also holds,
since,
by the construction,
$\delta_{d,k}$ will write to all the registers that the various $\delta_{n,k}^{j}$ write
(see line \ref{line:determinization:register_writing} in the determinization algorithm).

The third part also holds.
This is the part that actually ensures that the contents of the registers are the same.
First, 
note that a register can appear only once in a run of $A_{n}$,
because of its tree-like structure.
Second,
by the construction,
we know that $\delta_{d,k}.W = \bigcup\limits_{j=1}^{m} { \delta_{n,k}^{j}.W }$
(see again line \ref{line:determinization:register_writing} in the algorithm).
Therefore, we know that $\delta_{d,k}$ will write only to registers that had not appeared before in the run of the \nsra\ and will leave every other register that had appeared unaffected.
This observation is critical.
We could not claim the same for non-windowed \sra,
as in Figure \ref{fig:determinization}.
If we attempted to determinize this \nsra,
without unrolling its cycles,
the resulting \sra\ could overwrite $r_{1}$.
Now, since $\delta_{d,k}$ and all the $\delta_{n,k}^{j}$ write the same element
and $\delta_{d,k}$ does not affect any previously appearing registers,
the proposition holds.
\end{proof}

Since the above proposition holds for accepting runs as well,
we can conclude that there exists an accepting run of $A_{n}$ iff there exists an accepting run of $A_{d}$.
According to the above proposition,
the union of the last states over all $\varrho_{n}$ is equal to the last state of $\varrho_{d}$.
Thus, if $\varrho_{n}$ reaches a final state,
then the last state of $\varrho_{d}$ will contain this final state and hence be itself a final state.
Conversely, if $\varrho_{d}$ reaches a final state of $A_{d}$,
it means that this state contains a final state of $A_{n}$.
Then, there must exist a $\varrho_{n}$ that reached this final state.

\end{proof}

%% file: algorithms_unroll_cycles.tex
\begin{algorithm}
\SetAlgoNoLine
\KwIn{Windowed \sremo\ $e' := e^{[1..w]}$}
\KwOut{\sra\ $A_{e'}$ equivalent to $e'$}
$A_{e,\epsilon} \leftarrow ConstructSRA(e)$; \tcp{{\footnotesize As described in Appendix \ref{sec:proof:sremo2srt}.}}\
$A_{e,ms} \leftarrow \mathit{EliminateEpsilon}(A_{e,\epsilon})$; \tcp{{\footnotesize See Algorithm \ref{algorithm:epsilon}. $A_{e,ms}$ might be multi-register.}}\
$A_{e} \leftarrow \mathit{ConvertToSingleRegister}(A_{e,ms})$; \tcp{{\footnotesize As described in \cite{DBLP:journals/corr/abs-2110-04032}.}}\
$A_{e'} \leftarrow Unroll(A_{e},w)$; \tcp{{\footnotesize See Algorithm \ref{algorithm:unroll_cycles}.}}\ 
$\mathtt{return}\ A_{e'}$\;
\caption{Constructing \sra\ for a windowed expression ($\mathit{ConstructWSRA}$).}
\label{algorithm:wsrem2sra}
\end{algorithm}

\begin{algorithm}
\SetAlgoNoLine
\KwIn{\sra\ $A$ and integer $k \geq 0$}
\KwOut{\sra\ $A_{k}$ with runs of length up to $k$}
\eIf{$k=0$}{
	$(A_{k},\mathit{Frontier},\mathit{CopyOfQ},\mathit{CopyOfR}) \leftarrow \mathit{Unroll0}(A)$;
	\tcp{{\footnotesize Algorithm \ref{algorithm:unroll_cycles0}}}
}
{
	$(A_{k},\mathit{Frontier},\mathit{CopyOfQ},\mathit{CopyOfR}) \leftarrow \mathit{UnrollK}(A,k)$;
	\tcp{{\footnotesize Algorithm \ref{algorithm:unroll_cyclesk}}}
}
$\mathtt{return}\ (A_{k},\mathit{Frontier},\mathit{CopyOfQ},\mathit{CopyOfR})$\;
\caption{Unrolling cycles for windowed expressions\ ($\mathit{Unroll}$).}
\label{algorithm:unroll_cycles}
\end{algorithm}

\begin{algorithm}
\SetAlgoNoLine
\KwIn{\sra\ $A$}
\KwOut{\sra\ $A_{0}$ with runs of length 0}
$q \leftarrow \mathit{CreateNewState}()$\;
$\mathit{CopyOfQ} \leftarrow \{q \rightarrow A.q_{s}\}$\; \label{line:unroll_cycles:qs_descendent}
$\mathit{CopyOfR} \leftarrow \emptyset$\;
$\mathit{Frontier} \leftarrow \{q\}$\;
$Q_{f} \leftarrow \emptyset$\;
\If{$A.q_{s} \in A.Q_{f}$}{
	$Q_{f} \leftarrow Q_{f} \cup \{q\}$\;
}
$A_{0} \leftarrow (\{q\},q,Q_{f},\emptyset,\emptyset)$\;
$\mathtt{return}\ (A_{0},\mathit{Frontier},\mathit{CopyOfQ},\mathit{CopyOfR})$\;
\caption{Unrolling cycles for windowed expressions, base case ($\mathit{Unroll0}$).}
\label{algorithm:unroll_cycles0}
\end{algorithm}

\begin{algorithm}
%\SetAlgoNoLine
\KwIn{\sra\ $A$ and integer $k > 0$}
\KwOut{\sra\ $A_{k}$ with runs of length up to $k$}
$(A_{k-1},\mathit{Frontier},\mathit{CopyOfQ},\mathit{CopyOfR}) \leftarrow \mathit{Unroll}(A,k-1)$\;
$\mathit{NextFrontier} \leftarrow \emptyset$\;
$Q_{k} \leftarrow A_{k-1}.Q$; $Q_{k,f} \leftarrow A_{k-1}.Q_{f}$;
$R_{k} \leftarrow A_{k-1}.R$; $\Delta_{k} \leftarrow A_{k-1}.\Delta$\;
\ForEach{$q \in \mathit{Frontier}$}{
	$q_{c} \leftarrow \mathit{CopyOfQ}(q)$\;
	\ForEach{$\delta \in A.\Delta: \delta.\mathit{source} = q_{c}$}{
		$q_{new} \leftarrow \mathit{CreateNewState}()$\;
		$Q_{k} \leftarrow Q_{k} \cup \{q_{new}\}$\;
		$\mathit{CopyOfQ} \leftarrow \mathit{CopyOfQ} \cup \{ q_{new} \rightarrow \delta.\mathit{target} \}$\;
		\If{$\delta.\mathit{target} \in A.Q_{f}$}{
			$Q_{k,f} \leftarrow Q_{k,f} \cup \{q_{new}\}$\;
		}	
		\eIf{$\delta.W = \emptyset$}{
			$R_{new} \leftarrow \emptyset$\;	
		}
		{
			$r_{new} \leftarrow \mathit{CreateNewRegister}()$\;
			$R_{k} \leftarrow R_{k} \cup \{ r_{new} \}$\;
			$R_{new} \leftarrow \{ r_{new} \}$\;
			$\mathit{CopyOfR} \leftarrow \mathit{CopyOfR} \cup \{r_{new} \rightarrow \delta.r \}$; \label{line:unroll_cycles:rn_descendent}
			\tcp{{\footnotesize $\delta.r$ single element of $\delta.W$}}
		}		
		$\phi_{new} \leftarrow \delta.\phi$\;
		$rs_{new} \leftarrow ()$\;
		\tcc{{\footnotesize By $\delta.\phi.rs$ we denote the register selection of $\delta.\phi$, i.e., all the registers referenced by $\delta.\phi$ in its arguments. $rs$ is represented as a list.}}\
		\ForEach{$r \in \delta.\phi.rs$}{
			\tcc{{\footnotesize $\mathit{FindLastAppearance}$ returns a register that is a copy of $r$ and appears last in the trail of $A_{k-1}$ to $q$ (no other copies of $r$ appear after $r_{latest}$). Due to the construction, only a single walk/trail to $q$ exists.}}\
			$r_{latest} \leftarrow \mathit{FindLastAppearance}(r,q,A_{k-1})$\; \label{line:unroll_cycles:latest_appearance}
			\tcc{{\footnotesize $::$ denotes the operation of appending an element at the end of a list. $r_{latest}$ is appended at the end of $rs_{new}$.}}\
			$rs_{new} \leftarrow rs_{new} :: r_{latest}$\;
		}
		$\delta_{new} \leftarrow \mathit{CreateNewTransition}(q,\phi_{new}(rs_{new}) \downarrow R_{new} \rightarrow q_{new})$\;
		$\Delta_{k} \leftarrow \Delta_{k} \cup \{ \delta_{new} \}$\;
		$\mathit{NextFrontier} \leftarrow \mathit{NextFrontier} \cup \{q_{new}\}$\;
	}
}
$A_{k} \leftarrow (Q_{k}, A_{k-1}.q_{s}, Q_{k,f}, R_{k}, \Delta_{k})$\;
$\mathtt{return}\ (A_{k},\mathit{NextFrontier},\mathit{CopyOfQ},\mathit{CopyOfR})$\;
\caption{Unrolling cycles for windowed expressions, $k > 0$ ($\mathit{UnrollK}$).}
\label{algorithm:unroll_cyclesk}
\end{algorithm}

%% file: algorithms_determinization.tex
\begin{algorithm}
%\SetAlgoNoLine
\KwIn{Windowed \sremo\ $e' := e^{[1..n]}$}
\KwOut{Deterministic \sra\ $A_{d}$ equivalent to $e'$}
$A_{n} \leftarrow \mathit{ConstructWSRA}(e')$; \tcp{{\footnotesize See Algorithm \ref{algorithm:wsrem2sra}}}\
$Q_{d} \leftarrow \mathit{ConstructPowerSet}(A_{n}.Q)$\;
$\Delta_{d} \leftarrow \emptyset$;	$Q_{f,d} \leftarrow \emptyset$\;
\ForEach{$q_{d} \in Q_{d}$}{
	\If{$q_{d} \cap A_{n}.Q_{f} \neq \emptyset$}{
		$Q_{f,d} \leftarrow Q_{f,d} \cup \{ q_{d} \}$\;
	}
	$\mathit{Conditions} \leftarrow ()$;	$rs_{d} \leftarrow ()$\;
	\ForEach{$q_{n} \in q_{d}$}{
		\ForEach{$\delta_{n} \in A_{n}.\Delta: \delta_{n}.\mathit{source} = q_{n}$ }{
			$\mathit{Conditions} \leftarrow \mathit{Conditions} :: \delta_{n}.\phi$\;
			$rs_{d} \leftarrow rs_{d} :: \delta_{n}.\phi.rs$\; \label{line:determinization:register_selection}
		}
	}
	\tcc{$\mathit{ConstructMinTerms}$ returns the min-terms from a set of conditions. 
	For example, if $\mathit{Conditions} = (\phi_{1},\phi_{2})$, then $\mathit{MinTerms} = (\phi_{1} \wedge \phi_{2}, \neg \phi_{1} \wedge \phi_{2}, \phi_{1} \wedge \neg \phi_{2}, \neg \phi_{1} \wedge \neg \phi_{2})$}\
	$\mathit{MinTerms} \leftarrow \mathit{ConstructMinTerms}(\mathit{Conditions})$\; 
	\ForEach{$mt \in \mathit{MinTerms}$}{
		$p_{d} \leftarrow \emptyset$;	$W_{d} \leftarrow \emptyset$\;
		\ForEach{$q_{n} \in q_{d}$}{
			\ForEach{$\delta_{n} \in A_{n}.\Delta: \delta_{n}.\mathit{source} = q_{n}$ }{
				\tcc{$\phi \vDash \psi$ denotes entailment, i.e., if $\phi$ is true then $\psi$ is necessarily also true. 
				For example, $\phi_{1} \wedge \neg \phi_{2} \vDash \phi_{1}$.}\
				\If{$mt \vDash \delta_{n}.\phi$}{
					$p_{d} \leftarrow p_{d} \cup \{\delta_{n}.\mathit{target}\}$\;
					$W_{d} \leftarrow W_{d} \cup \{\delta_{n}.W\}$\; \label{line:determinization:register_writing}
				}
			}
		}
		$\delta_{d} \leftarrow \mathit{CreateNewTransition}(q_{d},mt(rs_{d}) \downarrow W_{d} \rightarrow p_{d})$\;
		$\Delta_{d} \leftarrow \Delta_{d} \cup \{\delta_{d}\}$\;
	}
}
$q_{d,s} \leftarrow \{A_{n}.q_{s}\}$\; \label{line:determinization:start_state}
$A_{d} \leftarrow (Q_{d},q_{d,s},Q_{f,d},A_{N}.R,\Delta_{d})$\;
$\mathtt{return}\ A_{d}$\;
\caption{Determinization.}
\label{algorithm:determinization}
\end{algorithm}

%% file: proofs_wsra_complement.tex
\begin{corollary*}
Windowed \srt\ with ignored outputs are closed under complement.
\end{corollary*}

\begin{proof}

Since we ignore outputs,
we will be focusing again on \sra.
Let $A$ be a windowed \sra.
We first determinize it to obtain $A_{d}$.
Although $A_{d}$ is deterministic,
it might still be incomplete,
i.e., there might be states from which it might be impossible to move to another state.
This may happen if it is possible that the conditions on all of the outgoing transitions of such a state are not triggered.
As in classical automata,
such a behavior implies that the string provided to the automaton is not accepted by it.

\input{algorithms_complement}

We can make $A_{d}$ complete by adding a so-called ``dead'' state $q_{dead}$ (non-final) to $A_{d}$.
See Algorithm \ref{algorithm:complement}.
For each state $q$ of $A_{d}$,
we then gather all the conditions on its outgoing transitions.
Let $\Phi$ denote this set of conditions.
We can then create the conjunction of all the negated conditions in $\Phi$:
$\phi_{dead} := (\neg \phi_{1}) \wedge (\neg \phi_{2}) \wedge \cdots \wedge (\neg \phi_{n})$,
where $\phi_{i} \in \Phi$ and $\bigcup\limits_{i=1}^{n} \phi_{i} = \Phi$.
We then add a transition from $q$ to $q_{dead}$ with $\phi_{dead}$ as its condition and $\emptyset$ as its write registers.
If we do this for every state $q \in A_{d}.Q$,
we will have created a \sra\ that is equivalent to $A_{d}$,
since transitions to $q_{dead}$ are only triggered if none of the other conditions in $\Phi$ are triggered.
If there exists a condition $\phi_{i}$ that is triggered,
the new automaton will behave exactly as $A_{d}$ and if no $\phi$ is triggered it will go to $q_{dead}$.
Now, if we add a self-loop transition on $q_{dead}$ with $\top$ as its condition,
we also ensure that the new automaton will always stay in $q_{dead}$ once it enters it.
$q_{dead}$ thus acts as a sink state.
This new automaton $A_{d,c}$ will therefore be equivalent to $A_{d}$ and it will also be both deterministic and complete. 

The final move is to flip all the states of $A_{d,c}$,
i.e., make all of its final states non-final and all of its non-final states final, to obtain an automaton $A_{complement}$.
This then ensures that if a string $S$ is accepted by $A$ (or $A_{d}$),
it will not be accepted by $A_{complement}$ and if it is accepted by $A_{complement}$ it will not be accepted by $A$.
This is indeed possible because $A$ (and $A_{complement}$) is deterministic and complete.
Therefore, 
for every string $S$,
there exists exactly one run of $A$ (and $A_{complement}$) over $S$.
If $A$, after reading $S$, reaches a final state,
$A_{complement}$ necessarily reaches a non-final state and vice versa.
Therefore,
for every windowed \sra\ $A$ we can indeed construct a \sra\ which accepts the complement of the language of $A$. 

Notice that this trick of flipping the states would not be possible if $A$ were non-deterministic.
To see this,
assume that $A$ is non-deterministic and at the end of $S$ it reaches states $q_{1}$ and $q_{2}$,
where $q_{1}$ is non-final and $q_{2}$ is final.
This means that $S$ is accepted by $A$.
If we flip the states of the non-deterministic $A$ to get its complement $A_{complement}$,
we would again reach $q_{1}$ and $q_{2}$,
where,
in this case,
$q_{1}$ is final and $q_{2}$ is non-final.
$A_{complement}$ would thus again accept $S$,
which is not the desired behavior for $A_{complement}$.
\end{proof}

%% file: algorithms_complement.tex
\begin{algorithm}
\SetAlgoNoLine
\KwIn{Windowed \sra\ $A$}
\KwOut{\sra\ $A_{complement}$ accepting the complement of $A$'s language}
$A_{d} \leftarrow \mathit{Determinize}(A)$; \tcp{{\footnotesize See Algorithm \ref{algorithm:determinization}.}}\
$q_{dead} \leftarrow \mathit{CreateNewState}()$\;
$\Delta_{dead} \leftarrow \emptyset$\;
\ForEach{$q \in A_{d}.Q$}{
	$\Phi \leftarrow \emptyset$\;
	\ForEach{$\delta \in A_{d}.\Delta: \delta.\mathit{source} = q$}{
		$\Phi \leftarrow \Phi \cup \delta.\phi$\;
	}
	$\phi_{dead} \leftarrow \top$\;
	\ForEach{$\phi_{i} \in \Phi$}{
		$\phi_{dead} \leftarrow \phi_{dead} \wedge  (\neg \phi_{i})$\;
	} 
	$\delta_{dead} \leftarrow \mathit{CreateNewTransition}(q,\phi_{dead} \downarrow \emptyset \rightarrow q_{dead})$\;
	$\Delta_{dead} \leftarrow \Delta_{dead} \cup \delta_{dead}$\;
}
$\delta_{loop} \leftarrow \mathit{CreateNewTransition}(q_{dead},\top \downarrow \emptyset \rightarrow q_{dead})$\;
$\Delta_{dead} \leftarrow \Delta_{dead} \cup \delta_{loop}$\;
$Q_{comp} \leftarrow A.Q \cup \{q_{dead}\}$\;
$q_{comp,s} \leftarrow A.q_{s}$\;
$Q_{comp,f} \leftarrow A.Q \setminus A.Q_{f}$\;
$R_{comp} \leftarrow A.R$\;
$\Delta_{comp} \leftarrow A.\Delta \cup \Delta_{dead}$\;
$A_{complement} \leftarrow (Q_{comp},q_{comp,s},Q_{comp,f},R_{comp},\Delta_{comp})$\; 
$\mathtt{return}\ A_{complement}$\;
\caption{Constructing the complement of a \sra\ ($\mathit{Complement}$).}
\label{algorithm:complement}
\end{algorithm}

%% file: ms.bbl
\begin{thebibliography}{10}

\bibitem{Esper}
Esper.
\newblock \url{https://www.espertech.com/esper/}.
\newblock [Online; accessed 23-May-2024].

\bibitem{EsperComplexity}
Esper complexity.
\newblock
  \url{http://esper.espertech.com/release-8.9.0/reference-esper/html/performance.html}.
\newblock [Online; accessed 23-May-2024].

\bibitem{FlinkMR}
Flink - pattern recognition.
\newblock
  \url{https://nightlies.apache.org/flink/flink-docs-release-1.18/docs/dev/table/sql/queries/match_recognize/}.
\newblock [Online; accessed 23-May-2024].

\bibitem{FlinkCEP}
Flinkcep - complex event processing for flink.
\newblock
  \url{https://nightlies.apache.org/flink/flink-docs-release-1.17/docs/libs/cep/}.
\newblock [Online; accessed 23-May-2024].

\bibitem{FlinkCEPNFA}
Flinkcep nfa source code.
\newblock
  \url{https://github.com/apache/flink/blob/master/flink-libraries/flink-cep/src/main/java/org/apache/flink/cep/nfa/NFA.java}.
\newblock [Online; accessed 23-May-2024].

\bibitem{MRISO}
Iso/iec 19075-5:2021 information technology — guidance for the use of
  database language sql — part 5: Row pattern recognition.
\newblock
  \url{https://standards.iteh.ai/catalog/standards/iso/f753ca23-4b3c-4c9f-8a0a-1113f39bc404/iso-iec-19075-5-2021}.
\newblock [Online; accessed 23-May-2024].

\bibitem{SASE}
Sase open source system.
\newblock \url{https://github.com/haopeng/sase}.
\newblock [Online; accessed 23-May-2024].

\bibitem{DBLP:conf/sigmod/AgrawalDGI08}
Jagrati Agrawal, Yanlei Diao, Daniel Gyllstrom, and Neil Immerman.
\newblock Efficient pattern matching over event streams.
\newblock In {\em {SIGMOD} Conference}, pages 147--160. {ACM}, 2008.

\bibitem{DBLP:conf/lpar/AlevizosAP18}
Elias Alevizos, Alexander Artikis, and George Paliouras.
\newblock Wayeb: a tool for complex event forecasting.
\newblock In {\em {LPAR}}, volume~57 of {\em EPiC Series in Computing}, pages
  26--35. EasyChair, 2018.

\bibitem{DBLP:journals/corr/abs-1804-09999}
Elias Alevizos, Alexander Artikis, and Georgios Paliouras.
\newblock Symbolic automata with memory: a computational model for complex
  event processing.
\newblock {\em CoRR}, abs/1804.09999, 2018.

\bibitem{DBLP:journals/corr/abs-2110-04032}
Elias Alevizos, Alexander Artikis, and Georgios Paliouras.
\newblock Symbolic register automata for complex event recognition and
  forecasting.
\newblock {\em CoRR}, abs/2110.04032, 2021.

\bibitem{DBLP:journals/vldb/AlevizosAP22}
Elias Alevizos, Alexander Artikis, and Georgios Paliouras.
\newblock Complex event forecasting with prediction suffix trees.
\newblock {\em {VLDB} J.}, 31(1):157--180, 2022.

\bibitem{DBLP:journals/tocl/BojanczykDMSS11}
Mikolaj Bojanczyk, Claire David, Anca Muscholl, Thomas Schwentick, and Luc
  Segoufin.
\newblock Two-variable logic on data words.
\newblock {\em {ACM} Trans. Comput. Log.}, 12(4):27:1--27:26, 2011.

\bibitem{DBLP:journals/corr/abs-2111-04635}
Marco Bucchi, Alejandro Grez, Andr{\'{e}}s Quintana, Cristian Riveros, and
  Stijn Vansummeren.
\newblock {CORE:} a complex event recognition engine.
\newblock {\em CoRR}, abs/2111.04635, 2021.

\bibitem{DBLP:journals/pvldb/BucchiGQRV22}
Marco Bucchi, Alejandro Grez, Andr{\'{e}}s Quintana, Cristian Riveros, and
  Stijn Vansummeren.
\newblock {CORE:} a complex event recognition engine.
\newblock {\em Proc. {VLDB} Endow.}, 15(9):1951--1964, 2022.

\bibitem{DBLP:journals/pvldb/ChandramouliGM10}
Badrish Chandramouli, Jonathan Goldstein, and David Maier.
\newblock High-performance dynamic pattern matching over disordered streams.
\newblock {\em Proc. {VLDB} Endow.}, 3(1):220--231, 2010.

\bibitem{DBLP:conf/debs/CugolaM10}
Gianpaolo Cugola and Alessandro Margara.
\newblock {TESLA:} a formally defined event specification language.
\newblock In {\em {DEBS}}, pages 50--61. {ACM}, 2010.

\bibitem{DBLP:journals/csur/CugolaM12}
Gianpaolo Cugola and Alessandro Margara.
\newblock Processing flows of information: From data stream to complex event
  processing.
\newblock {\em {ACM} Comput. Surv.}, 44(3):15:1--15:62, 2012.

\bibitem{DBLP:conf/cav/DAntoniFS019}
Loris D'Antoni, Tiago Ferreira, Matteo Sammartino, and Alexandra Silva.
\newblock Symbolic register automata.
\newblock In {\em {CAV} {(1)}}, volume 11561 of {\em Lecture Notes in Computer
  Science}, pages 3--21. Springer, 2019.

\bibitem{DBLP:conf/cav/DAntoniV17}
Loris D'Antoni and Margus Veanes.
\newblock The power of symbolic automata and transducers.
\newblock In {\em {CAV} {(1)}}, volume 10426 of {\em Lecture Notes in Computer
  Science}, pages 47--67. Springer, 2017.

\bibitem{demers2005general}
Alan Demers, Johannes Gehrke, Mingsheng Hong, Mirek Riedewald, and Walker
  White.
\newblock A general algebra and implementation for monitoring event streams.
\newblock Technical report, Cornell University, 2005.

\bibitem{DBLP:conf/edbt/DemersGHRW06}
Alan~J. Demers, Johannes Gehrke, Mingsheng Hong, Mirek Riedewald, and Walker~M.
  White.
\newblock Towards expressive publish/subscribe systems.
\newblock In {\em {EDBT}}, volume 3896 of {\em Lecture Notes in Computer
  Science}, pages 627--644. Springer, 2006.

\bibitem{DBLP:conf/cidr/DemersGPRSW07}
Alan~J. Demers, Johannes Gehrke, Biswanath Panda, Mirek Riedewald, Varun
  Sharma, and Walker~M. White.
\newblock Cayuga: {A} general purpose event monitoring system.
\newblock In {\em {CIDR}}, pages 412--422. www.cidrdb.org, 2007.

\bibitem{DBLP:conf/ijcai/DoussonM07}
Christophe Dousson and Pierre~Le Maigat.
\newblock Chronicle recognition improvement using temporal focusing and
  hierarchization.
\newblock In {\em {IJCAI}}, pages 324--329, 2007.

\bibitem{DBLP:books/daglib/0024062}
Opher Etzion and Peter Niblett.
\newblock {\em Event Processing in Action}.
\newblock Manning Publications Company, 2010.

\bibitem{DBLP:conf/kr/Ghallab96}
Malik Ghallab.
\newblock On chronicles: Representation, on-line recognition and learning.
\newblock In {\em {KR}}, pages 597--606. Morgan Kaufmann, 1996.

\bibitem{DBLP:journals/vldb/GiatrakosAADG20}
Nikos Giatrakos, Elias Alevizos, Alexander Artikis, Antonios Deligiannakis, and
  Minos~N. Garofalakis.
\newblock Complex event recognition in the big data era: a survey.
\newblock {\em {VLDB} J.}, 29(1):313--352, 2020.

\bibitem{DBLP:conf/icdt/GrezRU19}
Alejandro Grez, Cristian Riveros, and Mart{\'{\i}}n Ugarte.
\newblock A formal framework for complex event processing.
\newblock In {\em {ICDT}}, volume 127 of {\em LIPIcs}, pages 5:1--5:18. Schloss
  Dagstuhl - Leibniz-Zentrum fuer Informatik, 2019.

\bibitem{DBLP:conf/icdt/GrezRUV20}
Alejandro Grez, Cristian Riveros, Mart{\'{\i}}n Ugarte, and Stijn Vansummeren.
\newblock On the expressiveness of languages for complex event recognition.
\newblock In {\em {ICDT}}, volume 155 of {\em LIPIcs}, pages 15:1--15:17.
  Schloss Dagstuhl - Leibniz-Zentrum f{\"{u}}r Informatik, 2020.

\bibitem{DBLP:journals/corr/Halle17}
Sylvain Hall{\'{e}}.
\newblock From complex event processing to simple event processing.
\newblock {\em CoRR}, abs/1702.08051, 2017.
\newblock URL: \url{http://arxiv.org/abs/1702.08051}.

\bibitem{hedman2004first}
Shawn Hedman.
\newblock {\em A First Course in Logic: An introduction to model theory, proof
  theory, computability, and complexity}.
\newblock Oxford University Press Oxford, 2004.

\bibitem{hedtstuck_complex_2017}
Ulrich Hedtst{\"u}ck.
\newblock {\em Complex event processing: {Verarbeitung} von {Ereignismustern}
  in {Datenstr{\"o}men}}.
\newblock Springer Vieweg, Berlin, 2017.

\bibitem{DBLP:books/daglib/0016921}
John~E. Hopcroft, Rajeev Motwani, and Jeffrey~D. Ullman.
\newblock {\em Introduction to automata theory, languages, and computation, 3rd
  Edition}.
\newblock Pearson international edition. Addison-Wesley, 2007.

\bibitem{DBLP:journals/tcs/KaminskiF94}
Michael Kaminski and Nissim Francez.
\newblock Finite-memory automata.
\newblock {\em Theor. Comput. Sci.}, 134(2):329--363, 1994.

\bibitem{DBLP:conf/sigmod/KorberGS21}
Michael K{\"{o}}rber, Nikolaus Glombiewski, and Bernhard Seeger.
\newblock Index-accelerated pattern matching in event stores.
\newblock In {\em {SIGMOD} Conference}, pages 1023--1036. {ACM}, 2021.

\bibitem{DBLP:journals/jcss/LibkinTV15}
Leonid Libkin, Tony Tan, and Domagoj Vrgoc.
\newblock Regular expressions for data words.
\newblock {\em J. Comput. Syst. Sci.}, 81(7):1278--1297, 2015.

\bibitem{DBLP:conf/lpar/LibkinV12}
Leonid Libkin and Domagoj Vrgoc.
\newblock Regular expressions for data words.
\newblock In {\em {LPAR}}, volume 7180 of {\em Lecture Notes in Computer
  Science}, pages 274--288. Springer, 2012.

\bibitem{DBLP:books/daglib/0017658}
David~C. Luckham.
\newblock {\em The power of events - an introduction to complex event
  processing in distributed enterprise systems}.
\newblock {ACM}, 2005.

\bibitem{DBLP:conf/kr/MantenoglouKA23}
Periklis Mantenoglou, Dimitrios Kelesis, and Alexander Artikis.
\newblock Complex event recognition with allen relations.
\newblock In {\em {KR}}, pages 502--511, 2023.

\bibitem{DBLP:conf/sigmod/MeiM09}
Yuan Mei and Samuel Madden.
\newblock Zstream: a cost-based query processor for adaptively detecting
  composite events.
\newblock In {\em {SIGMOD} Conference}, pages 193--206. {ACM}, 2009.

\bibitem{DBLP:journals/tocl/NevenSV04}
Frank Neven, Thomas Schwentick, and Victor Vianu.
\newblock Finite state machines for strings over infinite alphabets.
\newblock {\em {ACM} Trans. Comput. Log.}, 5(3):403--435, 2004.

\bibitem{DBLP:journals/dbsk/Petkovic22}
Dusan Petkovic.
\newblock Specification of row pattern recognition in the {SQL} standard and
  its implementations.
\newblock {\em Datenbank-Spektrum}, 22(2):163--174, 2022.

\bibitem{DBLP:conf/csl/Segoufin06}
Luc Segoufin.
\newblock Automata and logics for words and trees over an infinite alphabet.
\newblock In {\em {CSL}}, volume 4207 of {\em Lecture Notes in Computer
  Science}, pages 41--57. Springer, 2006.

\bibitem{DBLP:journals/jair/TsilionisAP22}
Efthimis Tsilionis, Alexander Artikis, and Georgios Paliouras.
\newblock Incremental event calculus for run-time reasoning.
\newblock {\em J. Artif. Intell. Res.}, 73:967--1023, 2022.

\bibitem{DBLP:journals/grammars/NoordG01}
Gertjan van Noord and Dale Gerdemann.
\newblock Finite state transducers with predicates and identities.
\newblock {\em Grammars}, 4(3):263--286, 2001.

\bibitem{DBLP:conf/wia/Veanes13}
Margus Veanes.
\newblock Applications of symbolic finite automata.
\newblock In {\em {CIAA}}, volume 7982 of {\em Lecture Notes in Computer
  Science}, pages 16--23. Springer, 2013.

\bibitem{DBLP:conf/lpar/VeanesBM10}
Margus Veanes, Nikolaj Bj{\o}rner, and Leonardo~Mendon{\c{c}}a de~Moura.
\newblock Symbolic automata constraint solving.
\newblock In {\em {LPAR} (Yogyakarta)}, volume 6397 of {\em Lecture Notes in
  Computer Science}, pages 640--654. Springer, 2010.

\bibitem{DBLP:conf/pods/WhiteRGD07}
Walker~M. White, Mirek Riedewald, Johannes Gehrke, and Alan~J. Demers.
\newblock What is "next" in event processing?
\newblock In {\em {PODS}}, pages 263--272. {ACM}, 2007.

\bibitem{DBLP:conf/sigmod/ZhangDI14}
Haopeng Zhang, Yanlei Diao, and Neil Immerman.
\newblock On complexity and optimization of expensive queries in complex event
  processing.
\newblock In {\em {SIGMOD} Conference}, pages 217--228. {ACM}, 2014.

\bibitem{DBLP:journals/pvldb/ZhuHC23}
Erkang Zhu, Silu Huang, and Surajit Chaudhuri.
\newblock High-performance row pattern recognition using joins.
\newblock {\em Proc. {VLDB} Endow.}, 16(5):1181--1194, 2023.

\end{thebibliography}
